\documentclass{amsproc}

\usepackage{amsrefs}
\usepackage{booktabs,amssymb}

\newtheorem{theorem}{Theorem}[section]
\newtheorem{proposition}[theorem]{Proposition}
\newtheorem{lemma}[theorem]{Lemma}
\newtheorem{corollary}[theorem]{Corollary}

\theoremstyle{definition}
\newtheorem{definition}[theorem]{Definition}

\theoremstyle{remark}
\newtheorem{remark}[theorem]{Remark}

\numberwithin{equation}{section}

\DeclareMathOperator{\Ad}{Ad}
\DeclareMathOperator{\Aut}{Aut}
\DeclareMathOperator{\Bimod}{Bimod}
\DeclareMathOperator{\bimod}{bimod}
\DeclareMathOperator{\comp}{comp}
\DeclareMathOperator{\diag}{diag}

\DeclareMathOperator{\Id}{id}
\DeclareMathOperator{\im}{im}

\DeclareMathOperator{\mult}{mult}
\DeclareMathOperator{\nullity}{null}
\DeclareMathOperator{\Orth}{O}
\DeclareMathOperator{\rank}{rank}

\DeclareMathOperator{\sgn}{sgn}
\DeclareMathOperator{\skel}{skel}
\DeclareMathOperator{\Sp}{Sp}
\DeclareMathOperator{\SU}{SU}
\DeclareMathOperator{\supp}{supp}
\DeclareMathOperator{\Sym}{Sym}
\DeclareMathOperator{\tr}{tr}
\DeclareMathOperator{\U}{U}

\newcommand{\Cstar}{{\ensuremath{C^*}}}
\newcommand{\even}{\textup{even}}
\newcommand{\latin}[1]{{\it{#1}\/}}
\newcommand{\odd}{\textup{odd}}

\newcommand{\sa}{\textup{sa}}
\newcommand{\Star}{\ensuremath{\ast}}
\newcommand{\term}[1]{{\it{#1}\/}}

\newcommand{\alg}[1]{\mathcal{#1}}
\newcommand{\bdd}{\mathcal{L}}
\newcommand{\crep}[1]{\overline{\bf {#1}}}
\newcommand{\eps}{\varepsilon}
\newcommand{\epsp}{\varepsilon^\prime}
\newcommand{\epspp}{\varepsilon^{\prime\prime}}
\newcommand{\field}[1]{{\mathbb{#1}}}
\newcommand{\form}[1]{\left\langle {#1} \right\rangle}
\newcommand{\hs}[1]{{\mathcal{#1}}}
\newcommand{\lbdd}{\bdd^{\textup{L}}}
\newcommand{\lrbdd}{\bdd^{\textup{LR}}}
\newcommand{\lrU}{\U^{\textup{LR}}}

\newcommand{\ms}[1]{\mathcal{#1}}
\newcommand{\norm}[1]{{\| {#1} \|}}
\newcommand{\rbdd}{\bdd^{\textup{R}}}
\newcommand{\rep}[1]{{\bf #1}}
\newcommand{\ring}[1]{{\mathbb{#1}}}
\newcommand{\semiring}[1]{{\mathbb{#1}}}
\newcommand{\fsl}{\mathfrak{sl}}
\newcommand{\fso}{\mathfrak{so}}

\newcommand{\spec}[1]{{\widehat{#1}}}
\newcommand{\ud}{\mathop{}\!\mathrm{d}}
\newcommand{\unit}[1]{{\tilde{#1}}}
\newcommand{\ZO}{\semiring{Z}_{\geq 0}}
\newcommand{\ie}{{\it i.e.\/}\ }

\newcommand{\cf}{{\it cf.\/}\ }

\begin{document}

\title[Finite Spectral Triples]{Moduli Spaces of Dirac Operators for Finite Spectral Triples}

%    author one information
\author{Branimir {\'C}a{\'c}i{\'c}}
\address{Max Planck Institute for Mathematics\\ Vivatsgasse 7\\ 53111 Bonn\\ Germany}
\curraddr{California Institute of Technology\\ Department of Mathematics\\ MC 253-37\\ Pasadena, CA 91125\\ U.S.A.}
\email{branimir@caltech.edu}


\begin{abstract}
The structure theory of finite real spectral triples developed by Krajewski and by Paschke and Sitarz is generalised to allow for arbitrary $KO$-dimension and the failure of orientability and Poincar{\'e} duality, and moduli spaces of Dirac operators for such spectral triples are defined and studied. This theory is then applied to recent work by Chamseddine and Connes towards deriving the finite spectral triple of the noncommutative-geometric Standard Model.
\end{abstract}

\maketitle

\section{Introduction}

From the time of Connes's 1995 paper~\cite{Connes95}, spectral triples with finite-di{\-}men{\-}sion{\-}al \Star-algebra and Hilbert space, or \term{finite spectral triples}, have been central to the noncommutative-geometric (NCG) approach to the Standard Model of elementary particle physics, where they are used to encode the fermionic physics. As a result, they have been the focus of considerable research activity.

The study of finite spectral triples began in earnest with papers by Paschke and Sitarz~\cite{PS98} and by Krajewski~\cite{Kraj98}, first released nearly simultaneously in late 1996 and early 1997, respectively, which gave detailed accounts of the structure of finite spin geometries, \ie of finite real spectral triples of $KO$-dimension $0 \bmod 8$ satisfying orientability and Poincar{\'e} duality. In their approach, the study of finite spectral triples is reduced, for the most part, to the study of \term{multiplicity matrices}, integer-valued matrices that explicitly encode the underlying representation-theoretic structure. Krajewski, in particular, defined what are now called \term{Krajewski diagrams} to facilitate the classification of such spectral triples. Iochum, Jureit, Sch{\"u}cker, and Stephan have since undertaken a programme of classifying Krajewski diagrams for finite spectral triples satisfying certain additional physically desirable  assumptions~\cites{ACG1,ACG2,ACG3,Sch05} using combinatorial computations~\cite{JS08}, with the aim of fixing the finite spectral triple of the Standard Model amongst all other such triples.

However, there were certain issues with the then-current version of the NCG Standard Model, including difficulty with accomodating massive neutrinos and the so-called fermion doubling problem, that were only to be resolved in the 2006 papers by Connes~\cite{Connes06} and by Chamseddine, Connes and Marcolli~\cite{CCM07}, which use the Euclidean signature of earlier papers, and by Barrett~\cite{Bar07}, which instead uses Lorentzian signature; we restrict our attention to the Euclidean signature approach of~\cite{Connes06} and~\cite{CCM07}, which has more recently been set forth in the monograph~\cite{CM08} of Connes and Marcolli. The finite spectral triple of the current version has $KO$-dimension $6 \bmod 8$ instead of $0 \bmod 8$, fails to be orientable, and only satisfies a certain modified version of Poincar{\'e} duality. It also no longer satisfies $S^0$-reality, another condition that holds for the earlier finite geometry of~\cite{Connes95}, though only because of the Dirac operator. Jureit, and Stephan~\cites{ACG4,ACG5} have since adopted the new value for the $KO$-dimension, but further assume orientability and Poincar{\'e} duality. As well, Stephan~\cite{St06} has proposed an alternative finite spectral triple for the current NCG Standard Model with the same physical content but satisfying Poincar{\'e} duality; it also just fails to be $S^0$-real in the same manner as the finite geometry of~\cite{CCM07}; in the same paper, Stephan also discusses non-orientable finite spectral triples.

More recently, Chamseddine and Connes~\cites{CC08a,CC08b} have sought a purely algebraic method of isolating the finite spectral triple of the NCG Standard Model, by which they have obtained the correct \Star-algebra, Hilbert space, grading and real structure using a small number of fairly elementary assumptions. In light of these successes, it would seem reasonable to try to view this new approach of Chamseddine and Connes through the lens of the structure theory of Krajewski and Paschke--Sitarz, at least in order to understand better their method and the assumptions involved. This, however, would require adapting that structure theory to handle the failure of orientability and Poincar{\'e} duality, yielding the initial motivation of this work.

To that end, we provide, for the first time, a comprehensive account of the structure theory of Krajewski and Paschke--Sitarz for finite real spectral triples of arbitrary $KO$-dimension, without the assumptions of orientability or Poincar{\'e} duality; this consists primarily of straightforward generalisations of the results and techniques of~\cite{PS98} and~\cite{Kraj98}. In this light, the main features of the approach presented here are the following:
\begin{enumerate}
 \item A finite real spectral triple with algebra $\alg{A}$ is to be viewed as an $\alg{A}$-bimodule with some additional structure, together with a choice of Dirac operator compatible with that structure.
 \item For fixed algebra $\alg{A}$, an $\alg{A}$-bimodule is entirely characterised by its multiplicity matrix (in the ungraded case) or matrices (in the graded case), which also completely determine(s) what sort of additional structure the bimodule can admit; this additional structure is then unique up to unitary equivalence.
 \item The form of suitable Dirac operators for an $\alg{A}$-bimodule with real structure is likewise determined completely by the multiplicity matrix or matrices of the bimodule and the choice of additional structure.
\end{enumerate}
However, we do not discuss Krajewski diagrams, though suitable generalisation thereof should follow readily from the generalised structure theory for Dirac operators.

Once we view a real spectral triple as a certain type of bimodule together with a \emph{choice} of suitable Dirac operator, it then becomes natural to consider moduli spaces of suitable Dirac operators, up to unitary equivalence, for a bimodule with fixed additional structure, yielding finite real spectral triples of the appropriate $KO$-dimension. The construction and study of such moduli spaces of Dirac operators first appear in~\cite{CCM07}, though the focus there is on the sub-moduli space of Dirac operators commuting with a certain fixed subalgebra of the relevant \Star-algebra. Our last point above almost immediately leads us to relatively concrete expressions for general moduli spaces of Dirac operators, which also appear here for the first time. Multiplicity matrices and moduli spaces of Dirac operators are then worked out for the bimodules appearing in the Chamseddine--Connes--Marcolli formulation of the NCG Standard Model~\cites{CCM07,CM08} as examples.

Finally, we apply these methods to the work of Chamseddine and Connes~\cites{CC08a,CC08b}, offering concrete proofs and some generalisations of their results. In particular, the choices determining the finite geometry of the current NCG Standard Model within their framework are made explicit.

This work, a revision of the author's qualifying year project (master's thesis equivalent) at the Bonn International Graduate School in Mathematics (BIGS) at the University of Bonn, is intended as a first step towards a larger project of investigating in generality the underlying noncommutative-geometric formalism for field theories found in the NCG Standard Model, with the aim of both better understanding current versions of the NCG Standard Model and facilitating the further development of the formalism itself.

The author would like to thank his supervisor, Matilde Marcolli, for her extensive comments and for her advice, support, and patience, Tobias Fritz for useful comments and corrections, and George Elliott for helpful conversations. The author also gratefully acknowledges the financial and administrative support of BIGS and of the Max Planck Institute for Mathematics, as well as the hospitality and support of the Department of Mathematics at the California Institute of Technology and of the Fields Institute.

\section{Preliminaries and Definitions}

\subsection{Real \Cstar-algebras}

In light of their relative unfamiliarity compared to their complex counterparts, we begin with some basic facts concerning real \Cstar-algebras.

First, recall that a \term{real \Star-algebra} is a real associative algebra $\alg{A}$ together with an \term{involution} on $\alg{A}$, namely an antihomomorphism $\ast$ satisfying $\ast^2 = \Id$, and that the \term{unitalisation} of a real \Star-algebra $\alg{A}$ is the unital real \Star-algebra $\tilde{\alg{A}}$ defined to be $\alg{A} \oplus \field{R}$ as a real vector space, together with the multiplication $(a,\alpha)(b,\beta) := (ab + \alpha b + \beta a, \alpha\beta)$ for $a$, $b \in \alg{A}$, $\alpha$, $\beta \in \field{R}$ and the involution $\star \oplus \Id_{\field{R}}$. Note that if $\alg{A}$ is already unital, then $\unit{\alg{A}}$ is simply $\alg{A}\oplus\field{R}$.

\begin{definition}
 A \term{real \Cstar-algebra} is a real \Star-algebra $\alg{A}$ endowed with a norm $\norm{\cdot}$ making $\alg{A}$ a real Banach algebra, such that the following two conditions hold:
\begin{enumerate}
 \item $\forall a \in \alg{A}$, $\norm{a^* a} = \norm{a}^2$ (\term{\Cstar-identity});
 \item $\forall a \in \unit{\alg{A}}$, $1 + a^* a$ is invertible in $\unit{\alg{A}}$ (\term{symmetry}).
\end{enumerate}
\end{definition}

The symmetry condition is redundant for complex \Cstar-algebras, but not for real \Cstar-algebras. Indeed, consider $\field{C}$ as a real algebra together with the trivial involution $\Star = \Id$ and the usual norm $\norm{\zeta} = |\zeta|$, $\zeta \in \field{C}$. Then $\field{C}$ with this choice of involution and norm yields a real Banach \Star-algebra satisfying the \Cstar-identity but not symmetry, for $1 + i^* i = 0$ is certainly not invertible in $\unit{\field{C}} = \field{C}\oplus\field{R}$.

Now, in the finite-dimensional case, one can give a complete description of real \Cstar-algebras, which we shall use extensively in what follows:

\begin{theorem}[Wedderburn's theorem for real \Cstar-algebras \cite{Fare}]
 Let $\alg{A}$ be a finite-dimensional real \Cstar-algebra. Then
\begin{equation}
 \alg{A} \cong \bigoplus_{i=1}^N M_{n_i}(\field{K}_i),
\end{equation}
where $\field{K}_i = \field{R}$, $\field{C}$, or $\field{H}$, and $n_i \in \semiring{N}$. Moreover, this decomposition is unique up to permutation of the direct summands.
\end{theorem}

Note, in particular, that a finite-dimensional real \Cstar-algebra is necessarily unital.

Given a finite-dimensional real \Cstar-algebra $\alg{A}$ with fixed \term{Wedderburn decomposition} $\oplus_{i=1}^N M_{n_i}(\field{K}_i)$ we can associate to $\alg{A}$ a finite dimensional complex \Cstar-algebra $\alg{A}_{\field{C}}$, the \term{complex form} of $\alg{A}$, by setting
\begin{equation}
 \alg{A}_{\field{C}} := \bigoplus_{i=1}^N M_{m_i}(\field{C}),
\end{equation}
where $m_i = 2n_i$ if $\field{K}_i = \field{H}$, and $m_i = n_i$ otherwise. Then $\alg{A}$ can be viewed as a real \Star-subalgebra of $\alg{A}_\field{C}$ such that $\alg{A}_{\field{C}} = \alg{A} + i \alg{A}$, that is, as a \term{real form} of $\alg{A}_{\field{C}}$. Here, $\field{H}$ is considered as embedded in $M_2(\field{C})$ by
\[
 \zeta_1 + j \zeta_2 \mapsto
 \begin{pmatrix}
  \zeta_1 & \zeta_2 \\
  -\overline{\zeta_2} & \overline{\zeta_1}
 \end{pmatrix},
\]
for $\zeta_1$, $\zeta_2 \in \field{C}$.

In what follows, we will consider only finite-dimensional real \Cstar-algebras with fixed Wedderburn decomposition.

\subsection{Representation theory}

In keeping with the conventions of noncommutative differential geometry, we shall consider \Star-representations of real \Cstar-algebras on complex Hilbert spaces. Recall that such a (left) representation of a real \Cstar-algebra $\alg{A}$ consists of a complex Hilbert space $\hs{H}$ together with a \Star-homomorphism $\lambda : \alg{A} \to \bdd(\hs{H})$ between real \Cstar-algebras. Similarly, a \term{right representation} of $\alg{A}$ is defined to be a complex Hilbert space $\hs{H}$ together with a \Star-\emph{antihomomorphism} $\rho : \alg{A} \to \bdd(\hs{H})$ between real \Cstar-algebras. For our purposes, then, an \term{$\alg{A}$-bimodule} consists of a complex Hilbert space $\hs{H}$ together with a left \Star-representation $\lambda$ and a right \Star-representation $\rho$ that commute, \ie such that $[\lambda(a),\rho(b)]=0$ for all $a$, $b \in \alg{A}$. In what follows, we will consider only finite-dimensional representations and hence only finite-dimensional bimodules; since finite-dimensional \Cstar-algebras are always unital, we shall require all representations to be unital as well.

Now, given a left [right] representation $\alpha = (\hs{H},\pi)$ of an algebra $\alg{A}$, one can define its \term{transpose} to be the right [left] representation $\alpha^T = (\hs{H}^*,\pi^T)$ , where $\pi^T(a) := \pi(a)^T$ for all $a \in \alg{A}$. Note that for any left or right representation $\alpha$, $(\alpha^T)^T$ can naturally be identified with $\alpha$ itself. In the case that $\hs{H} = \field{C}^N$, we shall identify $\hs{H}^*$ with $\hs{H}$ by identifying the standard ordered basis on $\hs{H}$ with the corresponding dual basis on $\hs{H}^*$. The notion of the transpose of a representation allows us to reduce discussion of right representations to that of left representations.

Since real \Cstar-algebras are semisimple, any left representation can be written as a direct sum of irreducible representations, unique up to permutation of the direct summands, and hence any right representation can be written as a direct sum of transposes of irreducible representations, again unique up to permutation of the direct summands.

\begin{definition}
 The \term{spectrum} $\spec{\alg{A}}$ of a real \Cstar-algebra $\alg{A}$ is the set of unitary equivalence classes of irreducible representations of $\alg{A}$.
\end{definition}

Now, let $\alg{A}$ be a real \Cstar-algebra with Wedderburn decomposition $\oplus_{i=1}^N M_{k_i}(\field{K}_i)$. Then
\begin{equation}
 \spec{\alg{A}} = \bigsqcup_{i=1}^N \spec{M_{k_i}(\field{K}_i)},
\end{equation}
where the embedding of $\spec{M_{k_i}(\field{K}_i)}$ in $\spec{\alg{A}}$ is given by composing the representation maps with the projection of $\alg{A}$ onto the direct summand $M_{k_i}(\field{K}_i)$. The building blocks for $\spec{\alg{A}}$ are as follows:
\begin{enumerate}
 \item $\spec{M_{n}(\field{R})} = \{[(\field{C}^n,\lambda)]\}$,
 \item $\spec{M_{n}(\field{C})} = \{[(\field{C}^n,\lambda)],[(\field{C}^N,\overline{\lambda})]\}$,
 \item $\spec{M_{n}(\field{H})} = \{[(\field{C}^{2n},\lambda)]\}$,
\end{enumerate}
where $\lambda(a)$ denotes left multiplication by $a$ and $\overline{\lambda}(a)$ denotes left multiplication by $\overline{a}$.

\begin{definition}
 Let $\alg{A}$ be a real \Cstar-algebra, and let $\alpha \in \spec{\alg{A}}$. We shall call $\alpha$ \term{conjugate-linear} if it arises from the conjugate-linear irreducible representation $(a \mapsto \overline{a},\field{C}^{n_i})$ of a direct summand of $\alg{A}$ of the form $M_{n_i}(\field{C})$; otherwise we shall call it \term{complex-linear}.
\end{definition}

Thus, a representation $\alpha$ of the real \Cstar-algebra $\alg{A}$ extends to a $\field{C}$-linear \Star-representation of $\alg{A}_{\field{C}}$ if and only if $\alpha$ is the sum of complex-linear irreducible representations of $\alg{A}$.

Finally, for an individual direct summand $M_{k_i}(\field{K}_i)$ of $\alg{A}$, let $e_i$ denote its unit, $n_i$ the dimension of its irreducible representations (which is therefore equal to $2 k_i$ if $\field{K}_i = \field{H}$, and to $k_i$ itself otherwise), $\rep{n}_i$ its complex-linear irreducible representation, and, if $\field{K}_i = \field{C}$, $\crep{n}_i$ its conjugate-linear irreducible representation. We define a strict ordering $<$ on $\spec{\alg{A}}$ by setting $\alpha < \beta$ whenever $\alpha \in \spec{M_{n_i}(\field{K}_i)}$, $\beta \in \spec{M_{n_j}(\field{K}_j)}$ for $i < j$, and by setting $\rep{n}_i < \crep{n}_i$ in the case that $\field{K}_i = \field{C}$. Note that the ordering depends on the choice of Wedderburn decomposition, \ie on the choice of ordering of the direct summands. Let $S$ denote the cardinality of $\spec{\alg{A}}$. We shall identify $M_S(\field{R})$ with the real algebra of functions $\spec{\alg{A}}^2 \to \ring{R}$, and hence index the standard basis $\{E_{\alpha\beta}\}$ of $M_S(\field{R})$ by $\spec{\alg{A}}^2$.

\subsection{Bimodules and spectral triples}

Let us now turn to spectral triples. Recall that we are considering only finite-dimensional algebras and representations (\ie Hilbert spaces), so that we are dealing only with what are termed \term{finite} or \term{discrete} spectral triples.

Let $\hs{H}$ and $\hs{H}^\prime$ be $\alg{A}$-bimodules. We shall denote by $\lbdd_\alg{A}(\hs{H},\hs{H}^\prime)$, $\rbdd_\alg{A}(\hs{H},\hs{H}^\prime)$, and $\lrbdd_\alg{A}(\hs{H},\hs{H}^\prime)$ the subspaces of $\bdd(\hs{H},\hs{H}^\prime)$ consisting of left $\alg{A}$-linear, right $\alg{A}$-linear, and left and right $\alg{A}$-linear operators, respectively. In the case that $\hs{H}^\prime = \hs{H}$, we shall write simply $\lbdd_\alg{A}(\hs{H})$, $\rbdd_\alg{A}(\hs{H})$ and $\lrbdd_\alg{A}(\hs{H})$. If $N$ is a subalgebra or linear subspace of a real or complex \Cstar-algebra, we shall denote by $N_\sa$ the real linear subspace of $N$ consisting of the self-adjoint elements of $N$, and we shall denote by $\U(N)$ set of unitary elements of $N$. Finally, for operators $A$ and $B$ on a Hilbert space, we shall denote their anticommutator $AB + BA$ by $\{ A,B \}$.

\subsubsection{Conventional definitions}

We begin by recalling the standard definitions for spectral triples of various forms. Since we are working with the finite case, all analytical requirements become redundant, leaving behind only the algebraic aspects of the definitions.

The following definition first appeared in a 1995 paper~\cite{Connes95a} by Connes:

\begin{definition}
 A \term{spectral triple} is a triple $(\alg{A},\hs{H},D)$, where:
\begin{itemize}
 \item $\alg{A}$ is a unital real or complex \Star-algebra;
 \item $\hs{H}$ is a complex Hilbert space on which $\alg{A}$ has a left representation  $\lambda : \alg{A} \to \bdd(\hs{H})$;
 \item $D$, the \term{Dirac operator}, is a self-adjoint operator on $\hs{H}$.
\end{itemize}

Moreover, if there exists  a $\ring{Z}/2\ring{Z}$-grading $\gamma$ on $\hs{H}$ (\ie a self-adjoint unitary on $\hs{H}$) such that:
\begin{enumerate}
 \item $[\gamma, \lambda(a)]=0$ for all $a \in \alg{A}$,
 \item $\{\gamma,D\}=0$;
\end{enumerate}
then the spectral triple is said to be \term{even}. Otherwise, it is said to be \term{odd}.
\end{definition}

In the context of the general definition for spectral triples, a finite spectral triple necessarily has metric dimension $0$.

In a slightly later paper~\cite{Connes95}, Connes defines the additional structure on spectral triples necessary for defining the noncommutative spacetime of the NCG Standard Model; indeed, the same paper also contains the first version of the NCG Standard Model to use the language of spectral triples, in the form of a reformulation of the so-called Connes-Lott model.

\begin{definition}\label{realdef}
 A spectral triple $(\alg{A},\hs{H},D)$ is called a \term{real spectral triple of $KO$-dimension $n \bmod 8$} if, in the case of $n$ even, it is an even spectral triple, and if there exists an antiunitary $J: \hs{H} \to \hs{H}$ such that:
\begin{enumerate}
 \item $J$ satisfies $J^2 = \eps$, $JD = \epsp DJ$ and $J\gamma = \epspp \gamma J$ (in the case of even $n$), where $\eps$, $\epsp$, $\epspp \in \{-1,1\}$ depend on $n \bmod 8$ as follows:
 \begin{center}
  \begin{tabular}{crrrrrrrr}
   \toprule
   $n$ & $0$ & $1$ & $2$ & $3$ & $4$ & $5$ & $6$ & $7$ \\
   \midrule
   $\eps$ & $1$ & $1$ & $-1$ & $-1$ & $-1$ & $-1$ & $1$ & $1$ \\
   $\epsp$ & $1$ & $-1$ & $1$ & $1$ & $1$ & $-1$ & $1$ & $1$ \\
   $\epspp$ & $1$ & & $-1$ & & $1$ & & $-1$ & \\
   \bottomrule
  \end{tabular}
 \end{center}
 \item The \term{order zero condition} is satisfied, namely $[\lambda(a),J\lambda(b)J^*]=0$ for all $a$, $b \in \alg{A}$;
 \item The \term{order one condition} is satisfied, namely $[[D,\lambda(a)],J\lambda(b)J^*]=0$ for all $a$, $b \in \alg{A}$.
\end{enumerate}

Moreover, if there exists a self-adjoint unitary $\epsilon$ on $\hs{H}$ such that:
\begin{enumerate}
 \item $[\epsilon,\lambda(a)]=0$ for all $a \in \alg{A}$;
 \item $[\epsilon, D]=0$;
 \item $\{\epsilon, J\}=0$;
 \item $[\epsilon, \gamma]=0$ (even case);
\end{enumerate}
then the real spectral triple is said to be \term{$S^0$-real}.
\end{definition}

\begin{remark}[Krajewski~\cite{Kraj98}*{\S 2.2}, Paschke--Sitarz~\cite{PS98}*{Obs.~1}]
 If $(\alg{A},\hs{H},D)$ is a real spectral triple, then the order zero condition is equivalent to the statement that $\hs{H}$ is an $\alg{A}$-bimodule for the usual left action $\lambda$ and the right action $\rho: a \mapsto J\lambda(a^*)J^*$.
\end{remark}

It was commonly assumed until fairly recently that the finite geometry of the NCG Standard Model should be $S^0$-real. Though the current version of the NCG Standard Model no longer makes such an assumption~\cites{Connes06,CCM07}, we shall later see that its finite geometry can still be seen as satisfying a weaker version of $S^0$-reality.

\subsubsection{Structures on bimodules}

In light of the above remark, the order one condition, the strongest algebraic condition placed on Dirac operators for real spectral triples, should be viewed more generally as a condition applicable to operators on bimodules~\cite{Kraj98}*{\S 2.4}. This then motivates our point of view that a finite real spectral triple $(\alg{A},\hs{H},D)$ should be viewed rather as an $\alg{A}$-bimodule with additional structure, together with a Dirac operator satisfying the order one condition that is compatible with that additional structure. We therefore begin by defining a suitable notion of ``additional structure'' for bimodules.

\begin{definition}
 A \term{bimodule structure} $P$ consists of the following data:
\begin{itemize}
 \item A set $\alg{P} = \alg{P}_\gamma \sqcup \alg{P}_J \sqcup \alg{P}_\epsilon$, where each set $\alg{P}_X$ is either empty or the singleton $\{X\}$, and where $\alg{P}_\epsilon$ is non-empty only if $\alg{P}_J$ is non-empty;
 \item If $\alg{P}_J$ is non-empty, a choice of \term{$KO$-dimension} $n \bmod 8$, where $n$ is even if and only if $\alg{P}_\gamma$ is non-empty.
\end{itemize}

In particular, we call a structure $P$:
\begin{itemize}
 \item \term{odd} if $\alg{P}$ is empty;
 \item \term{even} if $\alg{P} = \alg{P}_\gamma = \{\gamma\}$;
 \item \term{real} if $\alg{P}_J$ is non-empty and $\alg{P}_\epsilon$ is empty
 \item \term{$S^0$-real} if $\alg{P}_\epsilon$ is non-empty.
\end{itemize}

Finally, if $P$ is a graded structure, we call $\gamma$ the \term{grading}, and if $P$ is real or $S^0$-real, we call $J$ the \term{charge conjugation}.
\end{definition}

Since this notion of $KO$-dimension is meant to correspond with the usual $KO$-dimension of a real spectral triple, we assign to each real or $S^0$-real structure $P$ of $KO$-dimension $n \bmod 8$ constants $\eps$, $\epsp$ and, in the case of even $n$, $\epspp$, according to the table in Definition~\ref{realdef}.

We now define the \term{structure algebra} of a structure $P$ to be the real associative algebra with generators $\alg{P}$ and relations, as applicable,
\[
 \gamma^2 = 1, \quad J^2 = \eps, \quad \epsilon^2 = 1; \gamma J = \epspp J \gamma, \quad [\gamma,\epsilon] = 0, \quad \{\epsilon,J\} = 0.
\]

\begin{definition}
 An $\alg{A}$-bimodule $\hs{H}$ is said to have structure $P$ whenever it admits a faithful representation of the structure algebra of $P$ such that, when applicable, $\gamma$ and $\epsilon$ are represented by self-adjoint unitaries in $\lrbdd_{\alg{A}}(\hs{H})$, and $J$ is represented by an antiunitary on $\hs{H}$ such that
\begin{equation}\label{realintertwine}
 \forall a \in \alg{A}, \quad \rho(a) = J \lambda(a^*) J.
\end{equation}
\end{definition}

Note that a $S^0$-real bimodule can always be considered as a real bimodule, and a real bimodule of even [odd] $KO$-dimension can always be considered as an even [odd] bimodule. Note also that an even bimodule is simply a graded bimodule such that the algebra acts from both left and right by degree $0$ operators, and the grading itself respects the Hilbert space structure; an odd bimodule is then simply an ungraded bimodule. We use the terms ``even'' and ``odd'' so as to keep the terminology consistent with that for spectral triples.

Note also that for a real or $S^0$-real structure $P$, the structure algebra of $P$ is independent of the value of $\epsp$. Thus the notions of real [$S^0$-real] $\alg{A}$-bimodule with $KO$-dimension $1 \bmod 8$ and $7 \bmod 8$ are identical, as are the notions of [$S^0$-real] $\alg{A}$-bimodule with $KO$-dimension $3 \bmod 8$ and $5 \bmod 8$; again, we make the distinction with an eye to the discussion of Dirac operators (and hence of spectral triples) later on.

Now, a \term{unitary equivalence} of $\alg{A}$-bimodules $\hs{H}$ and $\hs{H}^\prime$ with structure $P$ is a unitary equivalence of $\alg{A}$-bimodules (\ie a unitary element of $\lrbdd_\alg{A}(\hs{H},\hs{H}^\prime)$) that effects unitary equivalence of the representations of the structure algebra of $P$. We denote the set of all such unitary equivalences $\hs{H} \to \hs{H}^\prime$ by $\lrU_\alg{A}(\hs{H},\hs{H}^\prime;\alg{P})$. In particular, $\lrU_\alg{A}(\hs{H},\hs{H};\alg{P})$, which we denote by $\lrU_\alg{A}(\hs{H};\alg{P})$, is a subgroup of $\lrU_\alg{A}(\hs{H}) := \U(\lrbdd_\alg{A}(\hs{H}))$. In all such notation, we suppress the argument $\alg{P}$ whenever $\alg{P}$ is empty.

\begin{definition}
 Let $\alg{A}$ be a real \Cstar-algebra, and let $P$ be a bimodule structure. The abelian monoid $(\Bimod(\alg{A},P),+)$ of $\alg{A}$-bimodules with structure $P$ is defined as follows:
\begin{itemize}
 \item $\Bimod(\alg{A},P)$ is the set of unitary equivalence classes of $\alg{A}$-bimodules with structure $P$;
 \item For $[\hs{H}]$, $[\hs{H}^\prime] \in \Bimod(\alg{A},P)$, $[\hs{H}] + [\hs{H}^\prime] := [\hs{H} \oplus \hs{H}^\prime]$.
\end{itemize}
\end{definition}

For convenience, we shall denote $\Bimod(\alg{A},P)$ by:
\begin{itemize}
 \item $\Bimod(\alg{A})$ if $P$ is the odd structure;
 \item $\Bimod^\even(\alg{A})$ if $P$ is the even structure;
 \item $\Bimod(\alg{A},n)$ if $P$ is the real structure of $KO$-dimension $n \bmod 8$;
 \item $\Bimod^0(\alg{A},n)$ if $P$ is the $S^0$-real structure of $KO$-dimension $n \bmod 8$.
\end{itemize}
These monoids will be studied in depth in the next section. In light of our earlier comment, we therefore have that
\[
 \Bimod(\alg{A},1) = \Bimod(\alg{A},7), \quad \Bimod(\alg{A},3) = \Bimod(\alg{A},5).
\]
and
\[
 \Bimod^0(\alg{A},1) = \Bimod^0(\alg{A},7), \quad \Bimod^0(\alg{A},3) = \Bimod^0(\alg{A},5).
\]

Finally, for the sake of completeness, we now define the notions of orientabilty and Poincar{\'e} duality in this more general context; in the case of a real spectral triple $(\alg{A},\hs{H},D,\gamma, J)$ of even $KO$-dimension, where the right action is given by $\rho(a) := J\lambda(a^*)J^*$, these definitions yield precisely the usual ones (\cf~\cite{Kraj98}*{\S\S 2.2, 2.3}).

\begin{definition}
 We call an even $\alg{A}$-bimodule $(\hs{H},\gamma)$ \term{orientable} if there exist $a_1, \dotsc, a_k$, $b_1, \dotsc, b_k \in \alg{A}$ such that
\begin{equation}
 \gamma = \sum_{i=1}^k \lambda(a_i)\rho(b_i).
\end{equation}
\end{definition}

\begin{definition}
 Let $\alg{A}$ be a real $\Cstar$-algebra, and let $(\hs{H},\gamma)$ be an even $\alg{A}$-bimodule. Then the \term{intersection form} $\left\langle \cdot, \cdot \right\rangle : KO_0(\alg{A}) \times KO_0(\alg{A}) \to \ring{Z}$ associated with $(\hs{H},\gamma)$ is defined by setting
\begin{equation}
 \left\langle \left[e\right], \left[f\right] \right\rangle := \tr(\gamma\lambda(e)\rho(f))
\end{equation}
for projections $e$, $f \in \alg{A}$.

In the case that the intersection form is non-degenerate, we shall say that $(\hs{H},\gamma)$ satisfies \term{Poincar{\'e} duality}.
\end{definition}

The orientability assumption was used extensively in \cite{PS98} and \cite{Kraj98}, as it leads to considerable algebraic simplifactions; we shall later define a weakened version of orientability that will yield precisely those simplifications.

\subsubsection{Bilateral spectral triples}

We now turn to Dirac operators on bimodules satisfying a generalised order one condition, and define the appropriate notion of compatibility with additional structure on the bimodule.

\begin{definition}
 A \term{Dirac operator} for an $\alg{A}$-bimodule $\hs{H}$ with structure $P$ is a self-adjoint operator $D$ on $\hs{H}$ satisfying the \term{order one condition}:
\begin{equation}
 \forall a, b \in \alg{A}, \quad [[D,\lambda(a)],\rho(b)] = 0,
\end{equation}
together with the following relations, as applicable:
\[
 \{D,\gamma\} = 0, \quad DJ = \epsp JD, \quad [D,\epsilon]=0.
\]
\end{definition}

We denote the finite-dimensional real vector space of Dirac operators for an an $\alg{A}$-bimodule $\hs{H}$ with structure $P$ by $\ms{D}_0(\alg{A},\hs{H},\alg{P})$.

\begin{definition}
 A \term{bilateral spectral triple} with structure $P$ is a triple of the form $(\alg{A},\hs{H},D)$, where $\alg{A}$ is a real \Cstar-algebra, $\hs{H}$ is an $\alg{A}$-bimodule with structure $P$, and $D$ is a Dirac operator for $(\hs{H},P)$.
\end{definition}

We shall generally denote such a spectral triple by $(\alg{A},\hs{H},D;\alg{P})$, where $\alg{P}$ is the set of generators of the structure algebra; in cases where the presence or absence of a grading $\gamma$ is immaterial, we will suppress the generator $\gamma$ in this notation.

\begin{remark}
 In the case that $P$ is a real [$S^0$-real] structure of $KO$-dimension $n \bmod 8$, a bilateral spectral triple with structure $P$ is precisely a real [$S^0$-real] spectral triple of $KO$-dimension $n \bmod 8$.

More generally, an odd [even] bilateral spectral triple $(\alg{A},\hs{H},D)$ is equivalent to an odd [even] spectral triple $(\alg{A} \otimes \alg{A}^\textnormal{op},\hs{H},D)$ such that $[[D,\alg{A}\otimes 1],1\otimes\alg{A}^\textnormal{op}]=\{0\}$, an object that first appears in connection with $S^0$-real spectral triples~\cite{Connes95}
\end{remark}

A \term{unitary equivalence} of spectral triples $(\alg{A},\hs{H},D)$ and $(\alg{A},\hs{H}^\prime, D^\prime)$ is then a unitary $U \in \lrU_\alg{A}(\hs{H},\hs{H}^\prime)$ such that $D^\prime = U D U^*$. This concept leads us to the following definition:

\begin{definition}
 Let $\alg{A}$ be a real \Cstar-algebra, and let $\hs{H}$ be an $\alg{A}$-bimodule with structure $P$. The \term{moduli space of Dirac operators} for $\hs{H}$ is defined by
\begin{equation}
 \ms{D}(\alg{A},\hs{H},\alg{P}) := \ms{D}_0(\alg{A},\hs{H},\alg{P})/\lrU_\alg{A}(\hs{H},\alg{P}),
\end{equation}
where $\lrU_\alg{A}(\hs{H},\alg{P})$ acts on $\ms{D}_0(\alg{A},\hs{H},\alg{P})$ by conjugation.
\end{definition}

If $\alg{C}$ is a central subalgebra of $\alg{A}$, we can form the subspace
\begin{equation}
 \ms{D}_0(\alg{A},\hs{H},\alg{P};\alg{C}) := \{D \in \ms{D}_0(\alg{A},\hs{H},\alg{P}) \mid [D,\lambda(\alg{C})] = [D,\rho(\alg{C})] = \{0\}\}.
\end{equation}
and hence the sub-moduli space
\begin{equation}
 \ms{D}(\alg{A},\hs{H},\alg{P};\alg{C}) := \ms{D}_0(\alg{A},\hs{H},\alg{P};\alg{C})/\lrU_\alg{A}(\hs{H},\alg{P}),
\end{equation}
of $\ms{D}_0(\alg{A},\hs{H},\alg{P})$; the moduli space of Dirac operators studied by Chamseddine, Connes and Marcolli~\cite{CCM07}*{\S 2.7},\cite{CM08}*{\S 13.4} is in fact a sub-moduli space of this form.

Since $\ms{D}(\alg{A},\hs{H},\alg{P})$ [$\ms{D}(\alg{A},\hs{H},\alg{P};\alg{C})$] is the orbit space of a smooth finite-di\-men\-sion\-al representation of a compact Lie group, it is \latin{a priori} locally compact Hausdorff, and is thus homeomorphic to a semialgebraic subset of $\field{R}^d$ for some $d$~\cite{Schwarz75}. The dimension of $\ms{D}(\alg{A},\hs{H},\alg{P})$ [$\ms{D}(\alg{A},\hs{H},\alg{P};\alg{C})$] can then be defined as the dimension of this semialgebraic set. Such moduli spaces will be discussed in some detail.

\subsubsection{$S^0$-reality}

Following Connes~\cite{Connes95}, we now describe how to reduce the study of $S^0$-real bimodules of even [odd] $KO$-dimension to the study of even [odd] bimodules.

Let $(\hs{H},J,\epsilon)$ be an $S^0$-real $\alg{A}$-bimodule of even [odd] $KO$-dimension. Define mutually orthogonal projections $P_i$, $P_{-i}$ in $\lrbdd_\alg{A}(\hs{H})$ by $P_{\pm i} = \frac{1}{2}(1 \pm \epsilon)$. Then, at the level of even [odd] bimodules, $\hs{H} = \hs{H}_i \oplus \hs{H}_{-i}$ for $\hs{H}_{\pm i} := P_{\pm i} \hs{H}$, where the left and right actions on $\hs{H}_{\pm i}$ are given by
\[
 \lambda_{\pm i}(a) := P_{\pm i}\lambda(a)P_{\pm i}, \quad \rho_{\pm i}(a) := P_{\pm i}\rho(a)P_{\pm i},
\]
for $a \in \alg{A}$, and, in the case of even $KO$-dimension, the grading on $\hs{H}_{\pm i}$ is given by $\gamma_{\pm i} := P_{\pm i} \gamma P_{\pm i}$. Moreoever,
\[
 J =
\begin{pmatrix}
 0 & \eps \tilde{J}^*\\
 \tilde{J} & 0
\end{pmatrix},
\]
where $\tilde{J} := P_{-i} J P_i$ is an antiunitary $\hs{H}_i \to \hs{H}_{-i}$, so that for $a \in \alg{A}$,
\[
 \lambda_{-i}(a) = \tilde{J} \rho_i(a^*) \tilde{J}^*, \quad \rho_{-i}(a) = \tilde{J} \lambda_i(a^*) \tilde{J}^*,
\]
and in the case of even $KO$-dimension, $\gamma_{-i} = \epspp \tilde{J} \gamma \tilde{J}^*$. Finally, note that $\tilde{J}$ can also be viewed as a unitary $\overline{\hs{H}_i} \to \hs{H}_{-i}$, where $\overline{\hs{H}_i}$ denotes the conjugate space of $\hs{H}$. Hence, for fixed $KO$-dimension, an $S^0$-real $\alg{A}$-bimodule $\hs{H}$ is determined, up to unitary equivalence, by the bimodule $\hs{H}_i$.

On the other hand, if $\hs{V}$ is an even [odd] $\alg{A}$-bimodule, we can construct an $S^0$-real $\alg{A}$-bimodule $\hs{H}$ for any even [odd] $KO$-dimension $n \bmod 8$ such that $\hs{H}_i = \hs{V}$, by setting $\hs{H} := \hs{H}_i \oplus \hs{H}_{-i}$ for $\hs{H}_i := \hs{V}$, $\hs{H}_{-i} := \overline{\hs{V}}$, defining $\tilde{J} : \hs{H}_i \to \hs{H}_{-i}$ as the identity map on $\hs{V}$ viewed as an antiunitary $\hs{V} \to \overline{\hs{V}}$, then using the above formulas to define $J$, $\gamma$ (as necessary), $\lambda$, $\rho$, and finally setting $\epsilon = 1_\hs{V} \oplus (-1_{\overline{\hs{V}}})$. In the case that $\hs{V}$ is already $\hs{H}_i$ for some $S^0$-real bimodule $\hs{H}$, this procedure reproduces $\hs{H}$ up to unitary equivalence. We have therefore proved the following:

\begin{proposition}\label{s0reduction}
 Let $\alg{A}$ be a real \Cstar-algebra, and let $n \in \ring{Z}_8$. Then the map
\[
 \Bimod^0(\alg{A},n) \to
 \begin{cases}
  \Bimod(\alg{A}), &\text{if $n$ is odd,}\\
  \Bimod^\even(\alg{A}), &\text{if $n$ is even,}
 \end{cases}
\]
defined by $[\hs{H}] \mapsto [\hs{H}_i]$ is an isomorphism of monoids.
\end{proposition}

Now, let $\hs{H}$ is an $S^0$-real $\alg{A}$-bimodule, and suppose that $D$ is a Dirac operator for $\hs{H}$. We can define Dirac operators $D_i$ and $D_{-i}$ on $\hs{H}_i$ and $\hs{H}_{-i}$, respectively, by $D_{\pm i} := P_{\pm i} D P_{\pm i}$; then $D = D_i \oplus D_{-i}$ and, in fact, $D_{-i} = \epsp \tilde{J} D_i \tilde{J}^*$. Thus, a Dirac operator $D$ on $\hs{H}$ is completely determined by $D_i$; indeed, the map $D \mapsto D_i$ defines an isomorphism $\ms{D}_0(\alg{A},\hs{H},J,\epsilon) \cong \ms{D}_0(\alg{A},\hs{H})$. 

Along similar lines, one can show that $\lrU_{\alg{A}}(\hs{H},J) \cong \lrU_\alg{A}(\hs{H}_i)$ by means of the map $U \mapsto U_i := P_i U P_i$; this isomorphism is compatible with the isomorphism $\ms{D}_0(\alg{A},\hs{H},J,\epsilon) \cong \ms{D}_0(\alg{A},\hs{H})$. Hence, the functional equivalence between $\hs{H}$ and $\hs{H}_i$ holds at the level of moduli spaces of Dirac operators:

\begin{proposition}\label{diracs0reduction}
 Let $\hs{H}$ be an $S^0$-real $\alg{A}$-bimodule. Then
\begin{equation}
 \ms{D}(\alg{A},\hs{H},J,\epsilon) \cong \ms{D}(\alg{A},\hs{H}_i).
\end{equation}
\end{proposition}

One can similarly show that for a central subalgebra $\alg{C}$ of $\alg{A}$,
\[
 \ms{D}(\alg{A},\hs{H},J,\epsilon;\alg{C}) \cong \ms{D}(\alg{A},\hs{H}_i;\alg{C}).
\]

Let us conclude by considering the relation between orientability and Poincar{\'e} duality for an $S^0$-real bimodule $\hs{H}$ of even $KO$-dimension and orientability and Poincar{\'e} duality, respectively, for the associated even bimodule $\hs{H}_i$.

\begin{proposition}\label{s0orient}
 Let $\hs{H}$ be an $S^0$-real $\alg{A}$-bimodule of even $KO$-dimension. Then $\hs{H}$ is orientable if and only if there exist $a_1, \dotsc, a_k, b_1, \dotsc, b_k \in \alg{A}$ such that
\begin{equation}
 \gamma_i = \sum_{j=1}^k \lambda_i(a_j)\rho_i(b_j) = \epspp \sum_{j=1}^k \lambda_i(b_j^*)\rho_i(a_j^*).
\end{equation}
\end{proposition}

\begin{proof}
 Let $a_1, \dotsc, a_k$, $b_1, \dotsc, b_k \in \alg{A}$, and set $T = \sum_{j=1}^k \lambda(a_j)\rho(b_j)$. Then
\[
 T_i := P_i T P_i = \sum_{j=1}^k \lambda_i(a_j)\rho_i(b_j),
\]
while
\[
 \quad T_{-i} := P_{-i} T P_{-i} = \sum_{j=1}^k \lambda_{-i}(a_j)\rho_{-i}(b_j) = \tilde{J} \biggl( \sum_{j=1}^k \lambda_i(b_j^*)\rho_i(a_j^*) \biggr) \tilde{J}^*.
\]
Hence, $T_{-i} = \epspp \tilde{J} T_i \tilde{J}^*$ if and only if
\[
 \epspp \sum_{j=1}^k \lambda_i(b_j^*)\rho_i(a_j^*) = T_i = \sum_{j=1}^k \lambda_i(a_j)\rho_i(b_j).
\]
Applying this intermediate result to $a_j$ and $b_j$ such that $\gamma = \sum_{j=1}^k \lambda(a_j)\rho(b_j)$, in the case that $\hs{H}$ is orientable, and then to $a_j$ and $b_j$ such that $\gamma_i = \sum_{j=1}^k \lambda_i(a_j)\rho_i(b_j)$, in the case that $\hs{H}_i$ is orientable, yields the desired result.
\end{proof}

Thus, orientability of an $S^0$-real bimodule $\hs{H}$ is equivalent to a stronger version of orientability on the bimodule $\hs{H}_i$.

Turning to Poincar{\'e} duality, we can obtain the following result:

\begin{proposition}\label{s0poincare}
 Let $\hs{H}$ be an $S^0$-real $\alg{A}$-bimodule of even $KO$-dimension with intersection form $\form{\cdot,\cdot}$, and let $\form{\cdot,\cdot}_i$ be the intersection form for $\hs{H}_i$. Then for any $p$, $q \in KO_0(\alg{A})$,
\[
 \form{p,q} = \form{p,q}_i + \epspp \form{q,p}_i.
\]
\end{proposition}

\begin{proof}
 Let $e$, $f \in \alg{A}$ be projections. Then
\begin{align*}
 \form{[e],[f]} &= \tr(\gamma\lambda(e)\rho(f))\\
 &= \tr(\gamma_i \lambda_i(e) \rho_i(f)) + \tr(\gamma_{-i}\lambda_{-i}(e)\rho_{-i}(f))\\
 &= \tr(\gamma_i \lambda_i(e) \rho_i(f)) + \epspp \tr(\tilde{J}\gamma_i \lambda_i(f) \rho_i(e) \tilde{J}^*)\\
 &= \tr(\gamma_i \lambda_i(e) \rho_i(f)) + \epspp \overline{\tr(\gamma_i \lambda_i(f) \rho_i(e))}\\
 &= \form{[e],[f]}_i + \epspp \form{[f],[e]}_i,
\end{align*}
where we have used the fact that the intersection forms are integer-valued.
\end{proof}

Thus, Poincar{\'e} duality on an $S^0$-real bimodule $\hs{H}$ is equivalent to nondegeneracy of either the symmetrisation or antisymmetrisation of the intersection form on $\hs{H}_i$, as the case may be.

\section{Bimodules and Multiplicity Matrices}

We now turn to the study of bimodules, and in particular, to their characterisation by multiplicity matrices. We shall find that a bimodule admits, up to unitary equivalence, at most one real structure of any given $KO$-dimension, and that the multiplicity matrix or matrices of a bimodule will determine entirely which real structures, if any, it does admit.

In what follows, $\alg{A}$ will be a fixed real \Cstar-algebra.

\subsection{Odd bimodules}

Let us begin with the study of odd bimodules.

For $m \in M_S(\ZO)$, we define an $\alg{A}$-bimodule $\hs{H}_m$ by setting
\begin{align*}
 \hs{H}_m &:= \bigoplus_{\alpha,\beta \in \spec{\alg{A}}} \field{C}^{n_\alpha} \otimes \field{C}^{m_{\alpha\beta}} \otimes \field{C}^{n_\beta},\\
 \lambda_m(a) &:= \bigoplus_{\alpha,\beta \in \spec{\alg{A}}} \lambda_\alpha(a) \otimes 1_{m_{\alpha\beta}} \otimes 1_{n_\beta}, \quad a \in \alg{A},\\
 \rho_m(a) &:= \bigoplus_{\alpha, \beta \in \spec{\alg{A}}} 1_{n_\alpha} \otimes 1_{m_{\alpha\beta}} \otimes \lambda_\beta(a)^T, \quad a \in \alg{A}.
\end{align*}
Here we use the convention that $1_n$ is the identity on $\field{C}^n$, with $\field{C}^0 := \{0\}$ and hence $1_0 := 0$. 

\begin{proposition}[Krajewski~\cite{Kraj98}*{\S 3.1}, Paschke--Sitarz~\cite{PS98}*{Lemmas 1, 2}] \label{oddmult}
The map $\bimod : M_S(\ZO) \to \Bimod(\alg{A})$ given by $m \mapsto [\hs{H}_m]$ is an isomorphism of monoids.
\end{proposition}

\begin{proof}
By construction, $\bimod$ is an injective morphism of monoids. It therefore suffices to show that $\bimod^\odd$ is also surjective.

Now, let $\hs{H}$ be an $\alg{A}$-bimodule. For $\alpha \in \spec{\alg{A}}$ define projections $P^L_\alpha$ and $P^R_\alpha$ by
\[
 P^L_\alpha :=
 \begin{cases}
  \lambda(e_i) &\text{if $\alpha = \rep{n}_i$ for $\field{K}_i \neq \field{C}$,}\\
  \frac{1}{2}\left(\lambda(e_i) - i \lambda(i e_i) \right) &\text{if $\alpha = \rep{n}_i$ for $\field{K}_i = \field{C}$,}\\
  \frac{1}{2}\left(\lambda(e_i) + i \lambda(i e_i) \right) &\text{if $\alpha = \crep{n}_i$ for $\field{K}_i = \field{C}$,}
 \end{cases}
\]
and
\[
 P^R_\alpha :=
 \begin{cases}
  \rho(e_i) &\text{if $\alpha = \rep{n}_i$ for $\field{K}_i \neq \field{C}$,}\\
  \frac{1}{2}\left(\rho(e_i) - i \rho(i e_i) \right) &\text{if $\alpha = \rep{n}_i$ for $\field{K}_i = \field{C}$,}\\
  \frac{1}{2}\left(\rho(e_i) + i \rho(i e_i) \right) &\text{if $\alpha = \crep{n}_i$ for $\field{K}_i = \field{C}$,}
 \end{cases}
\]
respectively; by construction, $P^L_\alpha \in \lambda(\alg{A}) + i \lambda(\alg{A})$ and $P^R_\alpha \in \rho(\alg{A}) + i \rho(\alg{A})$, so that for $\alpha$, $\beta \in \spec{\alg{A}}$, $P^L_\alpha$ and $P^R_\beta$ commute. We can therefore define projections $P_{\alpha\beta} := P^L_\alpha P^R_\beta$ for each $\alpha$, $\beta \in \spec{\alg{A}}$; it is then easy to see that each $\hs{H}_{\alpha\beta} := P_{\alpha\beta}\hs{H}$ is a sub-$\alg{A}$-bimodule of $\hs{H}$, and that $\hs{H} = \oplus_{\alpha,\beta \in \spec{\alg{A}}} \hs{H}_{\alpha\beta}$.

Let $\alpha$, $\beta \in \spec{\alg{A}}$. As noted before, the left action of $\alg{A}$ on $\hs{H}_{\alpha\beta}$ must decompose as a direct sum of irreducible representations, but by construction of $\hs{H}_{\alpha\beta}$, those irreducible representations must all be $\alpha$. Similarly, the right action on $\hs{H}_{\alpha\beta}$ must be a direct sum of copies of $\beta$. Since the left action and right action commute, we must therefore have that $\hs{H}_{\alpha\beta} \cong \hs{H}_{m_{\alpha\beta} E_{\alpha\beta}}$ for some $m_{\alpha\beta} \in \ZO$. Taking the direct sum of the $\hs{H}_{\alpha\beta}$, we therefore see that $\hs{H}$ is unitarily equivalent to $\hs{H}_m$ for $m = (m_{\alpha\beta}) \in M_S(\ZO)$, that is, $[\hs{H}] = \bimod(m)$.
\end{proof}

We denote the inverse map $\bimod^{-1} : \Bimod(\alg{A}) \to M_S(\ZO)$ by $\mult$.

\begin{definition}
 Let $\hs{H}$ be an $\alg{A}$-bimodule. Then the \term{multiplicity matrix} of $\alg{A}$ is the matrix $\mult[\hs{H}] \in M_S(\ZO)$.
\end{definition}

From now on, without any loss of generality, we shall assume that an $\alg{A}$-bimodule $\hs{H}$ with multiplicity matrix $m$ is $\hs{H}_m$ itself.

\begin{remark}
 Multiplicity matrices readily admit a $K$-theoretic interpretation~\cite{Ell08}. For simplicity, suppose that $\alg{A}$ is a complex \Cstar-algebra and consider only complex-linear representations. Then for $\hs{H}$  an $\alg{A}$-bimodule, $\mult[\hs{H}]$ is essentially the \term{Bratteli matrix} of the inclusion $\lambda(\alg{A}) \hookrightarrow \rho(\alg{A})^\prime \subset \bdd(\hs{H})$ (cf.~\cite{Ell07}*{\S 2}), and can thus be interpreted as representing the induced map $K_0(\lambda(\alg{A})) \to K_0(\rho(\alg{A})^\prime)$ in complex $K$-theory. Likewise, $\mult[\hs{H}]^T$ can be interpreted as representing the map $K_0(\rho(\alg{A})) \to K_0(\lambda(\alg{A})^\prime)$ induced by the inclusion $\rho(\alg{A}) \hookrightarrow \lambda(\alg{A})^\prime \subset \bdd(\hs{H})$. Similar interpretations can be made in the more general context of real \Cstar-algebras and $KO$-theory.
\end{remark}

We shall now characterise left, right, and left and right $\alg{A}$-linear maps between $\alg{A}$-bimodules. Let $\hs{H}$ and $\hs{H}^\prime$ be $\alg{A}$-bimodules with multiplicity matrices $m$ and $m^\prime$, respectively, let $P_{\alpha\beta}$ be the projections on $\hs{H}$ defined as in the proof of Proposition~\ref{oddmult}, and let $P_{\alpha\beta}^\prime$ be the analogous projections on $\hs{H}^\prime$. Then any linear map $T : \hs{H} \to \hs{H}^\prime$ is characterised by the components
\begin{equation}
 T_{\alpha\beta}^{\gamma\delta} := P^\prime_{\gamma\delta} T P_{\alpha\beta},
\end{equation}
which we view as maps $T_{\alpha\beta}^{\gamma\delta} : \field{C}^{n_\alpha} \otimes \field{C}^{m_{\alpha\beta}} \otimes \field{C}^{n_\beta} \to \field{C}^{n_\gamma} \otimes \field{C}^{m_{\gamma\delta}^\prime} \otimes \field{C}^{n_\delta}$, or equivalently, as elements $
 T_{\alpha\beta}^{\gamma\delta} \in M_{n_\gamma \times n_\alpha}(\field{C}) \otimes M_{m_{\gamma\delta}^\prime \times m_{\alpha\beta}} \otimes M_{n_\delta \times n_\beta}(\field{C})$. Thus we have an isomorphism
\[
 \comp : \bdd(\hs{H},\hs{H}^\prime) \to \bigoplus_{\alpha,\beta,\gamma,\delta \in \spec{\alg{A}}} M_{n_\gamma \times n_\alpha}(\field{C}) \otimes M_{m_{\gamma\delta}^\prime \times m_{\alpha\beta}} \otimes M_{n_\delta \times n_\beta}(\field{C})
\]
given by $\comp(T) := (T_{\alpha\beta}^{\gamma\delta})_{\alpha,\beta,\gamma,\delta \in \spec{\alg{A}}}$. Note that when $\hs{H} = \hs{H}^\prime$, $T$ is self-adjoint if and only if $T_{\gamma\delta}^{\alpha\beta} = (T_{\alpha\beta}^{\gamma\delta})^*$ for all $\alpha$, $\beta$, $\gamma$, $\delta \in \spec{\alg{A}}$.

\begin{proposition}[Krajewski~\cite{Kraj98}*{\S3.4}]\label{linear}
 Let $\hs{H}$ and $\hs{H}^\prime$ be $\alg{A}$-bimodules with multiplicity matrices $m$ and $m^\prime$, respectively. Then
\begin{align}
 \comp(\lbdd_{\alg{A}}(\hs{H},\hs{H}^\prime)) &= \bigoplus_{\alpha,\beta,\delta \in \spec{\alg{A}}} 1_{n_\alpha} \otimes M_{m_{\alpha\delta}^\prime \times m_{\alpha\beta}}(\field{C}) \otimes M_{n_\delta \times n_\beta}(\field{C}),\\
 \comp(\rbdd_\alg{A}(\hs{H},\hs{H}^\prime)) &= \bigoplus_{\alpha,\beta,\gamma \in \spec{\alg{A}}} M_{n_\gamma \times n_\alpha}(\field{C}) \otimes M_{m_{\gamma\beta}^\prime \times m_{\alpha\beta}}(\field{C}) \otimes 1_{n_\beta},\\
 \comp(\lrbdd_{\alg{A}}(\hs{H},\hs{H}^\prime)) &= \bigoplus_{\alpha,\beta \in \spec{\alg{A}}} 1_{n_\alpha} \otimes M_{m_{\alpha\beta}^\prime \times m_{\alpha\beta}}(\field{C}) \otimes 1_{n_\beta}.
\end{align} 
\end{proposition}

\begin{proof}
Observe that $T \in \bdd(\hs{H},\hs{H}^\prime)$ is left, right, or left and right $\alg{A}$-linear if and only if each $T_{\alpha\beta}^{\gamma\delta}$ is left, right, or left and right $\alg{A}$. Thus, let $\alpha$, $\beta$, $\gamma$ and $\delta \in \spec{\alg{A}}$ be fixed, and let $T \in M_{n_\gamma \times n_\alpha}(\field{C}) \otimes M_{m_{\gamma\delta}^\prime \times m_{\alpha\beta}} \otimes M_{n_\delta \times n_\beta}(\field{C})$.

First, write $T = \sum_{i=1}^k A_i \otimes B_i$ for $A_i \in M_{n_\gamma \times n_\alpha}(\field{C})$ and for linearly independent $B_i \in M_{m_{\gamma\delta}^\prime \times m_{\alpha\beta}} \otimes M_{n_\delta \times n_\beta}(\field{C})$. Then, for $a \in \alg{A}$, 
\[
 (\lambda_\gamma(a) \otimes 1_{m_{\gamma\delta}^\prime} \otimes 1_{n_\delta}) T - T (\lambda_\alpha(a) \otimes 1_{m_{\alpha\beta}} \otimes 1_{n_\beta}) = \sum_{i=1}^k (\lambda_\gamma(a) A_i - A_i \lambda_\alpha(a)) \otimes B_i, 
\]
so that by linear independence of the $B_i$, $T$ is left $\alg{A}$-linear if and only if each $A_i$ intertwines the irreducible representations $\alpha$ and $\gamma$, and hence, by Schur's lemma, if and only if $\alpha = \gamma$ and each $A_i$ is a constant multiple of $1_{n_\alpha}$ or each $A_i = 0$. Thus,
\begin{multline*}
 \lbdd_{\alg{A}}(\field{C}^{n_\alpha} \otimes \field{C}^{m_{\alpha\beta}} \otimes \field{C}^{n_\beta}, \field{C}^{n_\gamma} \otimes \field{C}^{m_{\gamma\delta}^\prime} \otimes \field{C}^{n_\delta})\\ =
 \begin{cases}
  1_{n_\alpha} \otimes M_{m_{\alpha\delta}^\prime \times m_{\alpha\beta}}(\field{C}) \otimes M_{n_\delta \times n_\beta}(\field{C}) &\text{if $\alpha = \gamma$,}\\
  \{0\} &\text{otherwise.}
 \end{cases}
\end{multline*}

Analogously, one can show that
\begin{multline*}
 \rbdd_{\alg{A}}(\field{C}^{n_\alpha} \otimes \field{C}^{m_{\alpha\beta}} \otimes \field{C}^{n_\beta}, \field{C}^{n_\gamma} \otimes \field{C}^{m_{\gamma\delta}^\prime} \otimes \field{C}^{n_\delta})\\ = 
 \begin{cases}
  M_{n_\gamma \times n_\alpha}(\field{C}) \otimes M_{m_{\gamma\beta}^\prime \times m_{\alpha\beta}}(\field{C}) \otimes 1_{n_\beta} &\text{if $\beta = \delta$,}\\
  \{0\} &\text{otherwise},
 \end{cases}
\end{multline*}
and then these first two results together imply that
\begin{multline*}
 \lrbdd_{\alg{A}}(\field{C}^{n_\alpha} \otimes \field{C}^{m_{\alpha\beta}} \otimes \field{C}^{n_\beta}, \field{C}^{n_\gamma} \otimes \field{C}^{m_{\gamma\delta}^\prime} \otimes \field{C}^{n_\delta})\\ = 
 \begin{cases}
  1_{n_\alpha} \otimes M_{m_{\alpha\beta}^\prime \times m_{\alpha\beta}}(\field{C}) \otimes 1_{n_\beta} &\text{if $(\alpha,\beta)=(\gamma,\delta)$,}\\
  \{0\} &\text{otherwise,}
 \end{cases}
\end{multline*}
as was claimed.
\end{proof}

An immediate consequence is the following description of the group $\lrU_\alg{A}(\hs{H})$:

\begin{corollary}\label{oddunitary}
 Let $\hs{H}$ be an $\alg{A}$-bimodule. Then
\[
 \comp(\lrU_\alg{A}(\hs{H})) = \bigoplus_{\alpha,\beta \in \spec{\alg{A}}} 1_{n_\alpha} \otimes \U(m_{\alpha\beta}) \otimes 1_{n_\beta} \cong \prod_{\alpha,\beta \in \spec{\alg{A}}} \U(m_{\alpha\beta}),
\]
with the convention that $\U(0) = \{0\}$ is the trivial group.
\end{corollary}

\subsection{Even bimodules}

We now turn to the study of even bimodules; let us begin by considering the decomposition of an even bimodule into its even and odd sub-bimodules.

Let $(\hs{H},\gamma)$ be an even $\alg{A}$-bimodule. Define mutually orthogonal projections $P^\even$ and $P^\odd$ by 
\[
 P^\even = \frac{1}{2}(1+\gamma), \quad P^\odd = \frac{1}{2}(1-\gamma).
\]
We can then define sub-bimodules $\hs{H}^\even$ and $\hs{H}^\odd$ of $\hs{H}$ by $\hs{H}^\even = P^\even \hs{H}$, $\hs{H}^\odd = P^\odd \hs{H}$; one has that $\hs{H} = \hs{H}^\even \oplus \hs{H}^\odd$ at the level of bimodules.

On the other hand, given $\alg{A}$-bimodules $\hs{H}_1$ and $\hs{H}_2$, we can construct an even $\alg{A}$-bimodule $(\hs{H},\gamma)$ such that $\hs{H}^\even = \hs{H}_1$ and $\hs{H}^\odd = \hs{H}_2$ by setting $\hs{H} = \hs{H}_1 \oplus \hs{H}_2$ and $\gamma = 1_{\hs{H}_1} \oplus (-1_{\hs{H}_2})$. If $\hs{H}_1$ and $\hs{H}_2$ are already $\hs{H}^\even$ and $\hs{H}^\odd$ for some $(\hs{H},\gamma)$, then this procedure precisely reconstructs $(\hs{H},\gamma)$. Since this procedure manifestly respects direct summation and unitary equivalence at either end, we have therefore proved the following:

\begin{proposition}\label{evensplit}
 Let $\alg{A}$ be a real \Cstar-algebra. The map 
\[
 C : \Bimod^\even(\alg{A}) \to \Bimod(\alg{A}) \times \Bimod(\alg{A})
\]
given by
\[
 C([\hs{H}]) := ([\hs{H}^\even],[\hs{H}^\odd])
\]
is an isomorphism of monoids.
\end{proposition}

One readily obtains a similar decomposition at the level of unitary groups:

\begin{corollary}
Let $(\hs{H},\gamma)$ be an even $\alg{A}$-bimodule. Then
\[
 \lrU_\alg{A}(\hs{H},\gamma) = \lrU_\alg{A}(\hs{H}^\even) \oplus \lrU_\alg{A}(\hs{H}^\odd).
\]
\end{corollary}

Another immediate consequence is the following analogue of Proposition~\ref{oddmult}:

\begin{proposition}\label{evenmult}
 Let $\alg{A}$ be a real \Cstar-algebra. The map 
\[
 \bimod^\even : M_S(\ZO) \times M_S(\ZO) \to \Bimod^\even(\alg{A})
\]
defined by $bimod^\even := C^{-1} \circ (\bimod \times \bimod)$ is an isomorphism of monoids.
\end{proposition}
Just as in the odd case, we will find it convenient to denote $(\bimod^\even)^{-1} : \Bimod^\even(\alg{A}) \to M_S(\ZO) \times M_S(\ZO)$ by $\mult^\even$. It then follows that $\mult^\even = (\mult \times \mult) \circ C$.

\begin{definition}
 Let $(\hs{H},\gamma)$ be an even $\alg{A}$-bimodule. Then the \term{multiplicity matrices} of $(\hs{H},\gamma)$ are the pair of matrices
\[
 (\mult[\hs{H}^\even],\mult[\hs{H}^\odd]) = \mult^\even[(\hs{H},\gamma)] \in M_S(\ZO) \times M_S(\ZO).
\]
\end{definition}

Let us now consider orientability of even bimodules.

\begin{lemma}[Krajewski~\cite{Kraj98}*{\S3.4}]\label{orientable}
 Let $(\hs{H},\gamma)$ be an even $\alg{A}$-bimodule. Then $(\hs{H},\gamma)$ is orientable only if $\lrbdd(\hs{H}^\even,\hs{H}^\odd)=\{0\}$.
\end{lemma}

\begin{proof}
 Suppose that $(\hs{H},\gamma)$ is orientable, so that $\gamma = \sum_{i=1}^k \lambda(a_i)\rho(b_i)$ for some $a_1,\dotsc,a_k$, $b_1,\dotsc,b_k \in \alg{A}$. Now, let $T \in \lrbdd_\alg{A}(\hs{H}^\even,\hs{H}^\odd)$, and define $\tilde{T} \in \lrbdd_\alg{A}(\hs{H})$ by
\[
 \tilde{T} =
 \begin{pmatrix}
  0 & T^*\\
  T & 0
 \end{pmatrix}.
\]
Then, on the one hand, since $\gamma = 1_{\hs{H}^\even} \oplus (-1_{\hs{H}^\odd})$, $\tilde{T}$ anticommutes with $\gamma$, and on the other, since $\gamma = \sum_{i=1}^k \lambda(a_i)\rho(b_i)$, $\tilde{T}$ commutes with $\gamma$, so that $\tilde{T} = 0$. Hence, $T = 0$.
\end{proof}

This last result motivates the following weaker notion of orientability:

\begin{definition}
 An even $\alg{A}$-bimodule $(\hs{H},\gamma)$ shall be called \term{quasi-orientable} whenever $\lrbdd_\alg{A}(\hs{H}^\even,\hs{H}^\odd) = \{0\}$.
\end{definition}

The subset of $\Bimod^\even(\alg{A})$ consisting of the unitary equivalence classes of the quasi-orientable even $\alg{A}$-bimodules will be denoted by $\Bimod^\even_q(\alg{A})$. 

We define the \term{support} of a real $p \times q$ matrix $A$ to be the set
\[
 \supp(A) := \{(i,j) \in \{1,\dotsc,p\} \times \{1,\dotsc,q\} \mid A_{ij} \neq 0\}.
\]
For $A \in M_S(\field{R})$, we shall view $\supp(A)$ as a subset of $\spec{\alg{A}}^2$ by means of the identification of $\{1,\dotsc,S\}$ with $\spec{\alg{A}}$ as ordered sets. We shall also find it convenient to associate to each matrix $m \in M_S(\ring{Z})$ a matrix $\widehat{m} \in M_N(\ring{Z})$ by
\begin{equation}
 \widehat{m}_{ij} := \sum_{\alpha \in \spec{M_{n_i}(\field{K}_i)}} \sum_{\beta \in \spec{M_{n_j}(\field{K}_j)}} m_{\alpha\beta}.
\end{equation}
One can check the map $M_S(\ring{Z}) \to M_N(\ring{Z})$ defined by $m \mapsto \widehat{m}$ is linear and respects transposes.

We can now offer the following characterisation of quasi-orientable bimodules:

\begin{proposition}[Krajewski~\cite{Kraj98}*{\S3.3}, Paschke--Sitarz~\cite{PS98}*{Lemma 3}]
 Let $\alg{A}$ be a real \Cstar-algebra. Then
\begin{multline}
 \mult^\even(\Bimod^\even_q(\alg{A})) \\ = \{(m^\even,m^\odd) \in M_S(\ZO)^2 \mid \supp(m^\even) \cap \supp(m^\odd) = \emptyset\}.
\end{multline}
\end{proposition}

\begin{proof}
Let $(\hs{H},\gamma)$ be an even $\alg{A}$-bimodule and let $(m^\even,m^\odd)$ be its multiplicity matrices. Then by Proposition~\ref{linear},
\[
 \lrbdd_\alg{A}(\hs{H}^\even,\hs{H}^\odd) \cong \bigoplus_{\alpha,\beta \in \spec{\alg{A}}} M_{m_{\alpha\beta}^\odd \times m_{\alpha\beta}^\even}(\field{C}),
\]
whence the result follows immediately.
\end{proof}

We therefore define the \term{signed multiplicity matrix} of a quasi-orientable even $\alg{A}$-bimodule $(\hs{H},\gamma)$, or rather, the unitary equivalence class thereof, to be the matrix
\[
 \mult_q[(\hs{H},\gamma)] := \mult[\hs{H}^\even] - \mult[\hs{H}^\odd] \in M_S(\ring{Z}).
\]
The map $\Bimod^\even_q(\alg{A}) \to M_S(\ring{Z})$ defined by 
\[
 [(\hs{H},\gamma)] \mapsto \mult_q[(\hs{H},\gamma)]
\]
is then bijective, and $\mult^\even[(\hs{H},\gamma)]$ is readily recovered from $\mult_q[(\hs{H},\gamma)]$. Indeed, if $(\hs{H},\gamma)$ is a quasi-orientable even $\alg{A}$-bimodule with signed multiplicity matrix $\mu$, then (cf.~\cite{PS98}*{Lemma 3},\cite{Kraj98}*{3.3})
\begin{equation}
 \gamma = \bigoplus_{\alpha,\beta \in \spec{\alg{A}}} \mu_{\alpha\beta} 1_{\hs{H}_{\alpha\beta}}.
\end{equation}
These algebraic consequences of quasi-orientability, which were derived from the stronger condition of orientability in the original papers \cite{PS98} and \cite{Kraj98}, are key to the formalism developed by Krajewski and Paschke--Sitarz, and hence to the later work by Iochum, Jureit, Sch{\"u}cker, and Stephan~\cites{ACG1,ACG2,ACG3,Sch05}.

We can now characterise orientable bimodules amongst quasi-orientable bimodules:

\begin{proposition}[Krajewski~\cite{Kraj98}*{\S 3.3}]\label{converseorientable}
 Let $(\hs{H},\gamma)$ be a quasi-orientable $\alg{A}$-bimodule with signed multiplicity matrix $\mu$. Then $(\hs{H},\gamma)$ is orientable if and only if the following conditions all hold:
\begin{enumerate}
 \item For each $i \in \{1,\dotsc,N\}$ such that $\field{K}_i = \field{C}$ and all $\beta \in \spec{\alg{A}}$,
\[
 \mu_{\rep{n}_i \beta}\mu_{\crep{n}_i \beta} \geq 0;
\]
 \item For all $\alpha \in \spec{\alg{A}}$ and each $j \in \{1,\dotsc,N\}$ such that $\field{K}_j = \field{C}$, 
\[
 \mu_{\alpha \rep{n}_j}\mu_{\alpha \crep{n}_j} \geq 0;
\]
 \item For all $i$, $j \in \{1,\dotsc,N\}$ such that $\field{K}_i = \field{K}_j = \field{C}$,
\[
 \mu_{\rep{n}_i \crep{n}_j}\mu_{\crep{n}_i \rep{n}_j} \geq 0.
\]
\end{enumerate}
In particular, if $(\hs{H},\gamma)$ is orientable, then
\begin{equation}
 \gamma = \sum_{i,j=1}^N \lambda(\sgn(\widehat{\mu}_{ij})e_i)\rho(e_j).
\end{equation}
\end{proposition}

\begin{proof}
 First, suppose that $(\hs{H},\gamma)$ is indeed orientable, so that there exist $a_1,\dotsc,a_n$, $b_1,\dotsc,b_n \in \alg{A}$ such that $\gamma = \sum_{l=1}^n \lambda(a_l)\rho(b_l)$; in particular, then, for each $\alpha$, $\beta \in \spec{\alg{A}}$,
\[
 \sgn(\mu_{\alpha\beta}) 1_{n_\alpha} \otimes 1_{|\mu_{\alpha\beta}|} \otimes 1_{n_\beta} = \gamma_{\alpha\beta}^{\alpha\beta} = \sum_{l=1}^n \lambda_\alpha(a_l) \otimes 1_{|\mu_{\alpha\beta}|} \otimes \lambda_\beta(b_l)^T.
\]

Now, let $i \in \{1,\dotsc,N\}$ be such that $\field{K}_i = \field{C}$, and let $\beta \in \spec{\alg{A}}$, and suppose that $\mu_{\rep{n}_i \beta}$ and $\mu_{\crep{n}_i \beta}$ are both non-zero. It then follows that
\[
 \sgn(\mu_{\rep{n}_i \beta}) 1_{n_i} \otimes 1_{n_\beta} = \sum_{l=1}^n (a_l)_i \otimes \lambda_\beta(b_l)^T, \quad
 \sgn(\mu_{\crep{n}_i \beta}) 1_{n_i} \otimes 1_{n_\beta} = \sum_{l=1}^n \overline{(a_l)_i} \otimes \lambda_\beta(b_l)^T,
\]
where $(a_l)_i$ denotes the component of $a_l$ in the direct summand $M_{k_i}(\field{C})$ of $\alg{A}$. If $X$ denotes complex conjugation on $\field{C}^{n_i}$, it then follows from this that
\[
 \sgn(\mu_{\rep{n}_i \beta}) 1_{n_i} \otimes 1_{n_\beta} = (X \otimes 1_{n_\beta}) (\sgn(\mu_{\rep{n}_i \beta}) 1_{n_i} \otimes 1_{n_\beta}) (X \otimes 1_{n_\beta}) = \sgn(\mu_{\crep{n}_i \beta}) 1_{n_i} \otimes 1_{n_\beta},
\]
so that $\sgn(\mu_{\rep{n}_i \beta}) = \sgn(\mu_{\crep{n}_i \beta})$, or equivalently $\mu_{\rep{n}_i \beta}\mu_{\crep{n}_i \beta} > 0$. One can similarly show that the other two conditions hold.

Now, suppose instead that the three conditions on $\mu$ hold. Then for all $i$, $j \in \{1,\dotsc,N\}$, all non-zero entries $\mu_{\alpha\beta}$ for $\alpha \in \spec{M_{k_i}(\field{K}_i)}$, $\beta \in \spec{M_{k_j}(\field{K}_j)}$, have the same sign, so set $\gamma_{ij}$ equal to this common value of non-zero $\sgn(\mu_{\alpha\beta})$ if at least one such $\mu_{\alpha\beta}$ is non-zero, and set $\gamma_{ij} = 0$ otherwise. One can then easily check that $\gamma = \sum_{i,j=1}^{N} \lambda(\gamma_{ij}e_i)\rho(e_j)$, so that $(\hs{H},\gamma)$ is indeed orientable. Moreover, using the same three conditions, one can readily check that $\gamma_{ij} = \sgn(\widehat{\mu}_{ij})$, which yields the last part of the claim.
\end{proof}

Let us now turn to intersection forms and Poincar{\'e} duality. In particular, we are now able to provide explicit expressions for intersection forms in terms of multiplicity matrices.

Recall that for $\field{K} = \field{R}$, $\field{C}$ or $\field{H}$, $KO_0(M_k(\field{K}))$ is the infinite cyclic group generated by $[p]$ for $p \in M_k(\field{K})$ a minimal projection, so that for $\alg{A}$ a real \Cstar-algebra with Wedderburn decomposition $\oplus_{i=1}^N M_{n_i}(\field{K}_i)$,
\[
 KO_0(\alg{A}) \cong \prod_{i=1}^N KO_0(M_{n_i}(\field{K}_i)) \cong \ring{Z}^N,
\]
which can be viewed as the infinite abelian group generated by $\{[p_i]\}_{i=1}^N$ for $p_i$ a minimal projection in $M_{n_i}(\field{K}_i)$. Since
\[
 \tau_i := \tr(p_i) =
 \begin{cases}
  2 &\text{if $\field{K}_i = \field{H}$,}\\
  1 &\text{otherwise,}
 \end{cases}
\]
it follows that for $\alpha \in \spec{\alg{A}}$,
\begin{equation}
 \tr(\lambda_\alpha(p_i)) =
 \begin{cases}
  \tau_i &\text{if $\alpha \in \spec{M_{n_i}(\field{K}_i)}$,}\\
  0 &\text{otherwise.}\\
 \end{cases}
\end{equation}

Now, if $(\hs{H},\gamma)$ is an even $\alg{A}$-bimodule with intersection form $\form{\cdot,\cdot}$, we can define a matrix $\cap \in M_N(\ring{Z})$ by
\begin{equation}
 \cap_{ij} := \form{[p_i],[p_j]}.
\end{equation}
The intersection form $\form{\cdot,\cdot}$ is completely determined by the matrix $\cap$, and in particular, $\form{\cdot,\cdot}$ is non-degenerate (\ie $(\hs{H},\gamma)$ satisfies Poincar{\'e} duality) if and only if $\cap$ is non-degenerate.

\begin{proposition}[Krajewski~\cite{Kraj98}*{\S 3.3}, Paschke--Sitarz~\cite{PS98}*{\S 2.4}]\label{intform}
 Let $(\hs{H},\gamma)$ be an even $\alg{A}$-bimodule with pair of multiplicity matrices $(m^\even,m^\odd)$. Then
\begin{equation}
 \cap_{ij} = \tau_i \tau_j \left(\widehat{m^\even}_{ij} - \widehat{m^\odd}_{ij}\right),
\end{equation}
so that $(\hs{H},\gamma)$ satisfies Poincar{\'e} duality if and only if the matrix $\widehat{m^\even} - \widehat{m^\odd}$ is non-degenerate.
\end{proposition}

\begin{proof}
First, since $\hs{H} = \hs{H}^\even \oplus \hs{H}^\odd$, we can write
\[
 \gamma = \bigoplus_{\alpha,\beta\in\spec{\alg{A}}} 1_{n_\alpha} \otimes \gamma_{\alpha\beta} \otimes 1_{n_\beta},
\]
where $\gamma_{\alpha\beta} = 1_{m_{\alpha\beta}^\even} \oplus (-1_{m_{\alpha\beta}^\odd})$. Then,
\begin{align*}
\cap_{ij} &= \form{[p_i],[p_j]}\\
&= \tr(\gamma\lambda(p_i)\rho(p_j))\\ 
&= \tr\left(\bigoplus_{\alpha,\beta\in\spec{\alg{A}}} \lambda_\alpha(p_i) \otimes \gamma_{\alpha\beta} \otimes \lambda_\beta(p_j)\right)\\
&= \sum_{\alpha,\beta \in \spec{\alg{A}}} \tr(\lambda_\alpha(p_i))\tr(\lambda_\beta(p_j))(m_{\alpha\beta}^\even - m_{\alpha\beta}^\odd)\\
&= \sum_{i,j=1}^N \tau_i \tau_j (\widehat{m^\even}_{ij} - \widehat{m^\odd}_{ij}).
\end{align*}
This calculation implies, in particular, that $\cap$ can be obtained from $\widehat{m^\even} - \widehat{m^\odd}$ by a finite sequence of elementary row or column operations, so that $\cap$ is indeed non-degenerate if and only if $\widehat{m^\even} - \widehat{m^\odd}$ is.
\end{proof}

\begin{corollary}
 Let $(\hs{H},\gamma)$ be a quasi-orientable $\alg{A}$-bimodule with signed multiplicity matrix $\mu$. Then $(\hs{H},\gamma)$ satisfies Poincar{\'e} duality if and only if $\widehat{\mu}$ is non-degenerate.
\end{corollary}

In particular, if we restrict ourselves to complex \Cstar-algebras and complex-linear representations, a quasi-orientable bimodule is completely characterised by the $K$-theoretic datum of its intersection form.

\subsection{Real bimodules of odd $KO$-dimension}

Let us now consider real bimodules of odd $KO$-dimension. Before continuing, recall that
\[
 \Bimod(\alg{A},1) = \Bimod(\alg{A},7), \quad \Bimod(\alg{A},3) = \Bimod(\alg{A},5).
\]

For $m \in \Sym_S(\ZO)$, we define an antilinear operator $X_m$ on $\hs{H}_m$ by defining $(X_m)_{\alpha\beta}^{\gamma\delta} : \field{C}^{n_\alpha} \otimes \field{C}^{m_{\alpha\beta}} \otimes \field{C}^{n_\beta} \to \field{C}^{n_\gamma} \otimes \field{C}^{m_{\gamma\delta}} \otimes \field{C}^{n_\delta}$ by
\begin{equation}
 (X_m)_{\alpha\beta}^{\beta\alpha}: \xi_1 \otimes \xi_2 \otimes \xi_3 \mapsto \overline{\xi_3} \otimes \overline{\xi_2} \otimes \overline{\xi_1},
\end{equation}
and by setting $(X_m)_{\alpha\beta}^{\gamma\delta} = 0$ whenever $(\gamma,\delta) \neq (\beta,\alpha)$.

\subsubsection{$KO$-dimension $1$ or $7 \bmod 8$}

We begin by determining the form of the multiplicity matrix for a real bimodule of $KO$-dimension $1$ or $7 \bmod 8$.

\begin{lemma}[Krajewski~\cite{Kraj98}*{\S3.2}, Paschke--Sitarz~\cite{PS98}*{Lemma 4}]\label{real17a}
 Let $(\hs{H},J)$ be a real $\alg{A}$-bimodule of $KO$-dimension $1$ or $7 \bmod 8$ with multiplicity matrix $m$. Then $m$ is symmetric, and the only non-zero components of $J$ are of the form $J_{\alpha\beta}^{\beta\alpha}$ for $\alpha$, $\beta \in \spec{\alg{A}}$, which are anti-unitaries $\hs{H}_{\alpha\beta} \to \hs{H}_{\beta\alpha}$ satisfing the relations $J_{\beta\alpha}^{\alpha\beta} = (J_{\alpha\beta}^{\beta\alpha})^*$.
\end{lemma}

\begin{proof}
Let the projections $P_{\alpha}^L$, $P_{\beta}^R$ and $P_{\alpha\beta}$ be defined as in the proof of Proposition~\ref{oddmult}, and recall that $P_{\alpha\beta} = P_{\alpha}^L P_{\beta}^R$. By Equation~\ref{realintertwine}, it follows that for all $\alpha \in \spec{\alg{A}}$,$J P_{\alpha}^L = P_{\alpha}^R J$ and $J P_{\alpha}^R = P_{\alpha}^L J$, and hence that for all $\alpha$, $\beta \in \spec{\alg{A}}$, $J P_{\alpha\beta} = J P_\alpha^L P_\beta^R = P_\alpha^R P_\beta^L J = P_{\beta\alpha} J$. Thus, the only non-zero components of $J$ are the anti-unitaries $J_{\alpha\beta}^{\beta\alpha} : \hs{H}_{\alpha\beta} \to \hs{H}_{\beta\alpha}$ which satisfy $J_{\beta\alpha}^{\alpha\beta} = (J_{\alpha\beta}^{\beta\alpha})^*$; this, in turn, implies that $m$ is indeed symmetric.
\end{proof}

Next, we show that for every $m \in \Sym_S(\ZO)$, not only does $\hs{H}_m$ admit a real structure of $KO$-dimension $1$ or $7 \bmod 8$, but it is also unique up to unitary equivalence.

\begin{lemma}[Krajewski~\cite{Kraj98}*{\S3.2}, Paschke--Sitarz~\cite{PS98}*{Lemma 5}]\label{real17b}
 Let $m \in \Sym_S(\ZO)$. Then, up to unitary equivalence, $J_m := X_m$ is the unique real structure on $\hs{H}_m$ of $KO$-dimension $1$ or $7 \bmod 8$.
\end{lemma}

\begin{proof}
 First, $X_m$ is indeed by construction a real structure on $\hs{H}_m$ of $KO$-dimension $1$ or $7 \bmod 8$.
 
 Now, let $J$ be another real structure on $\hs{H}_m$ of $KO$-dimension $1$ or $7 \bmod 8$. Define a unitary $K$ on $\hs{H}$ by $K = J X_m$; thus, $J = K X_m$. Since the intertwining condition of Equation~\ref{realintertwine} applies to both $J$ and $X_m$, we have, in fact, that $K \in \lrU_\alg{A}(\hs{H}_m)$, and hence
\[
 K = \bigoplus_{\alpha,\beta \in \spec{\alg{A}}} 1_{n_\alpha} \otimes K_{\alpha\beta} \otimes 1_{n_\beta},
\]
for $K_{\alpha\beta} \in \U(m_{\alpha\beta})$. In particular, since $K^* = X_m J = X_m K X_m$, we have that $K_{\beta\alpha} = K_{\alpha\beta}^T$.

Let $(\alpha,\beta) \in \supp(m)$, and suppose that $\alpha < \beta$. Let $K_{\alpha\beta} = V_{\alpha\beta} \tilde{K}_{\alpha\beta} V_{\alpha\beta}^*$ be a unitary diagonalisation of $K_{\alpha\beta}$, and let $L_{\alpha\beta}$ be a diagonal square root of $\tilde{K}_{\alpha\beta}$. Then$ K_{\alpha\beta} = V_{\alpha\beta} L_{\alpha\beta} L_{\alpha\beta} V_{\alpha\beta}^* = (V_{\alpha\beta}L_{\alpha\beta})(\overline{V_{\alpha\beta}} L_{\alpha\beta})^T$, and hence $K_{\beta\alpha} = (\overline{V_{\alpha\beta}} L_{\alpha\beta})(V_{\alpha\beta}L_{\alpha\beta})^T$. If, instead, $\alpha = \beta$, then $K_{\alpha\alpha}$ is unitary and complex symmetric, so that there exists a unitary $W_{\alpha\alpha}$ such that $K_{\alpha\alpha} = W_{\alpha\alpha}W_{\alpha\alpha}^T$. We can now define a unitary $U \in \lrU_\alg{A}(\hs{H}_m)$ by
\[
 U = \bigoplus_{\alpha,\beta \in \spec{\alg{A}}} 1_{n_\alpha} \otimes U_{\alpha\beta} \otimes 1_{n_\beta},
\]
where $U_{\alpha\beta} = 0$ if $m_{\alpha\beta} = 0$, and for $(\alpha,\beta) \in \supp(m)$,
\[
 U_{\alpha\beta} = 
\begin{cases}
 V_{\alpha\beta}L_{\alpha\beta}, &\text{if $\alpha < \beta$,}\\
 \overline{V_{\beta\alpha}}L_{\beta\alpha}, &\text{if $\alpha > \beta$,}\\
 W_{\alpha\alpha}, &\text{if $\alpha = \beta$.}
\end{cases}
\]
Then, by construction, $K = U  X_m U^* X_m$, and hence, $J = U X_m U^*$, so that $U$ is the required unitary equivalence between $(\hs{H}_m,X_m)$ and $(\hs{H}_m,J)$.
\end{proof}

We can now give our characterisation of real bimodules of $KO$-dimension $1$ or $7 \bmod 8$:

\begin{proposition}[Krajewski~\cite{Kraj98}*{\S3.2}]\label{real17mult}
Let $n = 1$ or $7 \bmod 8$. Then the map $\iota_n : \Bimod(\alg{A},n) \to \Bimod(\alg{A})$ defined by $\iota_n : [(\hs{H},J)] \mapsto [\hs{H}]$ is injective, and
\begin{equation}
 (\mult \circ \iota_n)(\Bimod(\alg{A},n)) = \Sym_S(\ZO).
\end{equation}
\end{proposition}

\begin{proof}
First, since a unitary equivalence of real $\alg{A}$-bimodules of $KO$-di\-men\-sion $n \bmod 8$ is, in particular, a unitary equivalence of odd $\alg{A}$-bimodules, the map $\iota_n$ is well defined.

Next, let $(\hs{H},J)$ and $(\hs{H}^\prime,J^\prime)$ be real $\alg{A}$-bimodules of $KO$-dimension $n \bmod 8$, and suppose that $\hs{H}$ and $\hs{H}^\prime$ are unitarily equivalent as bimodules; let $U \in \lrU_\alg{A}(\hs{H}^\prime,\hs{H})$. Now, if $m$ is the multiplicity matrix of $\hs{H}$, then $\hs{H}$ and $\hs{H}_m$ are unitarily equivalent, so let $V \in \lrU_\alg{A}(\hs{H},\hs{H}_m)$. Then $V J V^*$ and $V U J^\prime U^* V^*$ are both real structures of $KO$-dimension $n \bmod 8$, so by Lemma~\ref{real17b}, they are both unitarily equivalent to $J_m$. This implies that $J$ and $U J^\prime U^*$ are unitarily equivalent as real structures on $\hs{H}$, and hence that $(\hs{H},J)$ and $(\hs{H}^\prime,J^\prime)$ are unitarily equivalent. Thus, $\iota_n$ is injective.

Finally, Lemma~\ref{real17a} implies that $(\mult \circ \iota_n)(\Bimod(\alg{A},n)) \subseteq \Sym_S(\ZO)$, while Lemma~\ref{real17b} implies the reverse inclusion.
\end{proof}

Thus, without any loss of generality, a real bimodule $\hs{H}$ of $KO$-dimension $1$ or $7 \bmod 8$ with multiplicity matrix $m$ can be assumed to be simply $(\hs{H}_m,J_m)$.

One following characterisation of $\lrU_\alg{A}(\hs{H},J)$ now follows by direct calculation:

\begin{proposition}\label{real17unitary}
Let $(\hs{H},J)$ be a real $\alg{A}$-bimodule of $KO$-dimension $1$ or $7 \bmod 8$ with multiplicity matrix $m$. Then
\begin{equation}\begin{split}
 \comp(\lrU_\alg{A}(\hs{H},J)) &= \{(1_{n_\alpha} \otimes U_{\alpha\beta} \otimes 1_{n_\beta})_{\alpha,\beta \in \spec{\alg{A}}} \in \comp(\lrU_\alg{A}(\hs{H})) \mid U_{\beta\alpha} = \overline{U_{\alpha\beta}}\}\\ &\cong \prod_{\alpha \in \spec{\alg{A}}} \biggl( \Orth(m_{\alpha\alpha}) \times \prod_{\substack{\beta \in \spec{\alg{A}}\\ \beta > \alpha}} \U(m_{\alpha\beta}) \biggr).
\end{split}\end{equation}
\end{proposition}

\subsubsection{$KO$-dimension $3$ or $5 \bmod 8$}

Let us now turn to real bimodules of $KO$-dimension $3$ or $5 \bmod 8$. We begin with the relevant analogue of Lemma~\ref{real17a}.

\begin{lemma}\label{real35a}
Let $(\hs{H},J)$ be a real $\alg{A}$-bimodule of $KO$-dimension $3$ or $5 \bmod 8$ with multiplicity matrix $m$. Then $m$ is symmetric with even diagonal entries, and the only non-zero components of $J$ are of the form $J_{\alpha\beta}^{\beta\alpha}$ for $\alpha$, $\beta \in \spec{\alg{A}}$, which are anti-unitaries $\hs{H}_{\alpha\beta} \to \hs{H}_{\beta\alpha}$ satisfying the relations $J_{\beta\alpha}^{\alpha\beta} = -(J_{\alpha\beta}^{\beta\alpha})^*$.
\end{lemma}

\begin{proof}
The proof follows just as for Lemma~\ref{real17a}, except that the equation $J^2 = -1$ forces the relations $J_{\beta\alpha}^{\alpha\beta} = -(J_{\alpha\beta}^{\beta\alpha})^*$, which imply, in particular, that for each $\alpha \in \spec{\alg{A}}$, $(J_{\alpha\alpha}^{\alpha\alpha})^2 = -1$, so that $m_{\alpha\alpha}$ must be even.
\end{proof}

Let us denote by $\Sym_S^0(\ZO)$ the set of all matrices in $\Sym_S(\ZO)$ with even diagonal entries. For $n = 2k$, let
\[
 \Omega_n =
 \begin{pmatrix}
 0 & -1_k\\
 1_k & 0
 \end{pmatrix}.
\]

\begin{lemma}~\label{real35b}
Let $m \in \Sym_S^0(\ZO)$. Define an antiunitary $J_m$ on $\hs{H}_m$ by
\[
 (J_m)_{\alpha\beta}^{\gamma\delta} =
 \begin{cases}
  (X_m)_{\alpha\beta}^{\beta\alpha} &\text{if $(\gamma,\delta)=(\beta,\alpha)$ and $\alpha < \beta$,}\\
  -(X_m)_{\alpha\beta}^{\beta\alpha} &\text{if $(\gamma,\delta)=(\beta,\alpha)$ and $\alpha > \beta$,}\\
  \Omega_{m_{\alpha\alpha}}(X_m)_{\alpha\alpha}^{\alpha\alpha} &\text{if $\alpha = \beta = \gamma = \delta$,}\\
  0 &\text{otherwise.}
 \end{cases}
\]
Then, up to unitary equivalence, $J_m$ is the unique real structure on $\hs{H}_m$ of $KO$-dimension $3$ or $5 \bmod 8$.
\end{lemma}

\begin{proof}
The proof follows that of Lemma~\ref{real17b}, except we now have that
$K_{\alpha\alpha}^T = \Omega_{m_{\alpha\alpha}} K_{\alpha\alpha} \Omega_{m_{\alpha\alpha}}^T$ instead of $K_{\alpha\alpha}^T = K_{\alpha\alpha}$; each $K_{\alpha\alpha} \Omega_{m_{\alpha\alpha}}$ is therefore unitary and complex skew-symmetric, so that we choose $W_{\alpha\alpha}$ unitary such that
\[
 K_{\alpha\alpha}\Omega_{m_{\alpha\alpha}} = W_{\alpha\alpha}\Omega_{m_{\alpha\alpha}}W_{\alpha\alpha}^T,
\]
or equivalently, $K_{\alpha\alpha} = W_{\alpha\alpha}\Omega_{m_{\alpha\alpha}}W_{\alpha\alpha}^T\Omega_{m_{\alpha\alpha}}^T$. One can then construct the unitary equivalence $U$ between $(\hs{H}_m,J)$ and $(\hs{H},J_m)$ as before.
\end{proof}

Much as in the analogous case of $KO$-dimension $1$ or $7 \bmod 8$, Lemmas~\ref{real35a} and~\ref{real35b} together imply the following characterisation of real bimodules of $KO$-dimension $3$ or $5 \bmod 8$:

\begin{proposition}\label{real35mult}
Let $n=3$ or $5 \bmod 8$. Then the map $\iota_n : \Bimod(\alg{A},n) \to \Bimod(\alg{A})$ defined by $\iota_n : [(\hs{H},J)] \mapsto [\hs{H}]$ is injective, and
\begin{equation}
 (\mult \circ \iota_n)(\Bimod(\alg{A},n)) = \Sym_S^0(\ZO).
\end{equation}
\end{proposition}

Finally, these results immediately imply the following description of $\lrU_\alg{A}(\hs{H},J)$:

\begin{proposition}\label{real35unitary}
Let $(\hs{H},J)$ be a real $\alg{A}$-bimodule of $KO$-dimension $3$ or $5 \bmod 8$ with multiplicity matrix $m$. Then
\begin{equation}\begin{split}
 \comp(\lrU_\alg{A}(\hs{H},J)) &= \left\{(1_{n_\alpha} \otimes U_{\alpha\beta} \otimes 1_{n_\beta})_{\alpha,\beta \in \spec{\alg{A}}} \in \comp(\lrU_\alg{A}(\hs{H})) \mid \substack{U_{\alpha\alpha} \in \Sp(m_{\alpha\alpha}), \\ U_{\beta\alpha} = \overline{U_{\alpha\beta}}, \: \alpha \neq \beta}\right\}\\
 &\cong \prod_{\alpha \in \spec{\alg{A}}} \biggl(\Sp(m_{\alpha\alpha}) \times \prod_{\substack{\beta \in \spec{\alg{A}}\\ \beta > \alpha}} \U(m_{\alpha\beta})\biggr).
\end{split}\end{equation}
\end{proposition}

\subsection{Real bimodules of even $KO$-dimension}

We now come to the case of even $KO$-dimension. Before continuing, note that for $(\hs{H},\gamma,J)$ a real bimodule of even $KO$-dimension,
\[
 \forall p,q \in KO_0(\alg{A}), \: \form{q,p} = \epspp \form{p,q},
\]
as a direct result of the relation $J \gamma = \epspp \gamma J$; this is then equivalent to the condition
\begin{equation}
 \cap = \epspp \cap^T,
\end{equation}
where $\cap$ is the matrix of the intersection form. Thus, for $KO$-dimension $0$ or $4 \bmod 8$, the intersection form is symmetric, whilst for $KO$-dimension $2$ or $6 \bmod 8$, it is anti-symmetric. It then follows, in particular, that a real $\alg{A}$-bimodule of $KO$-dimension $2$ or $6 \bmod 8$ satisfies Poincar{\'e} duality only if $\alg{A}$ has an even number of direct summands in its Wedderburn decomposition, as an anti-symmetric $k \times k$ matrix for $k$ odd is necessarily degenerate.

\subsubsection{$KO$-dimension $0$ or $4 \bmod 8$}

We begin with the case where $\epspp = 1$ and hence $[\gamma,J]=0$, \ie of $KO$-dimension $0$ or $4 \bmod 8$.

Let $(\hs{H},\gamma,J)$ be a real $\alg{A}$-bimodule of $KO$-dimension $n \bmod 8$, for $n = 0$ or $4$; let the mutually orthogonal projections $P^\even$ and $P^\odd$ on $\hs{H}$ be defined as before. Then, since $[J,\gamma]=0$, we have that $J = J^\even \oplus J^\odd$, where $J^\even = P^\even J P^\even$ and $J^\odd = P^\odd J P^\odd$. One can then check that $(\hs{H}^\even,J^\even)$ and $(\hs{H}^\odd,J^\odd)$ are real $\alg{A}$-bimodules of $KO$-dimension $1$ or $7 \bmod 8$ if $n = 0$, and $3$ or $5 \bmod 8$ if $n=4$. On the other hand, given $(\hs{H}^\even,J^\even)$ and $(\hs{H}^\odd,J^\odd)$, one can immediately reconstruct $(\hs{H},\gamma,J)$ by setting $\gamma = 1_{\hs{H}^\even} \oplus (-1_{\hs{H}^\odd})$ and $J = J^\even \oplus J^\odd$. Thus we have proved the following analogue of Proposition~\ref{evensplit}:

\begin{proposition}~\label{real04split}
 Let $\alg{A}$ be a real \Cstar-algebra. Let $k_0$ denote $1$ or $7 \bmod 8$, and let $k_4$ denote $3$ or $5 \bmod 8$. Then for $n = 0, 4 \bmod 8$, the map
 \[
  C_n : \Bimod(\alg{A},n) \to \Bimod(\alg{A},k_n) \times \Bimod(\alg{A},k_n)
 \]
 given by $C_n([(\hs{H},\gamma,J)]) := ([(\hs{H}^\even,J^\even)],[(\hs{H}^\odd,J^\odd)])$ is an isomorphism of monoids.
\end{proposition}

One can then apply this decomposition to the group $\lrU_\alg{A}(\hs{H},\gamma,J)$ to find:

\begin{corollary}\label{real04unitary}
 Let $(\hs{H},\gamma,J)$ be a real $\alg{A}$-bimodule of $KO$-dimension $0$ or $4 \bmod 8$. Then
 \begin{equation}
 \lrU_\alg{A}(\hs{H},\gamma,J) = \lrU_\alg{A}(\hs{H}^\even,J^\even) \oplus \lrU_\alg{A}(\hs{H}^\odd,J^\odd).
 \end{equation}
\end{corollary}

Combining Proposition~\ref{real04split} with our earlier characterisations of real bimodules of odd $KO$-dimension, we immediately obtain the following:

\begin{proposition}\label{real04mult}
 Let $n=0$ or $4 \bmod 8$. Then the map $\iota_n : \Bimod(\alg{A},n) \to \Bimod^\even(\alg{A})$ defined by
 $[(\hs{H},\gamma,J)] \mapsto ([(\hs{H},\gamma)])$ is injective, and
 \[
  (\mult^\even \circ \iota_n)(\Bimod(\alg{A},n)) =
  \begin{cases}
   \Sym_S(\ZO) \times \Sym_S(\ZO) &\text{if $n = 0 \bmod 8$,}\\
   \Sym_S^0(\ZO) \times \Sym_S^0(\ZO) &\text{if $n = 4 \bmod 8$.}
  \end{cases}
 \]
\end{proposition}

In particular,
\[
 \Bimod_q(\alg{A},n) := \iota_n^{-1}(\Bimod^\even_q(\alg{A}))
\]
is thus the set of all equivalence classes of quasi-orientable real $\alg{A}$-bimodules of $KO$-dimension $n \bmod 8$; the last Proposition then implies the following:

\begin{corollary}
 Let $n=0$ or $4 \bmod 8$. Then
\begin{equation}
 (\mult_q \circ \iota_n)(\Bimod_q(\alg{A},n)) = \Sym_S(\ring{Z}).
\end{equation}
\end{corollary}

\subsubsection{$KO$-dimension $2$ or $6 \bmod 8$}

Finally, let us consider the remaining case where $\epspp = -1$ and hence $\{\gamma,J\}=0$, \ie of $KO$-dimensions $2$ and $6 \bmod 8$.

Let $(\hs{H},\gamma,J)$ be a real $\alg{A}$-bimodule of $KO$-dimension $n \bmod 8$ for $n = 2$ or $6$. Since $\{J,\gamma\}$, we have that
\[
 J =
 \begin{pmatrix}
  0 & \eps \tilde{J}^*\\
  \tilde{J} & 0
 \end{pmatrix},
\]
where $\tilde{J} := P^\odd J P^\even$ is an antiunitary $\hs{H}^\even \to \hs{H}^\odd$, so that for $a \in \alg{A}$,
\[
 \lambda^\odd(a) = \tilde{J} \rho^\even(a^*) \tilde{J}^*, \quad \rho^\odd(a) = \tilde{J} \lambda^\even(a^*) \tilde{J}^*.
\]
It then follows, in particular, that $\mult[\hs{H}^\odd] = \mult[\hs{H}^\even]^T$. 

Now, let $J^\prime$ be another real structure on $(\hs{H},\gamma)$ of $KO$-dimension $n \bmod 8$, and let $\tilde{J^\prime} = P^\odd J^\prime P^\even$. Define $K \in \lrU_\alg{A}(\hs{H},\gamma)$ by $K = 1_{\hs{H}^\even} \oplus (\tilde{J^\prime} \tilde{J}^*)$. Then, by construction, $J^\prime = KJK^*$, \ie $K$ is a unitary equivalence of real structures between $J$ and $J^\prime$. Thus, real structures of $KO$-dimension $2$ or $6 \bmod 8$ are unique. As a result, we have proved the following analogue of Proposition~\ref{s0reduction}:

\begin{proposition}~\label{real26reduction}
Let $\alg{A}$ be a real \Cstar-algebra, and let $n = 2$ or $6 \bmod 8$. Then the map
\[
 C_n : \Bimod(\alg{A},n) \to \Bimod(\alg{A})
\]
given by $C_n([(\hs{H},\gamma,J)]) := ([\hs{H}^\even])$ is an isomorphism of monoids.
\end{proposition}

Again, as an immediate consequence, we obtain the following characterisation of $\lrU_\alg{A}(\hs{H},\gamma,J)$:

\begin{corollary}\label{real26unitary}
Let $(\hs{H},\gamma,J)$ be a real $\alg{A}$-bimodule of $KO$-dimension $2$ or $6 \bmod 8$. Then
\begin{equation}\begin{split}
 \lrU_\alg{A}(\hs{H},\gamma,J) &= \{ U^\even \oplus U^\odd \in \lrU_\alg{A}(\hs{H}^\even) \oplus \lrU_\alg{A}(\hs{H}^\odd) \mid U^\odd = \tilde{J}U^\even\tilde{J}^*\}\\ &\cong \lrU_\alg{A}(\hs{H}^\even).
\end{split}\end{equation}
\end{corollary}

Finally, one can combine Proposition~\ref{real26reduction} with our observation concerning the uniqueness up to unitary equivalence of real structures of $KO$-dimension $2$ or $6 \bmod 8$ and earlier results on multiplicity matrices to obtain the following characterisation:

\begin{proposition}
 Let $n=2$ or $6 \bmod 8$. Then the map $\iota_n : \Bimod(\alg{A},n) \to \Bimod^\even(\alg{A})$ defined by $[(\hs{H},\gamma,J)] \mapsto ([\hs{H},\gamma])$ is injective, and
\begin{equation}\begin{split}
 (\mult^\even \circ \iota_n)(\Bimod(\alg{A},n)) &= \{(m^\even,m^\odd) \in M_S(\ZO)^2 \mid m^\odd = (m^\even)^T\}\\ &\cong M_S(\ZO).
\end{split}\end{equation}
\end{proposition}

Once more, it follows that
\[
 \Bimod_q(\alg{A},n) := \iota_n^{-1}(\Bimod^\even_q(\alg{A})),
\]
is the set of all equivalence classes of quasi-orientable real $\alg{A}$-bimodules of $KO$-dimension $n \bmod 8$, for which we can again obtain a characterisation in terms of signed multiplicity matrices:

\begin{corollary}
 Let $n=2$ or $6 \bmod 8$. Then
\begin{equation}
 (\mult_q \circ \iota_n)(\Bimod_q(\alg{A},n)) = \{m \in M_S(\ring{Z}) \mid m^T = -m\}.
\end{equation}
\end{corollary}

\subsubsection{$S^0$-real bimodules of even $KO$-dimension}

Let us now characterise quasi-orientability, orientability and Poincar{\'e} duality for an even $KO$-dimensional $S^0$-real $\alg{A}$-bimodule $(\hs{H},\gamma,J,\epsilon)$ by means of suitable conditions on $(\hs{H}_i,\gamma_i)$ expressible entirely in terms of the pair of multiplicity matrices of $(\hs{H}_i,\gamma_i)$

We begin by considering quasi-orientability:

\begin{proposition}\label{s0quasiorient}
  Let $(\hs{H},\gamma,J,\epsilon)$ be an $S^0$-real $\alg{A}$-bimodule of even $KO$-dimension $n \bmod 8$. Then $(\hs{H},\gamma)$ is quasi-orientable if and only if $(\hs{H}_i,\gamma_i)$ is quasi-orientable and
\[
 \begin{cases}
  \supp(m^\even_i) \cap \supp((m^\odd_i)^T) = \emptyset &\text{if $n=0$, $4$,}\\
  \supp(m^\even_i) \cap \supp((m^\even_i)^T) = \supp(m^\odd_i) \cap \supp((m^\odd_i)^T) = \emptyset &\text{if $n=2$, $6$,}
 \end{cases}
\]
for $(m^\even_i,m^\odd_i)$ the multiplicity matrices of $(\hs{H}_i,\gamma_i)$, in which case, if $\mu$ and $\mu_i = m^\even_i - m^\odd_i$ are the signed multiplicity matrices of $(\hs{H},\gamma)$ and $(\hs{H}_i,\gamma_i)$, respectively, then
\begin{equation}
 \mu = \mu_i + \epspp \mu_i^T.
\end{equation}
\end{proposition}

\begin{proof}
 First, let $(m^\even,m^\odd)$ and $(m^\even_i,m^\odd_i)$ denote the pairs of multipicity matrices of $(\hs{H},\gamma)$ and $(\hs{H}_i,\gamma_i)$, respectively. It then follows that 
\begin{align*}
 m^\even &=
 \begin{cases}
  m^\even_i + (m^\even_i)^T &\text{if $n = 0$, $4$,}\\
  m^\even_i + (m^\odd_i)^T &\text{if $n = 2$, $6$;}
 \end{cases}\\
 m^\odd &=
 \begin{cases}
  m^\odd_i + (m^\odd_i)^T &\text{if $n = 0$, $4$,}\\
  m^\odd_i + (m^\even_i)^T &\text{if $n=2$, $6$.}
 \end{cases}
\end{align*}

Thus $\supp(m^\even) = \supp(m^\even_i) \cup S^\even, \quad \supp(m^\odd) = \supp(m^\odd_i) \cup S^\odd$, where
\[
 S^\even =
\begin{cases}
 \supp((m^\even_i)^T) &\text{if $n = 0$, $4$,}\\
 \supp((m^\odd_i)^T) &\text{if $n = 2$, $6$;}
\end{cases} \quad
 S^\odd =
\begin{cases}
 \supp((m^\odd_i)^T) &\text{if $n = 0$, $4$,}\\
 \supp((m^\even_i)^T) &\text{if $n = 2$, $6$.}
\end{cases}
\]

Then,
\begin{multline*}
 \supp(m^\even) \cap \supp(m^\odd) = (\supp(m^\even_i) \cap \supp(m^\odd_i)) \cup (S^\even \cap \supp(m^\odd_i)) \\ \cup (\supp(m^\even_i) \cap S^\odd) \cup (S^\even \cap S^\odd),
\end{multline*}
so that $(\hs{H},\gamma)$ is quasi-orientable if and only if $(\hs{H}_i,\gamma_i)$ is quasi-orientable and
\[
 (S^\even \cap \supp(m^\odd_i)) \cup (\supp(m^\even_i) \cap S^\odd) = \emptyset,
\]
as required. 

Finally, if $\mu = m^\even - m^\odd$ and $\mu_i = m^\even_i - m^\odd_i$ are the signed multiplicity matrices of $(\hs{H},\gamma)$ and $(\hs{H}_i,\gamma_i)$, respectively, then the relations amongst $m^\even$, $m^\odd$, $m^\even_i$, and $m^\odd_i$ given at the beginning immediately yield the equation $\mu = \mu_i + \epspp \mu_i^T$.
\end{proof}

Let us now turn to orientability:

\begin{proposition}
 Let $(\hs{H},\gamma,J,\epsilon)$ be a quasi-orientable $S^0$-real $\alg{A}$-bimodule of even $KO$-dimension $n \bmod 8$. Then $(\hs{H},\gamma)$ is orientable if and only if $(\hs{H}_i,\gamma_i)$ is orientable and, if $n = 2$ or $6 \bmod 8$, for all $j \in \{1,\dotsc,N\}$ such that $\field{K}_j = \field{C}$, 
\begin{equation}
 (\mu_i)_{\rep{n}_j \crep{n}_j} = (\mu_i)_{\crep{n}_j \rep{n}_j},
\end{equation}
where $\mu_i$ is the signed multiplicity matrix of $(\hs{H}_i,\gamma_i)$.
\end{proposition}

\begin{proof}
 Let $\mu$ be the signed multiplicity matrix of $(\hs{H},\gamma)$. Propositions~\ref{s0orient} and~\ref{converseorientable} together imply that $(\hs{H},\gamma,J,\epsilon)$ is orientable if and only if
\[
 \gamma_i = \sum_{k,l=1}^N \lambda_i(\sgn(\widehat{\mu}_{kl})e_k)\rho_i(e_l) = \epspp \sum_{k,l=1}^N \lambda_i(e_l)\rho_i(\sgn(\widehat{\mu}_{kl})e_k),
\]
and by considering individual components $(\gamma_i)_{\alpha\beta}$, one can easily check that this in turn holds if and only if $(\hs{H}_i,\gamma_i)$ is orientable and for all $k \in \{1,\dotsc,N\}$,
\[
 \sgn(\widehat{\mu}_{kk}) = \epspp \sgn(\widehat{\mu}_{kk}).
\]

This last condition is trivial when $\epspp = 1$, \ie when $n = 0$ or $4 \bmod 8$, so let us suppose instead that $n = 2$ or $6 \bmod 8$, so that $\epspp = -1$. If $(\hs{H},\gamma)$ is orientable, then, by the above discussion, $(\hs{H}_i,\gamma_i)$ is orientable and the diagonal entries of $\widehat{\mu}$ vanish, which in turn implies by Proposition~\ref{converseorientable} that for each $l \in \{1,\dotsc,N\}$ and all $\alpha$, $\beta \in \spec{M_{k_l}(\field{K}_l)}$, $\mu_{\alpha\beta} = 0$. By antisymmetry of $\mu$, this is equivalent to having, for all $l \in \{1,\dotsc,N\}$ such that $\field{K}_l = \field{C}$, $\mu_{\rep{n}_l \crep{n}_l} = 0$, or equivalently,
\[
 (\mu_i)_{\rep{n}_j \crep{n}_j} = (\mu_i)_{\crep{n}_j \rep{n}_j},
\]
where $\mu_i$ is the signed multiplicity matrix of $(\hs{H}_i,\gamma_i)$. On the other hand, if $(\hs{H}_i,\gamma_i)$ is orientable and this condition on $\mu_i$ holds, then $\mu$ certainly satisfies the above condition, so that $(\hs{H},\gamma)$ is indeed orientable.
\end{proof}

Finally, let us consider Poincar{\'e} duality.

\begin{proposition}
 Let $(\hs{H},\gamma,J,\epsilon)$ be an $S^0$-real $\alg{A}$-bimodule of even $KO$-dimension $n \bmod 8$, let $(m^\even_i,m^\odd_i)$ denote the multiplicity matrices of $(\hs{H}_i,\gamma_i)$, and let $\cap$ denote the matrix of the intersection form of $(\hs{H},\gamma)$. Finally, let $\mu_i = m^\even_i - m^\odd_i$. Then
\begin{equation}
 \cap_{kl} = \tau_k \tau_l (\widehat{\mu_i} + \epspp \widehat{\mu_i}^T)_{kl},
\end{equation}
so that $(\hs{H},\gamma)$ satisfies Poincar{\'e} duality if and only if $\widehat{\mu_i} + \epspp \widehat{\mu_i}^T$ is non-degenerate.
\end{proposition}

\begin{proof}
 By Proposition~\ref{s0poincare}, $\cap = \cap_i + \epspp \cap_i^T$ for $\cap_i$ the matrix of the intersection form of $(\hs{H}_i,\gamma_i)$, which, together with Proposition~\ref{intform}, yields the desired result.
\end{proof}

\subsection{Bimodules in the Chamseddine--Connes--Marcolli model}

To illustrate the structure theory outlined thus far, let us apply it to the construction of the finite spectral triple of the NCG Standard Model given by Chamseddine, Connes and Marcolli~\cite{CCM07}*{\S\S 2.1, 2.2, 2.4} (cf.~also~\cite{CM08}*{\S 1.13}).

Let $\alg{A}_{LR} = \field{C} \oplus \field{H}_L \oplus \field{H}_R \oplus M_3(\field{C})$, where the labels $L$ and $R$ serve to distinguish the two copies of $\hs{H}$; we can therefore write $\spec{\alg{A}_{LR}} = \{\rep{1},\crep{1},\rep{2}_L,\rep{2}_R,\rep{3},\crep{3}\}$ without ambiguity. Now, let $(\hs{M}_F,\gamma_F,J_F)$ be the orientable real $\alg{A}_{LR}$-bimodule of $KO$-dimension $6 \bmod 8$ with signed multiplicity matrix
\[
 \mu =
\begin{pmatrix}
 0 & 0 & -1 & 1 & 0 & 0\\
 0 & 0 & 0 & 0 & 0 & 0\\
 1 & 0 & 0 & 0 & 1 & 0\\
 -1 & 0 & 0 & 0 & -1 & 0\\
 0 & 0 & -1 & 1 & 0 & 0\\
 0 & 0 & 0 & 0 & 0 & 0
\end{pmatrix}.
\]
This bimodule is, in fact, an $S^0$-real bimodule for $\epsilon_F = \lambda(-1,1,1,-1)$; $\hs{E} = (\hs{M}_F)_i$ is then the orientable even $\alg{A}_{LR}$-bimodule with signed multiplicity matrix
\[
 \mu_\hs{E} = 
\begin{pmatrix}
 0 & 0 & 0 & 0 & 0 & 0\\
 0 & 0 & 0 & 0 & 0 & 0\\
 1 & 0 & 0 & 0 & 1 & 0\\
 -1 & 0 & 0 & 0 & -1 & 0\\
 0 & 0 & 0 & 0 & 0 & 0\\
 0 & 0 & 0 & 0 & 0 & 0
\end{pmatrix}.
\]
Note, however, that neither $\hs{M}_F$ nor $\hs{E}$ satisfies Poincar{\'e} duality, as
\[
 \hat{\mu} =
\begin{pmatrix}
 0 & -1 & 1 & 0\\
 1 & 0 & 0 & 1\\
 -1 & 0 & 0 & -1\\
 0 & -1 & 1 & 0
\end{pmatrix}, \quad
 \widehat{\mu_\hs{E}} =
\begin{pmatrix}
 0 & 0 & 0 & 0\\
 1 & 0 & 0 & 1\\
 -1 & 0 & 0 & -1\\
 0 & 0 & 0 & 0
\end{pmatrix}
\]
are both clearly degenerate; the intersection forms of $\hs{M}_F$ and $\hs{E}$ are given by the matrices $\cap = 2 \hat{\mu}$ and $\cap_\hs{E} = 2 \widehat{\mu_\hs{E}}$, respectively. 

In order to introduce $N$ generations of fermions and anti-fermions, one now considers the real $\alg{A}$-bimodule $\hs{H}_F := (\hs{M}_F)^{\oplus N}$; by abuse of notation, $\gamma_F$, $J_F$ and $\epsilon_F$ now also denote the relevant structure operators on $\hs{H}_F$. In terms of multiplicity matrices and intersection forms, the sole difference from our discussion of $\hs{M}_F$ is that all matrices are now multiplied by $N$.

Now, let $\alg{A}_F = \field{C} \oplus \hs{H} \oplus M_3(\field{C})$, which we consider as a subalgebra of $\alg{A}_{LR}$ by means of the embedding
\[
 (\zeta,q,m) \mapsto \left(\zeta,q,\begin{pmatrix}\lambda & 0\\0 & \overline{\lambda}\end{pmatrix},m\right);
\]
just as we could for $\alg{A}_{LR}$, we can write $\spec{\alg{A}_F} = \{\rep{1},\crep{1},\rep{2},\rep{3},\crep{3}\}$ without ambiguity. We can therefore view $\hs{H}_F$ as a real $\alg{A}_F$-bimodule of $KO$-dimension $6 \bmod 8$, whose pair of multiplicity matrices $(m^\even,m^\odd)$ is then given by
\[
 m^\even = N
\begin{pmatrix}
1 & 1 & 0 & 0 & 0\\
0 & 0 & 0 & 0 & 0\\
1 & 0 & 0 & 1 & 0\\
1 & 1 & 0 & 0 & 0\\
0 & 0 & 0 & 0 & 0
\end{pmatrix}, \quad
m^\odd = N
\begin{pmatrix}
1 & 0 & 1 & 1 & 0\\
1 & 0 & 0 & 1 & 0\\
0 & 0 & 0 & 0 & 0\\
0 & 0 & 1 & 0 & 0\\
0 & 0 & 0 & 0 & 0
\end{pmatrix};
\]
the essential observation is that the irreducible representation $\rep{2}_R$ of $\alg{A}_{LR}$ corresponds to the representation $\rep{1}\oplus\crep{1}$ of $\alg{A}_F$, whilst $\rep{2}_L$, $\rep{3}$ and $\crep{3}$ correspond to $\rep{2}$, $\rep{3}$ and $\crep{3}$, respectively. 

Note that $\hs{H}_F$ now fails even to be quasi-orientable let alone orientable, with the sub-bimodule $(\hs{H}_F)_{\rep{1}\rep{1}}$ providing the obstruction, and even if we were to restore quasi-orientability by setting $(\hs{H}_F)_{\rep{1}\rep{1}} = 0$, $(\hs{H}_F)_{\rep{1}\crep{1}}$ and $(\hs{H}_F)_{\crep{1}\rep{1}}$ would still present an obstruction to orientability by Proposition~\ref{converseorientable}. Note also that $\hs{H}_F$ must necessarily fail to satisfy Poincar{\'e} duality, as the matrix $\cap_F$ of its intersection form is a $3 \times 3$ anti-symmetric matrix, and thus \latin{a priori} degenerate. Let us nonetheless compute $\cap_F$:
\[
 \widehat{m^\even} - \widehat{m^\odd} = 
 N \begin{pmatrix}
 2 & 0 & 0\\
 1 & 0 & 1\\
 2 & 0 & 0
\end{pmatrix} -
N \begin{pmatrix}
 2 & 1 & 2\\
 0 & 0 & 0\\
 0 & 1 & 0   
\end{pmatrix} =
N \begin{pmatrix}
 0 & -1 & -2\\
 1 & 0 & 1\\
 2 & -1 & 0
\end{pmatrix},
\]
and hence, by Proposition~\ref{intform},
\[
 \cap_F =
2N \begin{pmatrix}
 0 & -1 & -1\\
 1 & 0 & 1\\
 1 & -1 & 0
\end{pmatrix}.
\]

Finally, let us consider the $S^0$-real structure on $\hs{H}_F$ the $\alg{A}_F$-bimodule, inherited from $\hs{H}_F$ as an $\alg{A}_{LR}$-bimodule; we now denote $(\hs{H}_F)_i$ by $\hs{H}_f$. One still has that $\hs{H}_f = \hs{E}^{\oplus N}$, which is still orientable and thus specified by the signed multiplicity matrix
\[
 \mu_{f} = N
\begin{pmatrix}
 -1 & 0 & 0 & -1 & 0\\
 -1 & 0 & 0 & -1 & 0\\
 1 & 0 & 0 & 1 & 0\\
 0 & 0 & 0 & 0 & 0\\
 0 & 0 & 0 & 0 & 0
\end{pmatrix};
\]
the intersection form is then given by the matrix
\[
 \cap_{f} = 2N
\begin{pmatrix}
 -1 & 0 & -1\\
 1 & 0 & 1\\
 0 & 0 & 0
\end{pmatrix},
\]
so that $\hs{H}_f$ fails to satisfy Poincar{\'e} duality as an $\alg{A}_F$-bimodule.

\section{Dirac Operators and their Structure}

\subsection{The order one condition}

We now examine the structure of Dirac operators in detail. We will find it useful to begin with the study of operators between $\alg{A}$-bimodules (for fixed $\alg{A}$) satisfying a further generalisation of the order one condition. Thus, let $\alg{A}$ be a fixed real \Cstar-algebra, and let $\hs{H}_1$ and $\hs{H}_2$ be fixed $\alg{A}$-bimodules with multiplicity matrices $m_1$ and $m_2$, respectively. 

\begin{definition}
 We shall say that a map $T \in \bdd(\hs{H}_1,\hs{H}_2)$ satisfies the \term{generalised order one condition} if
\begin{equation}\label{genorderone}
  \forall a, b \in \alg{A}, \: (\lambda_2(a)T - T\lambda_1(a))\rho_1(b) = \rho_2(b)(\lambda_2(a)T - T\lambda_1(a)).
\end{equation}
\end{definition}

Note that if $\hs{H}_1 = \hs{H}_2$, then the generalised order one condition reduces to the usual order one condition on Dirac operators.

It is easy to check that the generalised order one condition is, in fact, equivalent to the following alternative condition:
\begin{equation}
 \forall a, b \in \alg{A}, \: (\rho_2(a)T - T\rho_1(a))\lambda_1(b) = \lambda_2(b)(\rho_2(a)T - T\rho_1(a)).
\end{equation}
Thus, the following are equivalent for $T \in \bdd(\hs{H}_1,\hs{H}_2)$:
\begin{enumerate}
 \item $T$ satisfies the generalised order one condition;
 \item For all $a \in \alg{A}$, $\lambda_2(a)T - T\lambda_1(a)$ is right $\alg{A}$-linear;
 \item For all $a \in \alg{A}$, $\rho_2(a)T - T\rho_1(a)$ is left $\alg{A}$-linear.
\end{enumerate}

Now, since the unitary group $\U(\alg{A})$ of $\alg{A}$ is a compact Lie group, let $\mu$ be the normalised bi-invariant Haar measure on $\U(\alg{A})$. 

\begin{lemma}~\label{decompproj}
Let $\hs{H}_1$ and $\hs{H}_2$ be $\alg{A}$-bimodules. Define operators $E_\lambda$ and $E_\rho$ on $\bdd^1_\alg{A}(\hs{H}_1,\hs{H}_2)$ by
\begin{equation}
 E_\lambda(T) := \int_{\U(\alg{A})} \ud\mu(u) \lambda_2(u)T\lambda_1(u^{-1}), \quad E_\rho(T) := \int_{\U(\alg{A})} \ud\mu(u) \rho_2(u^{-1})T\rho_1(u).
\end{equation}
Then $E_\lambda$ and $E_\rho$ are commuting idempotents such that
\[
 \im(E_\lambda) = \lbdd_\alg{A}(\hs{H}_1,\hs{H}_2), \quad \im(E_\rho) = \rbdd_\alg{A}(\hs{H}_1,\hs{H}_2),
\]
and
\[
 \ker(E_\lambda) = \im(\Id - E_\lambda) \subseteq \rbdd_\alg{A}(\hs{H}_1,\hs{H}_2), \quad \ker(E_\rho) = \im(\Id - E_\rho) \subseteq \lbdd_\alg{A}(\hs{H}_1,\hs{H}_2),
\]
while
\[
 \im(E_\lambda E_\rho) = \lrbdd_\alg{A}(\hs{H}_1,\hs{H}_2).
\]
\end{lemma}

\begin{proof}
First, the fact that $E_\lambda$ and $E_\rho$ are idempotents follows immediately from the Fubini-Tonelli theorem together with translation invariance of the Haar measure $\mu$, whilst commutation of $E_\lambda$ and $E_\rho$ follows from the Fubini-Tonelli theorem together with the commutation of left and right actions on $\hs{H}_1$ and on $\hs{H}_2$. Moreover, by construction, $E_\lambda$ and $E_\rho$ act as the identity on $\lbdd_\alg{A}(\hs{H}_1,\hs{H}_2)$ and $\rbdd_\alg{A}(\hs{H}_1,\hs{H}_2)$, respectively, so that
\[
 \im(E_\lambda) \supseteq \lbdd_\alg{A}(\hs{H}_1,\hs{H}_2), \quad \im(E_\rho) \supseteq \rbdd_\alg{A}(\hs{H}_1,\hs{H}_2).
\]

Now, let $T \in \bdd^1_\alg{A}(\hs{H}_1,\hs{H}_2)$. Then, by translation invariance of the Haar measure, it follows that for any $u \in \U(\alg{A})$,
\[
 E_\lambda(T) = \lambda_2(u) E_\lambda(T) \lambda_1(u)^*, \quad E_\rho(T) = \rho_2(u) E_\rho(T) \rho_1(u)^*,
\]
or equivalently,
\[
 \lambda_2(u) E_\lambda(T) = E_\lambda(T) \lambda_1(u), \quad \rho_2(u)E_\rho(T) = E_\rho(T) \rho_1(u).
\]
By the real analogue of the Russo-Dye theorem~\cite{Li}*{Lemma 2.15.16}, the convex hull of $\U(\alg{A})$ is weakly dense in the unit ball of $\alg{A}$, so that
\[
 \lambda_2(a) E_\lambda(T) = E_\lambda(T) \lambda_1(a), \quad \rho_2(a)E_\rho(T) = E_\rho(T) \rho_1(a)
\]
for all $a \in \alg{A}$, \ie $E_\lambda(T) \in \lbdd_\alg{A}(\hs{H}_1,\hs{H}_2)$ and $E_\rho(T) \in \rbdd_\alg{A}(\hs{H}_1,\hs{H}_2)$.

On the other hand,
\begin{align*}
 (\Id - E_\lambda)(T) &= \int_{\U(\alg{A})} \ud\mu(u) (T\lambda_1(u) - \lambda_2(u)T)\lambda_1(u^{-1}),\\ (\Id - E_\rho)(T) &= \int_{\U(\alg{A})} \ud\mu(u) (T\rho_1(u^{-1}) - \rho_2(u^{-1})T)\rho_1(u),
\end{align*}
so that by the generalised order one condition, $(\Id - E_\lambda)(T) \in \rbdd_\alg{A}(\hs{H}_1,\hs{H}_2)$ and $(\Id - E_\rho)(T) \in \lbdd_\alg{A}(\hs{H}_1,\hs{H}_2)$.

Finally, the commutation of $E_\lambda$ and $E_\rho$ together with our identification of $\im(E_\lambda)$ and of $\im(E_\rho)$ imply the desired result about $\im(E_\lambda E_\rho)$.
\end{proof}

Now, since
\[
 \im(\Id - E_\lambda) \subseteq \im(E_\rho), \quad \im(\Id - E_\rho) \subseteq \im(E_\lambda),
\]
one has that
\[
 (\Id - E_\lambda)E_\rho = \Id - E_\lambda, \quad (\Id - E_\rho)E_\lambda = \Id - E_\rho,
\]
which implies in turn that $\Id - E_\rho$, $E_\lambda E_\rho$ and $\Id - E_\lambda$ are mutually orthogonal idempotents such that
\[
 (\Id - E_\rho) + E_\lambda E_\rho + (\Id - E_\lambda) = \Id.
\]
We have therefore proved the following:

\begin{proposition}[Krajewski~\cite{Kraj98}*{\S 3.4}]\label{order1decomp}
Let $\rbdd_\alg{A}(\hs{H}_1,\hs{H}_2)^0$ denote $\ker(E_\lambda)$, and let $\lbdd_\alg{A}(\hs{H}_1,\hs{H}_2)^0$ denote $\ker(E_\rho)$. Then
\begin{equation}
 \bdd^1_\alg{A}(\hs{H}_1,\hs{H}_2) = \lbdd_\alg{A}(\hs{H}_1,\hs{H}_2)^0 \oplus \lrbdd_\alg{A}(\hs{H}_1,\hs{H}_2) \oplus \rbdd_\alg{A}(\hs{H}_1,\hs{H}_2)^0,
\end{equation}
where
\begin{equation}
 \lbdd_\alg{A}(\hs{H}_1,\hs{H}_2)^0 \oplus \lrbdd_\alg{A}(\hs{H}_1,\hs{H}_2) = \lbdd_\alg{A}(\hs{H}_1,\hs{H}_2)
\end{equation}
and
\begin{equation}
 \lrbdd_\alg{A}(\hs{H}_1,\hs{H}_2) \oplus \rbdd_\alg{A}(\hs{H}_1,\hs{H}_2)^0 = \rbdd_\alg{A}(\hs{H}_1,\hs{H}_2).
\end{equation}
\end{proposition}

Thus, elements of $\lbdd_\alg{A}(\hs{H}_1,\hs{H}_2)^0$ can be interpreted as the ``purely'' left $\alg{A}$-linear maps $\hs{H}_1 \to \hs{H}_2$, whilst elements of $\rbdd_\alg{A}(\hs{H}_1,\hs{H}_2)^0$ can be interpreted as the ``purely'' right $\alg{A}$-linear maps $\hs{H}_1 \to \hs{H}_2$.

One can readily check that the decomposition of Proposition~\ref{order1decomp} is respected by left multiplication by elements of $\lrbdd_\alg{A}(\hs{H}_2)$ and right multiplication by elements of $\lrbdd_\alg{A}(\hs{H}_1)$:

\begin{proposition}~\label{lrorder1decomp}
For any $T \in \bdd^1_\alg{A}(\hs{H}_1,\hs{H}_2)$, $A \in \lrbdd_\alg{A}(\hs{H}_1)$, $B \in \lrbdd_\alg{A}(\hs{H}_2)$,
\[
 E_\lambda(A T) = A E_\lambda(T), \quad E_\rho(T B) = E_\rho(T) B.
\]
\end{proposition}

Now, if $T \in \bdd(\hs{H}_1,\hs{H}_2)$, it is easy to see that $T$ satisfies the generalised order one condition if and only if each $T_{\alpha\beta}^{\gamma\delta}$ satisfies the generalised order one condition within $\bdd((\hs{H}_1)_{\alpha\beta},(\hs{H}_2)_{\gamma\delta})$; by abuse of notation, we will also denote by $E_\lambda$ and $E_\rho$ the appropriate idempotents on each $\bdd((\hs{H}_1)_{\alpha\beta},(\hs{H}_2)_{\gamma\delta})$. It then follows that
\[
 E_\lambda(T)_{\alpha\beta}^{\gamma\delta} = E_\lambda(T_{\alpha\beta}^{\gamma\delta}), \quad E_\rho(T)_{\alpha\beta}^{\gamma\delta} = E_\rho(T_{\alpha\beta}^{\gamma\delta}).
\]

Finally, let us turn to characterising $\ker(E_\lambda)$ and $\ker(E_\rho)$; before proceeding, we first need a technical lemma:

\begin{lemma}\label{schur}
 Let $G$ be a compact Lie group, and let $\mu$ be the bi-invariant Haar measure on $G$. Let $(\hs{H},\pi)$ and $(\hs{H}^\prime,\pi^\prime)$ be finite-dimensional irreducible unitary matrix representations of $G$. Then for any $T \in \bdd(\hs{H}^\prime,\hs{H})$, if $\pi \ncong \pi^\prime$ then
\begin{equation}
 \int_G \ud\mu(g) \pi(g)T\pi^\prime(g^{-1}) = 0,
\end{equation}
and if $\pi \cong \pi^\prime$, then for any unitary $G$-isomorphism $U: \hs{H}^\prime \to \hs{H}$,
\begin{equation}
 \int_G \ud\mu(g) \pi(g)T\pi^\prime(g^{-1}) = \frac{1}{\dim \hs{H}} \tr(T U^*)U.
\end{equation}
\end{lemma}

\begin{proof}
  Let
\[
 \tilde{T} = \int_G \ud\mu(g) \pi(g)T\pi^\prime(g^{-1}).
\]
which, by translation invariance of the Haar measure $\mu$, is a $G$-invariant map. If $\pi \ncong \pi^\prime$, then Schur's Lemma forces $\tilde{T}$ to vanish. If instead $\pi \cong \pi^\prime$, let $U: \hs{H} \to \hs{H}^\prime$ be a unitary $G$-isomorphism. Then by Schur's Lemma there exists some $\alpha \in \field{C}$ such that $\tilde{T} = \alpha U$; in fact,
\[
 \alpha = \alpha \frac{1}{\dim \hs{H}}\tr(U U^*) = \frac{1}{\dim \hs{H}} \tr(\tilde{T} U^*).
\]
One can then show that $\tr(\tilde{T} U^*) = \tr(T U^*)$ by introducing an orthonormal basis of $\hs{H}$ and then calculating directly.
\end{proof}

We now arrive at the desired characterisation:

\begin{proposition}~\label{zeromean}
If $T \in \rbdd_\alg{A}(\hs{H}_1,\hs{H}_2)$, then $E_\lambda(T) = 0$ if and only if for all $\alpha$, $\beta \in \supp(m_1) \cap \supp(m_2)$,
\[
 T_{\alpha\beta}^{\alpha\beta} \in \fsl(n_\alpha) \otimes M_{(m_2)_{\alpha\beta} \times (m_1)_{\alpha\beta}}(\field{C}) \otimes 1_{n_\beta},
\]
and if $T \in \lbdd_\alg{A}(\hs{H}_1,\hs{H}_2)$, then $E_\rho(T) = 0$ if and only if for all $\alpha$, $\beta \in \supp(m_1) \cap \supp(m_2)$,
\[
 T_{\alpha\beta}^{\alpha\beta} \in 1_{n_\alpha} \otimes M_{(m_2)_{\alpha\beta} \times (m_1)_{\alpha\beta}}(\field{C}) \otimes \fsl(n_\beta).
\]
\end{proposition}

\begin{proof}
Let $T \in \rbdd_\alg{A}(\hs{H}_1,\hs{H}_2)$. Then, by Proposition~\ref{linear}, it suffices to consider components $T_{\alpha\beta}^{\gamma\beta}$ for $\alpha$, $\beta$, $\gamma \in \spec{\alg{A}}$, which take the form
\[
 T_{\alpha\beta}^{\gamma\beta} = M_{\alpha\beta}^{\gamma} \otimes 1_{n_\beta}
\]
for $M_{\alpha\beta}^\gamma \in M_{n_\gamma \times n_\alpha}(\field{C}) \otimes M_{(m_2)_{\gamma\beta} \times (m_1)_{\alpha\beta}}(\field{C})$.

Now fix $\alpha$, $\beta$, $\gamma \in \spec{\alg{A}}$, and write
\[
 M_{\alpha\beta}^\gamma = \sum_{i=1}^k A_i \otimes B_i
\]
for $A_i \in M_{n_\gamma \times n_\alpha}(\field{C})$ and for $B_i \in M_{(m_2)_{\gamma\beta} \times (m_1)_{\alpha\beta}}(\field{C})$ linearly independent. It then follows by direct computation together with Lemma~\ref{schur} that
\[
 E_\lambda(T_{\alpha\beta}^{\gamma\beta}) = 
 \begin{cases}
  \frac{1}{n_\alpha} \left(\sum_{i=1}^k \tr(A_i) 1_{n_\alpha} \otimes B_i \right) \otimes 1_{n_\beta} &\text{if $\alpha = \gamma$,}\\
  0 &\text{otherwise},
 \end{cases} 
\]
so that by linear independence of the $B_i$,  $E_\lambda(T_{\alpha\beta}^{\gamma\beta})$ vanishes if and only if either $\alpha \neq \gamma$ or, $\alpha = \beta$ and each $A_i$ is traceless, and hence, if and only $\alpha \neq \gamma$ or, $\alpha = \beta$ and $M_{\alpha\beta}^\alpha \in \fsl(n_\alpha) \otimes M_{(m_2)_{\gamma\beta} \times (m_1)_{\alpha\beta}}(\field{C})$, as required.

\latin{Mutatis mutandis}, this argument also establishes the desired characterisation of $\ker(E_\rho)$.
\end{proof}

\subsection{Odd bilateral spectral triples}

Let us now take $\hs{H}_1 = \hs{H}_2 = \hs{H}$. By construction of $E_\lambda$ and $E_\rho$, the following conditions are readily seen to be equivalent for $T \in \bdd_\alg{A}^1(\hs{H})$:
\begin{enumerate}
 \item $T$ is self-adjoint;
 \item $E_\lambda(T)$ and $(\Id - E_\lambda)(T)$ are self-adjoint;
 \item $(\Id - E_\rho)(T)$ and $E_\rho(T)$ are self-adjoint;
 \item $(\Id - E_\rho)(T)$, $(E_\lambda E_\rho)(T)$ and $(\Id - E_\lambda)(T)$ are self-adjoint.
\end{enumerate}
Thus, in particular,
\begin{equation}
 \ms{D}_0(\alg{A},\hs{H}) = \lbdd_\alg{A}(\hs{H})^0_\sa \oplus \lrbdd_\alg{A}(\hs{H})_\sa \oplus \rbdd_\alg{A}(\hs{H})^0_\sa.
\end{equation}
In light of Proposition~\ref{lrorder1decomp}, we therefore have the following description of $\ms{D}(\alg{A},\hs{H})$:

\begin{proposition}
 Let $\hs{H}$ be an $\alg{A}$-bimodule. Then
\begin{equation}
 \ms{D}(\alg{A},\hs{H}) = \left( \lbdd_\alg{A}(\hs{H})^0_\sa \times \lrbdd_\alg{A}(\hs{H})_\sa \times \rbdd_\alg{A}(\hs{H})^0_\sa \right) / \lrU_\alg{A}(\hs{H}),
\end{equation}
where $\lrU_\alg{A}(\hs{H})$ acts diagonally by conjugation. 
\end{proposition}

Now, in light of Propositions~\ref{linear},~\ref{order1decomp} and~\ref{zeromean}, we can describe how to construct an arbitrary Dirac operator on an odd $\alg{A}$-bimodule $\hs{H}$ with multiplicity matrix $m$:
\begin{enumerate}
 \item For $\alpha$, $\beta$, $\gamma \in \spec{\alg{A}}$ such that $\alpha < \gamma$, choose $M_{\alpha\beta}^\gamma \in M_{n_\gamma m_{\gamma\beta} \times n_\alpha m_{\alpha\beta}}(\field{C})$;
 \item For $\alpha$, $\beta$, $\delta \in \spec{\alg{A}}$ such that $\beta < \delta$, choose $N_{\alpha\beta}^\delta \in M_{m_{\alpha\delta} n_\delta \times m_{\alpha\beta} n_\beta}(\field{C})$;
 \item For $\alpha$, $\beta \in \spec{\alg{A}}$, choose $M_{\alpha\beta}^\alpha \in M_{n_\alpha m_{\alpha\beta}}(\field{C})_\sa$ and $N_{\alpha\beta}^\beta \in M_{m_{\alpha\beta} n_\beta}(\field{C})_\sa$;
 \item Finally, for $\alpha$, $\beta$, $\gamma$, $\delta \in \spec{\alg{A}}$, set
 \begin{equation}
  D_{\alpha\beta}^{\gamma\delta} = 
  \begin{cases}
   M_{\alpha\beta}^\gamma \otimes 1_{n_\beta} &\text{if $\alpha < \gamma$ and $\beta = \delta$,}\\
   (M_{\gamma\beta}^\alpha)^* \otimes 1_{n_\beta} &\text{if $\alpha > \gamma$ and $\beta = \delta$,}\\
   1_{n_\alpha} \otimes N_{\alpha\beta}^\delta &\text{if $\alpha = \gamma$ and $\beta < \delta$,}\\
   1_{n_\alpha} \otimes (N_{\alpha\delta}^\beta)^* &\text{if $\alpha = \gamma$ and $\beta > \delta$,}\\
   M_{\alpha\beta}^\alpha \otimes 1_{n_\beta} + 1_{n_\alpha} \otimes N_{\alpha\beta}^\beta &\text{if $(\alpha,\beta) = (\gamma,\delta)$,}\\
   0 &\text{otherwise}.
  \end{cases}
 \end{equation}
\end{enumerate}

Note that for any $K = (1_{n_\alpha} \otimes K_{\alpha\beta} \otimes 1_{n_\beta})_{\alpha,\beta\in\spec{\alg{A}}} \in \lrbdd_\alg{A}(\hs{H})_\sa$ (so that each $K_{\alpha\beta}$ is self-adjoint), we can make the replacements
\[
 M_{\alpha\beta}^\alpha \mapsto M_{\alpha\beta}^\alpha + 1_{n_\alpha} \otimes K_{\alpha\beta}, \quad N_{\alpha\beta}^\beta \mapsto N_{\alpha\beta}^\beta - K_{\alpha\beta} \otimes 1_{n_\beta},
\]
and still obtain the same Dirac operator $D$; by Proposition~\ref{order1decomp}, this freedom is removed by requiring either that $M_{\alpha\beta}^\alpha \in \fsl(n_\alpha) \otimes M_{m_{\alpha\beta}}(\field{C})$ or that $N_{\alpha\beta}^\beta \in M_{m_{\alpha\beta}}(\field{C}) \otimes \fsl(n_\beta)$.

We now turn to the moduli space $\ms{D}(\alg{A},\hs{H})$ itself. By the above discussion and Corollary~\ref{oddunitary}, we can identify the space $\ms{D}_0(\alg{A},\hs{H})$ with
\begin{multline}
 \ms{D}_0(\alg{A},m) := \prod_{\alpha,\beta \in \spec{\alg{A}}} \prod_{\substack{\gamma \in \spec{\alg{A}} \\ \gamma > \alpha}} M_{n_\gamma m_{\gamma\beta} \times n_\alpha m_{\alpha\beta}}(\field{C}) \times \left(\fsl(n_\alpha) \otimes M_{m_{\alpha\beta}}(\field{C})\right)_\sa \\ \times \prod_{\substack{\delta \in \spec{\alg{A}} \\ \delta \geq \alpha}} M_{m_{\alpha\delta} n_\delta \times m_{\alpha\beta} n_\beta}(\field{C}) \times M_{m_{\alpha\beta} n_\beta}(\field{C})_\sa,
\end{multline}
and identify $\lrU_\alg{A}(\hs{H})$ with
\begin{equation}\label{modunit}
 \U(\alg{A},m) := \prod_{\alpha,\beta \in \spec{\alg{A}}} \U(m_{\alpha\beta}).
\end{equation}
By checking at the level of components, one sees that the action of $\lrU_\alg{A}(\hs{H})$ on the space $\ms{D}_0(\alg{A},\hs{H})$ corresponds under these identifications to the action of $\U(\alg{A},m)$ on $\ms{D}_0(\alg{A},m)$ defined by having $(U_{\alpha\beta}) \in \U(\alg{A},m)$ act on
\[
 (M_{\alpha\beta}^\gamma;M_{\alpha\beta}^\alpha;N_{\alpha\beta}^\delta;N_{\alpha\beta}^\beta) \in \ms{D}_0(\alg{A},m)
\]
by
\[
 M_{\alpha\beta}^\gamma \mapsto (1_{n_\gamma} \otimes U_{\gamma\beta}) M_{\alpha\beta}^\gamma (1_{n_\alpha\beta} \otimes U_{\alpha\beta}^*), \quad N_{\alpha\beta}^\delta \mapsto (U_{\alpha\delta} \otimes 1_{n_\delta}) N_{\alpha\beta}^\delta (U_{\alpha\beta}^* \otimes 1_{n_\beta}).
\]
We have therefore proved the following:

\begin{proposition}
 Let $\hs{H}$ be an odd $\alg{A}$-bimodule with multiplicity matrix $m$. Then
 \begin{equation}
  \ms{D}(\alg{A},\hs{H}) \cong \ms{D}_0(\alg{A},m) / \U(\alg{A},m).
 \end{equation}
\end{proposition}

\subsection{Even bilateral spectral triples}\label{evenbilateral}

For this section, let $(\hs{H},\gamma)$ be a fixed even $\alg{A}$-bimodule with pair of multiplicity matrices $(m^\even,m^\odd)$.

Now, let $D$ be a self-adjoint operator on $\hs{H}$ anticommuting with $\gamma$. Then, with respect to the decomposition $\hs{H} = \hs{H}^\even \oplus \hs{H}^\odd$ we can write
\[
 D =
 \begin{pmatrix}
 0 & \Delta^*\\
 \Delta & 0\\
 \end{pmatrix},
\]
where $\Delta = P^\odd D P^\even$, viewed as a map $\hs{H}^\even \to \hs{H}^\odd$. Thus, $D$ is uniquely determined by $\Delta$ and \latin{vice versa}. Moreover, one can  check that $D$ satisfies the order one condition if and only if $\Delta$ satisfies the generalised order one condition as a map $\hs{H}^\even \to \hs{H}^\odd$. We therefore have the following:

\begin{lemma}
Let $(\hs{H},\gamma)$ be an even $\alg{A}$-bimodule. Then the map $\ms{D}_0(\alg{A},\hs{H},\gamma) \to \bdd^1_\alg{A}(\hs{H}^\even,\hs{H}^\odd)$ defined by $D \mapsto P^\odd D P^\even$ is an isomorphism.
\end{lemma}

We now apply this Lemma to obtain our first result regarding the form of $\ms{D}(\alg{A},\hs{H},\gamma)$:

\begin{proposition}\label{evendirac}
The map 
\[
 \ms{D}(\alg{A},\hs{H},\gamma) \to \lrU_\alg{A}(\hs{H}^\odd) \backslash \bdd^1_\alg{A}(\hs{H}^\even,\hs{H}^\odd) / \lrU_\alg{A}(\hs{H}^\even)
\] 
defined by $[D] \mapsto [P^\odd D P^\even]$ is a homeomorphism.
\end{proposition}

\begin{proof}
Recall that $\lrU_\alg{A}(\hs{H},\gamma) = \lrU_\alg{A}(\hs{H}^\even,\hs{H}^\odd)$. We therefore have for $D \in \ms{D}_0(\alg{A},\hs{H},\gamma)$ and $U = U^\even \oplus U^\odd \in \lrU_\alg{A}(\hs{H}^\even,\hs{H}^\odd)$ that
\[
  P^\odd U D U^* P^\even = U^\odd P^\odd D P^\even (U^\even)^*.
\]
Thus, under the correspondence $\ms{D}_0(\alg{A},\hs{H},\gamma) \cong \bdd^1_\alg{A}(\hs{H}^\even,\hs{H}^\odd)$, the action of $\lrU_\alg{A}(\hs{H},\gamma)$ decouples into an action of $\lrU_\alg{A}(\hs{H}^\odd)$ by multiplication on the left and an action of $\lrU_\alg{A}(\hs{H}^\even)$ by multiplication by the inverse on the right. Thus, the map $[D] \to [P^\odd D P^\even]$ is not only well-defined but manifestly homeomorphic.
\end{proof}

Combining this last Proposition with Proposition~\ref{order1decomp}, we immediately obtain the following:

\begin{corollary}
 Let $(\hs{H},\gamma)$ be an even $\alg{A}$-bimodule. Then
 \begin{multline}
  \ms{D}(\alg{A},\hs{H},\gamma) \cong \lrU_\alg{A}(\hs{H}^\odd) \backslash ( \lbdd_\alg{A}(\hs{H}^\even,\hs{H}^\odd)^0 \\ \times \lrbdd_\alg{A}(\hs{H}^\even,\hs{H}^\odd) \times \rbdd_\alg{A}(\hs{H}^\even,\hs{H}^\odd)^0 ) / \lrU_\alg{A}(\hs{H}^\even),
 \end{multline}
 where $\lrU_\alg{A}(\hs{H}^\odd)$ acts diagonally by multiplication on the left, and $\lrU_\alg{A}(\hs{H}^\even)$ acts diagonally by multiplication on the right by the inverse.
\end{corollary}

Now, just as we did in the odd case, let us describe the construction of an arbitary Dirac operator $D$ on $(\hs{H},\gamma)$:
\begin{enumerate}
 \item For $\alpha$, $\beta$, $\gamma \in \spec{\alg{A}}$, choose $M_{\alpha\beta}^\gamma \in M_{n_\gamma m^\odd_{\gamma\beta} \times n_\alpha m^\even_{\alpha\beta}}(\field{C})$;
 \item For $\alpha$, $\beta$, $\delta \in \spec{\alg{A}}$, choose $N_{\alpha\beta}^\delta \in M_{m^\odd_{\alpha\delta} n_\delta \times m^\even_{\alpha\beta} n_\beta}(\field{C})$;
 \item Construct $\Delta \in \bdd^1_\alg{A}(\hs{H}^\even,\hs{H}^\odd)$ by setting, for $\alpha$, $\beta$, $\gamma$, $\delta \in \spec{\alg{A}}$,
 \begin{equation}
  \Delta_{\alpha\beta}^{\gamma\delta} =
  \begin{cases}
   M_{\alpha\beta}^\gamma \otimes 1_{n_\beta} &\text{if $\alpha \neq \gamma$ and $\beta = \delta$,}\\
   1_{n_\alpha} \otimes N_{\alpha\beta}^\delta &\text{if $\alpha = \gamma$ and $\beta \neq \delta$,}\\
   M_{\alpha\beta}^\alpha \otimes 1_{n_\beta} + 1_{n_\alpha} \otimes N_{\alpha\beta}^\beta &\text{if $(\alpha,\beta) = (\gamma,\delta)$,}\\
   0 &\text{otherwise;}\\
  \end{cases}
 \end{equation}
 \item Finally, set $D = \bigl( \begin{smallmatrix} 0 & \Delta^* \\ \Delta & 0 \end{smallmatrix} \bigr)$.
\end{enumerate}

Again, note that for any $K = (1_{n_\alpha} \otimes K_{\alpha\beta} \otimes 1_{n_\beta})_{\alpha,\beta\in\spec{\alg{A}}} \in \lrbdd_\alg{A}(\hs{H}^\even,\hs{H}^\odd)$ , we can make the replacements
\[
 M_{\alpha\beta}^\alpha \mapsto M_{\alpha\beta}^\alpha + 1_{n_\alpha} \otimes K_{\alpha\beta}, \quad N_{\alpha\beta}^\beta \mapsto N_{\alpha\beta}^\beta - K_{\alpha\beta} \otimes 1_{n_\beta},
\]
and still obtain the same Dirac operator $D$; by Proposition~\ref{order1decomp}, this freedom is removed by requiring either that
\[
 M_{\alpha\beta}^\alpha \in \fsl(n_\alpha) \otimes M_{m^\odd_{\alpha\beta} \times m^\even_{\alpha\beta}}(\field{C})
\]
or that
\[
N_{\alpha\beta}^\beta \in M_{m^\odd_{\alpha\beta} \times m^\even_{\alpha\beta}}(\field{C}) \otimes \fsl(n_\beta).
\]

Just as in the odd case, the above discussion and Corollary~\ref{oddunitary} imply that we can identify $\ms{D}_0(\alg{A},\hs{H},\gamma)$ with
\begin{multline}
 \ms{D}_0(\alg{A},m^\even,m^\odd) := \prod_{\alpha,\beta \in \spec{\alg{A}}} \prod_{\substack{\gamma \in \spec{\alg{A}} \\ \gamma \neq \alpha}} M_{n_\gamma m^\odd_{\gamma\beta} \times n_\alpha m^\even_{\alpha\beta}}(\field{C}) \\ \times \left(\fsl(n_\alpha) \otimes M_{m^\odd_{\alpha\beta} \times  m^\even_{\alpha\beta}}(\field{C})\right) \times \prod_{\delta \in \spec{\alg{A}}} M_{m^\odd_{\alpha\delta} n_\delta \times m^\even_{\alpha\beta} n_\beta}(\field{C}),
\end{multline}
and identify $\lrU_\alg{A}(\hs{H}^\even)$ and $\lrU_\alg{A}(\hs{H}^\odd)$ with $\U(\alg{A},m^\even)$ and $\U(\alg{A},m^\odd)$, respectively, which are defined according to Equation~\ref{modunit}. The actions of $\lrU_\alg{A}(\hs{H}^\even)$ and $\lrU_\alg{A}(\hs{H}^\odd)$ on $\bdd^1_\alg{A}(\hs{H}^\even,\hs{H}^\odd)$ therefore correspond under these identifications to the actions of $\U(\alg{A},m^\even)$ and $\U(\alg{A},m^\odd)$, respectively, on $\ms{D}_0(\alg{A},m^\even,m^\odd)$ defined by having $(U_{\alpha\beta}^\odd) \in \U(\alg{A},m^\odd)$ and $(U_{\alpha\beta}^\even) \in \U(\alg{A},m^\even)$ act on
\[
 (M_{\alpha\beta}^\gamma;M_{\alpha\beta}^\alpha;N_{\alpha\beta}^\delta) \in \ms{D}_0(\alg{A},m^\even,m^\odd)
\]
by
\[
 M_{\alpha\beta}^\gamma \mapsto (1_{n_\gamma} \otimes U_{\gamma\beta}^\odd)M_{\alpha\beta}^\gamma, \quad N_{\alpha\beta}^\delta \mapsto (U_{\alpha\delta}^\odd \otimes 1_{n_\delta}) N_{\alpha\beta}^\delta,
\]
and
\[
 M_{\alpha\beta}^\gamma \mapsto M_{\alpha\beta}^\gamma (1_{n_\alpha} \otimes (U_{\alpha\beta}^\even)^*), \quad N_{\alpha\beta}^\delta \mapsto  N_{\alpha\beta}^\delta ((U_{\alpha\beta}^\even)^* \otimes 1_{n_\beta}),
\]
respectively. Thus we have proved the following:

\begin{proposition}
Let $(\hs{H},\gamma)$ be an even $\alg{A}$-bimodule with multiplicity matrices $(m^\even,m^\odd)$. Then
\begin{equation}
 \ms{D}(\alg{A},\hs{H},\gamma) \cong \U(\alg{A},m^\odd) \backslash \ms{D}_0(\alg{A},m^\even,m^\odd) / \U(\alg{A},m^\even).
\end{equation}
\end{proposition}

In the quasi-orientable case, the picture simplifies considerably, as all components $\Delta_{\alpha\beta}^{\alpha\beta}$ necessarily vanish. One is then left, essentially, with the situation described by Krajewski~\cite{Kraj98}*{\S 3.4} and Paschke--Sitarz~\cite{PS98}*{\S 2.II} ; as mentioned before, one can find in the former the original definition of what are now called \term{Krajewski diagrams}. These diagrams, used extensively by Iochum, Jureit, Sch{\"u}cker and Stephan~\cites{ACG1,ACG2,ACG3,ACG4,ACG5,Sch05}, offer a concise, diagrammatic approach to the study of quasi-orientable even bilateral spectral triples that strongly emphasizes the underlying combinatorics. Though they do admit ready generalisation to the non-quasi-orientable case, we will not discuss them here.

We conclude our discussion of even bilateral spectral triples by recalling a result of Paschke and Sitarz of particular interest in relation to the NCG Standard Model.

\begin{proposition}[Paschke--Sitarz~\cite{PS98}*{Lemma 7}]
Let $(\hs{H},\gamma)$ be an orientable $\alg{A}$-bimodule. Then for all $D \in \ms{D}_0(\alg{A},\hs{H},\gamma)$,
\begin{equation}
 D = \sum_{\substack{i,j = 1\\ i \neq j}}^N \lambda(e_i)[D,\lambda(e_j)] + \sum_{\substack{k,l = 1\\ k \neq l}}^N \rho(e_k)[D,\rho(e_l)].
\end{equation}
\end{proposition}

\begin{proof}
 Fix $D \in \ms{D}_0(\alg{A},\hs{H},\gamma)$, and let
\begin{align*}
 T &:= D - \sum_{\substack{i,j = 1\\ i \neq j}}^N \lambda(e_i)[D,\lambda(e_j)] - \sum_{\substack{k,l = 1\\ k \neq l}}^N \rho(e_k)[D,\rho(e_l)]\\
 &= D - \sum_{\substack{i,j = 1\\ i \neq j}}^N \lambda(e_i)D\lambda(e_j) - \sum_{\substack{k,l = 1\\ k \neq l}}^N \rho(e_k)D\rho(e_l).
\end{align*}
Then for all $\alpha$, $\beta$, $\gamma$, $\delta \in \spec{\alg{A}}$,
\[
 T_{\alpha\beta}^{\gamma\delta} = 
\begin{cases}
 D_{\alpha\beta}^{\gamma\delta} &\text{if $r(\alpha) = r(\gamma)$, $r(\beta) = r(\delta)$,}\\
 -D_{\alpha\beta}^{\gamma\delta} &\text{if $r(\alpha) \neq r(\gamma)$, $r(\beta) \neq r(\delta)$,}\\
 0 &\text{otherwise,}
\end{cases}
\]
where for $\alpha \in \spec{\alg{A}}$, $r(\alpha)$ is the value of $j \in \{1,\dotsc,N\}$ such that $\alpha \in \spec{M_{k_j}(\field{K}_j)}$. However, by Proposition~\ref{order1decomp}, $D_{\alpha\beta}^{\gamma\delta}$ must vanish in the second case, whilst by Proposition~\ref{converseorientable}, $D_{\alpha\beta}^{\gamma\delta}$ must vanish in the first, so that $T = 0$.
\end{proof}

Now, let $(\alg{A},\hs{H},D,J,\gamma)$ be a real spectral triple of even $KO$-dimension. A \term{gauge potential} for the triple is then a self-adjoint operator on $\hs{H}$ of the form
\[
 \sum_{k=1}^n \lambda(a_k) [D,\lambda(b_k)],
\]
where $a_1,\dotsc,a_n$, $b_1,\dotsc,b_n \in \alg{A}$, and an \term{inner fluctuation of the metric} is a Dirac operator $D_A \in \ms{D}_0(\alg{A},\hs{H},J,\gamma)$ of the form
\[
 D_A := D + A + \epsp J A J^* = D + A + J A J^*,
\]
where $A$ is a gauge potential. One then has that for any gauge potential $A$, $(\alg{A},\hs{H},D,J,\gamma)$ and $(\alg{A},\hs{H},D_A,J,\gamma)$ are Morita equivalent. In this light, the last Proposition admits the following interpretation:

\begin{corollary}\label{orientmorita}
 Let $(\hs{H},J,\gamma)$ be an orientable real $\alg{A}$-bimodule of even $KO$-dimen\-sion. Then for all $D \in \ms{D}_0(\alg{A},\hs{H},\gamma,J)$, 
\begin{equation}
 A = - \sum_{\substack{i,j = 1\\ i \neq j}}^N \lambda(e_i)[D,\lambda(e_j)]
\end{equation}
is a gauge potential for the real spectral triple $(\alg{A},\hs{H},D,J,\gamma)$ such that $D_A = 0$.
\end{corollary}

Thus, every finite orientable real spectral triple $(\alg{A},\hs{H},D,J,\gamma)$ of even $KO$-dimension is Morita equivalent to the dynamically trivial triple $(\alg{A},\hs{H},0,J,\gamma)$.

\subsection{Real spectral triples of odd $KO$-dimension}

For this section, let $(\hs{H},J)$ be a real $\alg{A}$-bimodule of odd $KO$-dimension $n \bmod 8$ with multiplicity matrix $m$. We begin by reducing the study of Dirac operators on $(\hs{H},J)$ to that of self-adjoint right $\alg{A}$-linear operators on $\hs{H}$.

\begin{proposition}[Krajewski~\cite{Kraj98}*{\S 3.4}]\label{oddrealdirac}
Let $(\hs{H},J)$ be a real $\alg{A}$-bimodule of odd $KO$-dimension $n \bmod 8$. Then the map $R_n : \rbdd_\alg{A}(\hs{H})_\sa \to \ms{D}_0(\alg{A},\hs{H},J)$ defined by $R_n(M) := M + \epsp J M J^*$ is a surjection interwining the action of $\lrU_\alg{A}(\hs{H},J)$ on $\rbdd_\alg{A}(\hs{H})_\sa$ by conjugation with the action on $\ms{D}_0(\alg{A},\hs{H},J)$ by conjugation, and $\ker(R_n) \subseteq \lrbdd_\alg{A}(\hs{H})_\sa$.
\end{proposition}

\begin{proof}
First, note that $R_n$ is indeed well-defined, since by Equation~\ref{realintertwine}, for any $M \in \rbdd_\alg{A}(\hs{H})_\sa$, $J M J^* \in \lbdd_\alg{A}(\hs{H})_\sa$, and hence $R_n(M) \in \ms{D}_0(\alg{A},\hs{H},J)$.

Now, let $E_\lambda$ and $E_\rho$ be defined as in Lemma~\ref{decompproj}, and let $E_\lambda^\prime = \Id - E_\lambda$, $E_\rho^\prime = \Id - E_\rho$. Then, by construction of $E_\lambda$ and $E_\rho$ and Equation~\ref{realintertwine}, for any $T \in \bdd^1_\alg{A}(\hs{H})$,
\[
 E_\lambda(J T J^*) = J E_\rho(T) J^*,\quad E_\rho(J T J^*) = J E_\lambda(T) J^*.
\]
Hence, in particular, for $D \in \ms{D}_0(\alg{A},\hs{H},J)$, since $J D J^* = \epsp D$,
\begin{align*}
 D &= \frac{1}{2}(E_\lambda^\prime + E_\rho)(T) + \frac{1}{2}(E_\lambda + E_\rho^\prime)(T) \\ &= \frac{1}{2}(E_\lambda^\prime + E_\rho)(T) + \epsp J \frac{1}{2}(E_\lambda^\prime + E_\rho)(T) J^*\\ &= R_n\left(\frac{1}{2}(E_\lambda^\prime + E_\rho)(T)\right),
\end{align*}
where $\frac{1}{2}(E_\lambda^\prime + E_\rho)(T) \in \rbdd_\alg{A}(\hs{H})_\sa$.

Finally, that $R_n$ interwtines the actions of $\lrU_\alg{A}(\hs{H},J)$ follows from Proposition~\ref{lrorder1decomp} together with the fact that elements of $\lrU_\alg{A}(\hs{H},J)$, by definition, commute with $J$, whilst the fact that $R_n(M) = 0$ if and only if $M = -\epsp J M J^*$ implies that $\ker(R_n) \subseteq \lrbdd_\alg{A}(\hs{H})_\sa$.
\end{proof}

It follows, in particular, that $\ker(R_n)$ is invariant under the action of $\lrU_\alg{A}(\hs{H},J)$ by conjugation, so that the action of $\lrU_\alg{A}(\hs{H},J)$ on $\rbdd_\alg{A}(\hs{H})_\sa$ induces an action on the quotient $\rbdd_\alg{A}(\hs{H})_\sa / \ker(R_n)$, and hence $R_n$ induces an isomorphism
\begin{equation}
 \ms{D}_0(\alg{A},\hs{H},J) \cong \rbdd_\alg{A}(\hs{H})_\sa / \ker(R_n)
\end{equation}
of $\lrU_\alg{A}(\hs{H},J)$-representations. Thus we have proved the following:

\begin{corollary}
Let $(\hs{H},J)$ be a real $\alg{A}$-bimodule of odd $KO$-dimension $n \bmod 8$. Then
\begin{equation}
 \ms{D}(\alg{A},\hs{H},J) \cong \left(\rbdd_\alg{A}(\hs{H})_\sa / \ker(R_n)\right) / \lrU_\alg{A}(\hs{H},J).
\end{equation}
\end{corollary}

Discussion of $\ms{D}(\alg{A},\hs{H},J)$ thus requires discussion first of $\ker(R_n)$:

\begin{lemma}\label{oddker}
If $K = (1_{n_\alpha} \otimes K_{\alpha\beta} \otimes 1_{n_\beta})_{\alpha,\beta\in\spec{\alg{A}}} \in \lrbdd_\alg{A}(\hs{H})_\sa$, then $K \in \ker(R_n)$ if and only if for each $\alpha$, $\beta \in \spec{\alg{A}}$ such that $\alpha \neq \beta$,
\begin{equation}
 K_{\beta\alpha} = -\epsp K_{\alpha\beta}^T,
\end{equation}
and for each $\alpha \in \spec{\alg{A}}$,
\begin{equation}
 K_{\alpha\alpha} \in \hs{R}_\alpha(n) = 
 \begin{cases}
  \Sym_{m_{\alpha\alpha}}(\field{R}) &\text{if $n=1$,}\\
  i \mathfrak{sp}(m_{\alpha\alpha}) &\text{if $n=3$,}\\
  M_{m_{\alpha\alpha}/2}(\field{H})_\sa &\text{if $n=5$,}\\
  i \mathfrak{so}(m_{\alpha\alpha}) &\text{if $n=7$.}
 \end{cases}
\end{equation}
\end{lemma}

\begin{proof}
By definition of $R_n$, $K \in \ker(R_n)$ if and only if $K = -\epsp J K J^* = -\eps\epsp J K J$, and this in turn holds if and only if, for $\alpha$, $\beta \in \spec{\alg{A}}$ such that $\alpha \neq \beta$, 
\[
 K_{\alpha\beta} = -\epsp K_{\beta\alpha}^T,
\]
while for $\alpha \in \spec{\alg{A}}$,
\[
 K_{\alpha\alpha} =
 \begin{cases}
  -\epsp \overline{K_{\alpha\alpha}}, &\text{if $n=1$ or $7$,}\\
  \epsp I_{\alpha} K_{\alpha\alpha} I_{\alpha}^* &\text{if $n=3$ or $5$,}
 \end{cases}
\]
where $I_\alpha = \Omega_{m_{\alpha\alpha}} \circ \text{complex conjugation}$. In the case that $n =3$ or $5$, however, by construction, $M_{m_{\alpha\alpha}/2}(\field{H})$, viewed in the usual way as a real form of $M_{m_{\alpha\alpha}}(\field{C})$, is precisely the set of matrices in $M_{m_{\alpha\alpha}}(\field{C})$ commuting with $I_\alpha$. This, together with the hypothesis that $K$ is self-adjoint, so that each $K_{\alpha\beta}$ is self-adjoint, yields the desired result.
\end{proof}

We can now describe the the construction of an arbitrary Dirac operator $D$ on $(\hs{H},J)$:
\begin{enumerate}
 \item For $\alpha$, $\beta$, $\gamma \in \spec{\alg{A}}$ such that $\alpha < \gamma$, choose $M_{\alpha\beta}^\gamma \in M_{n_\gamma m_{\gamma\beta} \times n_\alpha m_{\alpha\beta}}(\field{C})$;
 \item For $\alpha$, $\beta \in \spec{\alg{A}}$, choose $M_{\alpha\beta}^\alpha \in M_{n_\alpha m_{\alpha\beta}}(\field{C})_\sa$;
 \item For $\alpha$, $\beta$, $\gamma$, $\delta \in \spec{\alg{A}}$, set
 \begin{equation}
  M_{\alpha\beta}^{\gamma\delta} = 
  \begin{cases}
   M_{\alpha\beta}^\gamma \otimes 1_{n_\beta} &\text{if $\alpha < \gamma$ and $\beta = \delta$,}\\
   (M_{\gamma\beta}^\alpha)^* \otimes 1_{n_\beta} &\text{if $\alpha > \gamma$ and $\beta = \delta$,}\\
   M_{\alpha\beta}^\alpha \otimes 1_{n_\beta} &\text{if $(\alpha,\beta) = (\gamma,\delta)$,}\\
   0 &\text{otherwise}.
  \end{cases}
 \end{equation}
 \item Finally, set $D = R_n(M)$.
\end{enumerate}

Now, let $K = (1_{n_\alpha} \otimes K_{\alpha\beta} \otimes 1_{n_\beta})_{\alpha,\beta\in\spec{\alg{A}}} \in \ker(R_n)$, so that each $K_{\alpha\beta}$ is self-adjoint, and for $\alpha$, $\beta \in \spec{\alg{A}}$ such that $\alpha \neq \beta$, $K_{\beta\alpha} = -\epsp K_{\alpha\beta}^T$ and $K_{\alpha\alpha} \in \hs{R}_\alpha(n)$. Thus, $K$ is uniquely specified by the matrices $K_{\alpha\beta} \in M_{m_{\alpha\beta}}(\field{C})_\sa$ for $\alpha < \beta$ and by the $K_{\alpha\alpha} \in \hs{R}_\alpha(n)$. Then, we can replace $M$ by $M + K$, \ie make the replacements, for $\alpha$, $\beta \in \spec{\alg{A}}$ such that $\alpha < \beta$,
\begin{gather*}
 M_{\alpha\beta}^\alpha \mapsto M_{\alpha\beta}^\alpha + 1_{n_\alpha} \otimes K_{\alpha\beta}, \quad M_{\beta\alpha}^\beta \mapsto M_{\beta\alpha}^\beta + 1_{n_\beta} \otimes (-\epsp K_{\alpha\beta}^T), \\ M_{\alpha\alpha}^\alpha \mapsto M_{\alpha\alpha}^\alpha + 1_{n_\alpha} \otimes K_{\alpha\alpha}
\end{gather*}
and obtain the same Dirac operator $D$. However, this is a freedom cannot generally be removed as we did in earlier cases, as it reflects precisely the non-injectivity of $R_n$. 

By the above discussion and Propositions~\ref{real17unitary} and~\ref{real35unitary}, we can identify the space $\ms{D}_0(\alg{A},\hs{H},J)$ with
\begin{multline}
 \ms{D}_0(\alg{A},m,n) := \prod_{\alpha \in \spec{\alg{A}}} \biggl[ M_{n_\alpha m_{\alpha\alpha}}(\field{C})_\sa / (1_{n_\alpha} \otimes \hs{R}_\alpha(n)) \\ \times \prod_{\substack{\beta \in \spec{\alg{A}} \\ \beta > \alpha}} (M_{n_\alpha m_{\alpha\beta}}(\field{C})_\sa \oplus M_{n_\beta m_{\alpha\beta}}(\field{C})_\sa)/M_{m_\alpha\beta}(\field{C})_\sa \times \prod_{\substack{\beta,\gamma \in \spec{\alg{A}}\\ \gamma > \alpha}} M_{n_\gamma m_{\gamma\beta} \times n_\alpha m_{\alpha\beta}}(\field{C}) \biggr],
\end{multline}
where $M_{m_{\alpha\beta}}(\field{C})_\sa$ is viewed as embedded in $M_{n_\alpha m_{\alpha\beta}}(\field{C})_\sa \oplus M_{n_\beta m_{\alpha\beta}}(\field{C})_\sa$ via the map
\[
 K \mapsto (1_{n_\alpha} \otimes K) \oplus (-\epsp 1_{n_\beta} \otimes K^T),
\]
and $\lrU_\alg{A}(\hs{H},J)$ with
\begin{equation}
 \U(\alg{A},m,n) := \prod_{\alpha \in \spec{\alg{A}}} \biggl( \hs{U}_\alpha(n) \times \prod_{\substack{\beta \in \spec{\alg{A}}\\ \beta > \alpha}} \U(m_{\alpha\beta}) \biggr),
\end{equation}
where
\[
 \hs{U}_\alpha(n) :=
\begin{cases}
 \Orth(m_{\alpha\alpha}) &\text{if $n = 1$ or $7$,}\\
 \Sp(m_{\alpha\alpha}) &\text{if $n=3$ or $5$.}
\end{cases}
\]
Then the action of $\lrU_\alg{A}(\hs{H};J)$ on $\ms{D}_0(\alg{A},\hs{H};J)$ corresponds under these identifications to the action of $\U(\alg{A},m,n)$ on $\ms{D}_0(\alg{A},m,n)$ defined by having the element $(U_{\alpha\alpha};U_{\alpha\beta}) \in \U(\alg{A},m,n)$ act on $\bigl([M_{\alpha\alpha}^\alpha];[(M_{\alpha\beta}^\alpha,M_{\beta\alpha}^\beta)];M_{\alpha\beta}^\gamma\bigr) \in \ms{D}_0(\alg{A},m,n)$ by
\begin{align*}
 [M_{\alpha\alpha}^\alpha] &\mapsto \bigl[(1_{n_\alpha} \otimes U_{\alpha\alpha}) M_{\alpha\alpha} (1_{n_\alpha} \otimes U_{\alpha\alpha}^*)\bigr];\\
 [(M_{\alpha\beta}^\alpha,M_{\beta\alpha}^\beta)] &\mapsto \biggl[\bigl((1_{n_\alpha} \otimes U_{\alpha\beta}) M_{\alpha\beta}^\alpha (1_{n_\alpha} \otimes U_{\alpha\beta}^*),(1_{n_\beta} \otimes \overline{U_{\alpha\beta}})M_{\beta\alpha}^\beta(1_{n_\beta} \otimes U_{\alpha\beta}^T)\bigr)\biggr];\\
 M_{\alpha\beta}^\gamma &\mapsto
 \begin{cases}
  (1_{n_\gamma} \otimes U_{\gamma\beta}) M_{\alpha\beta}^\gamma (1_{n_\alpha} \otimes U_{\alpha\beta}^*) &\text{if $\alpha < \beta$, $\gamma < \delta$,}\\
  (1_{n_\gamma} \otimes U_{\gamma\beta}) M_{\alpha\beta}^\gamma (1_{n_\alpha} \otimes U_{\beta\alpha}^T) &\text{if $\alpha > \beta$, $\gamma < \delta$,}\\
  (1_{n_\gamma} \otimes \overline{U_{\beta\gamma}}) M_{\alpha\beta}^\gamma (1_{n_\alpha} \otimes U_{\alpha\beta}^*) &\text{if $\alpha < \beta$, $\gamma > \delta$,}\\
  (1_{n_\gamma} \otimes \overline{U_{\beta\gamma}}) M_{\alpha\beta}^\gamma (1_{n_\alpha} \otimes U_{\beta\alpha}^T) &\text{if $\alpha > \beta$, $\gamma > \delta$.}
 \end{cases}
\end{align*}
We have therefore proved the following:

\begin{proposition}
 Let $(\hs{H},J)$ be a real $\alg{A}$-bimodule of odd $KO$-dimension $n \bmod 8$ with multiplicity matrix $m$. Then
 \begin{equation}
  \ms{D}(\alg{A},\hs{H},J) \cong \ms{D}_0(\alg{A},m,n) / \U(\alg{A},m,n).
 \end{equation}
\end{proposition}

\subsection{Real spectral triples of even $KO$-dimension}

We now turn to real spectral triples of even $KO$-dimension. Because of the considerable qualitative differences between the two cases, we consider separately the case of $KO$-dimension $0$ or $4 \bmod 8$ and $KO$-dimension $2$ or $6 \bmod 8$. 

In what follows, $(\hs{H},\gamma,J)$ is a fixed real $\alg{A}$-bimodule of even $KO$-dimension $n \bmod 8$ with multiplicity matrices $(m^\even,m^\odd)$; we  denote by $\bdd^1_\alg{A}(\hs{H}^\even,\hs{H}^\odd;J)$ the subspace of $\bdd^1_\alg{A}(\hs{H}^\even,\hs{H}^\odd)$ consisting of $\delta$ such that 
\[
 \begin{pmatrix}
 0 & \Delta^*\\
\Delta & 0
\end{pmatrix} \in \ms{D}_0(\alg{A},\hs{H};\gamma,J).
\]
It then follows that
\begin{equation}
 \ms{D}_0(\alg{A},\hs{H},\gamma,J) \cong \bdd^1_\alg{A}(\hs{H}^\even,\hs{H}^\odd;J)
\end{equation}
via the map $D \mapsto P^\odd D P^\even$.

\subsubsection{$KO$-dimension $0$ or $4 \bmod 8$}

Let us first consider the case where $n = 0$ or $4 \bmod 8$, \ie where $\epsp = 1$. Then $J = J^\even \oplus J^\odd$ for anti-unitaries $J^\even$ and $J^\odd$ on $\hs{H}^\even$ and $\hs{H}^\odd$, respectively, such that $(\hs{H}^\even,J^\even)$ and $(\hs{H}^\odd,J^\odd)$ are real $\alg{A}$-bimodules of $KO$-dimension $n^\prime \bmod 8$, where $n^\prime = 1$ or $7$ if $n=0$, $3$ or $5$ if $n=4$. In light of Corollary~\ref{real04unitary}, one can readily check the following analogue of Proposition~\ref{evendirac}:

\begin{proposition}\label{real04dirac}
 The map 
\[
 \ms{D}(\alg{A},\hs{H},\gamma,J) \to \lrU_\alg{A}(\hs{H}^\odd,J^\odd) \backslash \bdd^1_\alg{A}(\hs{H}^\even,\hs{H}^\odd;J) / \lrU_\alg{A}(\hs{H}^\even,J^\even)
\] 
defined by $[D] \mapsto [P^\odd D P^\even]$ is a homeomorphism.
\end{proposition}

Here, as before, $\lrU_\alg{A}(\hs{H}^\odd,J^\odd)$ acts by multiplication on the left, whilst the group $\lrU_\alg{A}(\hs{H}^\even,J^\even)$ acts by multiplication on the right by the inverse.

We now prove the relevant analogue of Proposition~\ref{oddrealdirac}:

\begin{proposition}\label{even04dirac}
 The map $R_n : \rbdd_\alg{A}(\hs{H}^\even,\hs{H}^\odd) \to \bdd^1_\alg{A}(\hs{H}^\even,\hs{H}^\odd,J)$ defined by $R_n(M) := M + J^\odd M (J^\even)^*$ is a surjection interwining the actions of $\lrU_\alg{A}(\hs{H}^\odd,J^\odd)$ by multiplication on the left and of $\lrU_\alg{A}(\hs{H}^\even,J^\even)$ by multiplication on the right by the inverse on $\rbdd_\alg{A}(\hs{H}^\even,\hs{H}^\odd)$ and $\bdd^1_\alg{A}(\alg{A},\hs{H}^\even,\hs{H}^\odd,J)$, and $\ker(R_n) \subseteq \lrbdd_\alg{A}(\hs{H}^\even,\hs{H}^\odd)$.
\end{proposition}

\begin{proof}
 First note that
\[
 \bdd^1_\alg{A}(\hs{H}^\even,\hs{H}^\odd,J) = \{\Delta \in \bdd^1_\alg{A}(\hs{H}^\even,\hs{H}^\odd) \mid \Delta = J^\odd \Delta (J^\even)^*\},
\]
so that $R_n$ is indeed well-defined by construction. Moreover, since $\lrU_\alg{A}(\hs{H}^\even,J^\even)$ and $\lrU_\alg{A}(\hs{H}^\odd,J^\odd)$ commute by definition with $J^\even$ and $J^\odd$, respectively, it then follows by construction of $R_n$ that $R_n$ does indeed have the desired intertwining properties.

Next, for $M \in \rbdd_\alg{A}(\hs{H}^\even,\hs{H}^\odd)$, we have that $R_n(M) = 0$ if and only if $M = -J^\odd M (J^\even)^*$, but $M$ is right $\alg{A}$-linear if and only if $J^\odd M (J^\even)^* = \eps J^\odd M J^\even$ is left $\alg{A}$-linear, so that $M \in \lrbdd_\alg{A}(\hs{H}^\even,\hs{H}^\odd)$ as claimed.

Finally, it is easy to check, just as in the proof of Proposition~\ref{oddrealdirac}, that for $\Delta \in \bdd^1_\alg{A}(\hs{H}^\even,\hs{H}^\odd,J)$,
\[
 \Delta = R_n\left(\frac{1}{2}(E_\lambda^\prime + E_\rho)(\Delta)\right),
\]
where $\frac{1}{2}(E_\lambda^\prime + E_\rho)(\Delta) \in \rbdd_\alg{A}(\hs{H}^\even,\hs{H}^\odd)$.
\end{proof}

Again, just as in the case of odd $KO$-dimension, this last result not only implies that the actions of $\lrU_\alg{A}(\hs{H}^\even,J^\even)$ and $\lrU_\alg{A}(\hs{H}^\odd,J^\odd)$ on $\rbdd_\alg{A}(\hs{H}^\even,\hs{H}^\odd)$ descend to actions on $\rbdd_\alg{A}(\hs{H}^\even,\hs{H}^\odd) / \ker(R_n)$, but that $R_n$ descends to an isomorphism $\rbdd_\alg{A}(\hs{H}^\even,\hs{H}^\odd) / \ker(R_n) \cong \bdd^1_\alg{A}(\hs{H}^\even,\hs{H}^\odd;J)$ intertwining the actions of $\lrU_\alg{A}(\hs{H}^\even,J^\even)$ and $\lrU_\alg{A}(\hs{H}^\odd,J^\odd)$, thereby yielding the following

\begin{corollary}
 Let $(\hs{H},\gamma,J)$ be a real $\alg{A}$-bimodule of $KO$-dimension $n \bmod 8$ for $n = 0$ or $4$. Then 
\begin{equation}
 \ms{D}(\alg{A},\hs{H},\gamma,J) \cong \lrU_\alg{A}(\hs{H}^\odd,J^\odd) \backslash \left(\rbdd_\alg{A}(\hs{H}^\even,\hs{H}^\odd) / \ker(R_n)\right) / \lrU_\alg{A}(\hs{H}^\even,J^\even) 
\end{equation}
\end{corollary}

\latin{Mutatis mutandis}, the proof of Lemma~\ref{oddker} yields the following characterisation of $\ker(R_n)$:

\begin{lemma}\label{even04ker}
If $K = (1_{n_\alpha} \otimes K_{\alpha\beta} \otimes 1_{n_\beta})_{\alpha,\beta\in\spec{\alg{A}}} \in \lrbdd_\alg{A}(\hs{H}^\even,\hs{H}^\odd)$, then $K \in \ker(R_n)$ if and only if for each $\alpha$, $\beta \in \spec{\alg{A}}$ such that $\alpha \neq \beta$,
\begin{equation}
 K_{\beta\alpha} = -\overline{K_{\alpha\beta}},
\end{equation}
and for each $\alpha \in \spec{\alg{A}}$,
\begin{equation}
 K_{\alpha\alpha} \in \hs{R}_\alpha(n) = 
 \begin{cases}
  i M_{m^\odd_{\alpha\alpha} \times m^\even_{\alpha\alpha}}(\field{R}) &\text{if $n=0$,}\\
  i M_{m^\odd_{\alpha\alpha}/2 \times m^\even_{\alpha\alpha}/2}(\field{H}) &\text{if $n=4$.}\\
 \end{cases}
\end{equation}
\end{lemma}

Note that such a map $K \in \lrbdd_\alg{A}(\hs{H}^\even,\hs{H}^\odd)$ is therefore entirely specified by the $K_{\alpha\beta} \in M_{m^\odd_{\alpha\beta} \times m^\even_{\alpha\beta}}(\field{C})$ for $\alpha < \beta$ and by the $K_{\alpha\alpha} \in \hs{R}_\alpha(n)$.

Let us now describe the construction of an arbitrary Dirac operator $D$ on the real $\alg{A}$-bimodule $(\hs{H},\gamma,J)$ of $KO$-dimension $n = 0$ or $4 \bmod 8$:
\begin{enumerate}
 \item For $\alpha$, $\beta$, $\gamma \in \spec{\alg{A}}$, choose $M_{\alpha\beta}^\gamma \in M_{n_\gamma m^\odd_{\gamma\beta} \times n_\alpha m^\even_{\alpha\beta}}(\field{C})$;
 \item Construct $M \in \rbdd_\alg{A}(\hs{H}^\even,\hs{H}^\odd)$ by setting for $\alpha$, $\beta$, $\gamma$, $\delta \in \spec{\alg{A}}$,
 \begin{equation}
  M_{\alpha\beta}^{\gamma\delta} =
  \begin{cases}
   M_{\alpha\beta}^\gamma \otimes 1_{n_\beta} &\text{if $\beta = \delta$,}\\
   0 &\text{otherwise;}\\
  \end{cases}
 \end{equation}
 \item Finally, set 
 \begin{equation}
  D =
  \begin{pmatrix}
   0 & R_n(M)^*\\
   R_n(M) & 0
  \end{pmatrix}.
 \end{equation}
\end{enumerate}

Just as before, if $R_n$ is non-injective, we can make the substitution $M \mapsto M + K$ for any $K \in \ker(R_n)$ and obtain the same Dirac operator D; at the level of components, we have for $\alpha$, $\beta \in \spec{\alg{A}}$ such that $\alpha < \beta$,
\begin{gather*}
 M_{\alpha\beta}^\alpha \mapsto M_{\alpha\beta}^\alpha + 1_{n_\alpha} \otimes K_{\alpha\beta}, \quad M_{\beta\alpha}^\beta \mapsto M_{\beta\alpha}^\beta + 1_{n_\alpha} \otimes (-\overline{K_{\alpha\beta}})\\
 M_{\alpha\alpha}^\alpha \mapsto M_{\alpha\alpha}^\alpha + 1_{n_\alpha} \otimes K_{\alpha\alpha}.
\end{gather*}
With these observations in hand, we can revisit the moduli space $\ms{D}(\alg{A},\hs{H},\gamma,J)$.

By the discussion above and Corollaries~\ref{real17unitary} and~\ref{real35unitary}, we can identify the space $\ms{D}_0(\alg{A},\hs{H},\gamma,J)$ with
\begin{multline}
 \ms{D}_0(\alg{A},m^\even,m^\odd,n) := \prod_{\alpha \in \spec{\alg{A}}} \biggl[ M_{n_\alpha m^\odd_{\alpha\alpha} \times n_\alpha m^\even_{\alpha\alpha}}(\field{C}) / (1_{n_\alpha} \otimes \hs{R}_\alpha(n)) \\ \times \prod_{\substack{\beta \in \spec{\alg{A}} \\ \beta > \alpha}} (M_{n_\alpha m^\odd_{\alpha\beta} \times n_\alpha m^\even_{\alpha\beta}}(\field{C}) \oplus M_{n_\beta m^\odd_{\alpha\beta} \times n_\beta m^\even_{\alpha\beta}}(\field{C}))/M_{m^\odd_{\alpha\beta} \times m^\even_{\alpha\beta}}(\field{C}) \\ \times \prod_{\substack{\beta,\gamma\in\spec{\alg{A}} \\ \gamma \neq \alpha}} M_{n_\gamma m^\odd_{\gamma\beta} \times n_{\alpha} m^\even_{\alpha\beta}}(\field{C})\biggr],
\end{multline}
where $M_{m^\odd_{\alpha\beta} \times m^\even_{\alpha\beta}}(\field{C})$ is viewed as embedded in the space
\[
 M_{n_\alpha m^\odd_{\alpha\beta} \times n_\alpha m^\even_{\alpha\beta}}(\field{C}) \oplus M_{n_\beta m^\odd_{\alpha\beta} \times n_\beta m^\even_{\alpha\beta}}(\field{C})
\]
via the map $K \mapsto (1_{n_\alpha} \otimes K) \oplus (-1_{n_\beta} \otimes \overline{K})$, and identify the groups $\lrU_\alg{A}(\hs{H}^\even;J^\even)$ and $\lrU_\alg{A}(\hs{H}^\odd;J^\odd)$ with $\U(\alg{A},m^\even,n^\prime)$ and $\U(\alg{A},m^\odd,n^\prime)$, respectively. Then the actions of $\lrU_\alg{A}(\hs{H}^\even;J^\even)$ and $\lrU_\alg{A}(\hs{H}^\odd;J^\odd)$ on $\bdd^1_\alg{A}(\hs{H}^\even,\hs{H}^\odd;J)$ corresponds under these identifications to the actions of the groups $\U(\alg{A},m^\even,n^\prime)$ and $\U(\alg{A},m^\odd,n^\prime)$, respectively, on $\ms{D}_0(\alg{A},m^\even,m^\odd,n)$ defined by having 
\[
 (U_{\alpha\alpha}^\odd;U_{\alpha\beta}^\odd) \in \U(\alg{A},m^\odd;n^\prime), \quad (U_{\alpha\alpha}^\even;U_{\alpha\beta}^\odd) \in \U(\alg{A},m^\even;n^\prime)
\]
act on $\bigl([M_{\alpha\alpha}^\alpha];[(M_{\alpha\beta}^\alpha,M_{\beta\alpha}^\beta)];M_{\alpha\beta}^\gamma\bigr) \in \ms{D}_0(\alg{A},m,n)$ by
\begin{align*}
 [M_{\alpha\alpha}^\alpha] &\mapsto \bigl[(1_{n_\alpha} \otimes U_{\alpha\alpha}^\odd) M_{\alpha\alpha}\bigr];\\
 [(M_{\alpha\beta}^\alpha,M_{\beta\alpha}^\beta)] &\mapsto \biggl[\bigl((1_{n_\alpha} \otimes U_{\alpha\beta}^\odd) M_{\alpha\beta}^\alpha,(1_{n_\beta} \otimes \overline{U_{\alpha\beta}^\odd})M_{\beta\alpha}^\beta\bigr)\biggr];\\
 M_{\alpha\beta}^\gamma &\mapsto
 \begin{cases}
  (1_{n_\gamma} \otimes U_{\gamma\beta}^\odd) M_{\alpha\beta}^\gamma  &\text{if $\gamma < \delta$,}\\
  (1_{n_\gamma} \otimes \overline{U_{\beta\gamma}^\odd}) M_{\alpha\beta}^\gamma &\text{if $\gamma > \delta$;}
 \end{cases}
\end{align*}
and
\begin{align*}
 [M_{\alpha\alpha}^\alpha] &\mapsto \bigl[M_{\alpha\alpha} (1_{n_\alpha} \otimes (U_{\alpha\alpha}^\even)^*)\bigr];\\
 [(M_{\alpha\beta}^\alpha,M_{\beta\alpha}^\beta)] &\mapsto \biggl[\bigl(M_{\alpha\beta}^\alpha (1_{n_\alpha} \otimes (U_{\alpha\beta}^\even)^*), M_{\beta\alpha}^\beta(1_{n_\beta} \otimes (U_{\alpha\beta}^\even)^T)\bigr)\biggr];\\
 M_{\alpha\beta}^\gamma &\mapsto
 \begin{cases}
  M_{\alpha\beta}^\gamma (1_{n_\alpha} \otimes (U_{\alpha\beta}^\even)^*) &\text{if $\alpha < \beta$,}\\
  M_{\alpha\beta}^\gamma (1_{n_\alpha} \otimes (U_{\beta\alpha}^\even)^T) &\text{if $\alpha > \beta$;}
 \end{cases}
\end{align*}
respectively. We have therefore proved the following:

\begin{proposition}
 Let $(\hs{H},\gamma,J)$ be a real $\alg{A}$-bimodule of even $KO$-dimension $n \bmod 8$ for $n = 0$ or $4$, with multiplicity matrices $(m^\even,m^\odd)$. Then
 \begin{equation}
  \ms{D}(\alg{A},\hs{H},\gamma,J) \cong \U(\alg{A},m^\odd,n^\prime) \backslash \ms{D}_0(\alg{A},m^\even,m^\odd,n) / \U(\alg{A},m^\even,n^\prime).
 \end{equation}
\end{proposition}

It is worth noting that considerable simplifications are obtained in the quasi-orientable case, as all components of the form $M_{\alpha\beta}^\alpha \otimes 1_{n_\beta}$ of $M \in \rbdd_\alg{A}(\hs{H}^\even,\hs{H}^\odd)$ must necessarily vanish, as must $\ker(R_n)$ itself. In particular, then, one is left with
\[
 \ms{D}_0(\alg{A},m^\even,m^\odd,n) = \prod_{\alpha,\beta,\gamma\in\spec{\alg{A}}} M_{n_\gamma m^\odd_{\gamma\beta} \times n_{\alpha} m^\even_{\alpha\beta}}(\field{C}).
\]

\subsubsection{$KO$-dimension $2$ or $6 \bmod 8$}\label{ko6}

Let us now consider the case where $n = 2$ or $n=6 \bmod 8$, \ie where $\epsp = -1$. Then
\[
 J =
\begin{pmatrix}
 0 & \eps \tilde{J}^*\\
 \tilde{J} & 0
\end{pmatrix}
\]
for $\tilde{J} : \hs{H}^\even \to \hs{H}^\odd$ anti-unitary, and $m^\odd = (m^\even)^T$. In light of Corollary~\ref{real26unitary}, one can easily establish, along the lines of Propositions~\ref{evendirac} and~\ref{even04dirac}, the following result:

\begin{proposition}
 Let $\lrU_\alg{A}(\hs{H}^\even)$ act on $\bdd^1_\alg{A}(\hs{H}^\even,\hs{H}^\odd;J)$ by
\[
 (U,\Delta) \mapsto \tilde{J}U\tilde{J}^* \Delta U^*
\]
for $U \in \lrU_\alg{A}(\hs{H}^\even)$ and $\Delta \in \bdd^1_\alg{A}(\hs{H}^\even,\hs{H}^\odd;J)$. Then the map
\[
 \ms{D}(\alg{A},\hs{H},\gamma,J) \to \bdd^1_\alg{A}(\hs{H}^\even,\hs{H}^\odd;J)/\lrU_\alg{A}(\hs{H}^\even)
\]
defined by $[D] \mapsto [P^\odd D P^\even]$ is a homeomorphism.
\end{proposition}

In the same way, we can define an action of $\lrU_\alg{A}(\hs{H}^\even)$ on $\rbdd_\alg{A}(\hs{H}^\even,\hs{H}^\odd)$. We now give the relevant analogue of Propositions~\ref{oddrealdirac} and~\ref{even04dirac}:

\begin{proposition}\label{even26dirac}
 The map $R_n : \rbdd_\alg{A}(\hs{H}^\even,\hs{H}^\odd) \to \bdd^1_\alg{A}(\hs{H}^\even,\hs{H}^\odd;J)$ defined by $R_n(M) := M + \eps \tilde{J} M^* \tilde{J}$ is a surjection intertwining the actions of the group $\lrU_\alg{A}(\hs{H}^\even)$ on $\rbdd_\alg{A}(\hs{H}^\even,\hs{H}^\odd)$ and $\bdd^1_\alg{A}(\hs{H}^\even,\hs{H}^\odd;J)$, and $\ker(R_n) \subset \lrbdd_\alg{A}(\hs{H}^\even,\hs{H}^\odd)$.
\end{proposition}

\begin{proof}
 First note that
\[
 \bdd^1_\alg{A}(\hs{H}^\even,\hs{H}^\odd;J) = \{\Delta \in \bdd^1_\alg{A}(\hs{H}^\even,\hs{H}^\odd) \mid \Delta = \eps \tilde{J} \Delta^* \tilde{J}\},
\]
as can be checked by direct calculation, so that $R_n$ is indeed well-defined. It also readily follows by construction of $R_n$ and the definition of the actions of $\lrU_\alg{A}(\hs{H}^\even)$ that $R_n$ has the desired intertwining properties.

Now, for $M \in \rbdd_\alg{A}(\hs{H}^\even,\hs{H}^\odd)$, one has that $R_n(M) = 0$ if and only if $M = -\eps \tilde{J} M^* \tilde{J}$, but $\tilde{J} M^* \tilde{J}$ is manifestly left $\alg{A}$-linear, so that $M \in \lrbdd_\alg{A}(\hs{H}^\even,\hs{H}^\odd)$, as claimed.

Finally, just as in the proof of Propositions~\ref{oddrealdirac} and~\ref{even04dirac}, one can easily check that for $\Delta \in \bdd^1_\alg{A}(\hs{H}^\even,\hs{H}^\odd;J)$,
\[
 \Delta = R_n\bigl(\frac{1}{2}(E^\prime_\lambda + E_\rho)(\Delta)\bigr),
\]
where $\frac{1}{2}(E^\prime_\lambda + E_\rho)(\Delta)$ is right $\alg{A}$-linear.
\end{proof}

Just as in the earlier cases, the action of $\lrU_\alg{A}(\hs{H}^\even)$ on $\rbdd_\alg{A}(\hs{H}^\even,\hs{H}^\odd)$ descends to an action on the quotient $\rbdd_\alg{A}(\hs{H}^\even,\hs{H}^\odd) / \ker(R_n)$, so that $R_n$ descends to an $\lrU_\alg{A}(\hs{H}^\even,\hs{H}^\odd)$-isomorphism
\begin{equation}
 \rbdd_\alg{A}(\hs{H}^\even,\hs{H}^\odd) / \ker(R_n) \cong \bdd^1_\alg{A}(\hs{H}^\even,\hs{H}^\odd;J),
\end{equation}
thereby yielding the following:

\begin{corollary}
 Let $(\hs{H},\gamma,J)$ be a real $\alg{A}$-bimodule of $KO$-dimension $n \bmod 8$ for $n = 2$ or $6$. Then
\begin{equation}
 \ms{D}(\alg{A},\hs{H},\gamma,J) \cong (\rbdd_\alg{A}(\hs{H}^\even,\hs{H}^\odd)/\ker(R_n)) / \lrU_\alg{A}(\hs{H}^\even).
\end{equation}
\end{corollary}

Again, \latin{mutatis mutandis}, the proof of Lemma~\ref{oddker} yields the following characterisation of $\ker(R_n)$:

\begin{lemma}
 If $K = (1_{n_\alpha} \otimes K_{\alpha\beta} \otimes 1_{n_\beta})_{\alpha,\beta\in\spec{\alg{A}}} \in \lrbdd_\alg{A}(\hs{H}^\even,\hs{H}^\odd)$, then $K \in \ker(R_n)$ if and only if for each $\alpha$, $\beta \in \spec{\alg{A}}$ such that $\alpha \neq \beta$,
\begin{equation}
 K_{\beta\alpha} = -\eps K_{\alpha\beta}^T,
\end{equation}
and for each $\alpha \in \spec{\alg{A}}$,
\begin{equation}
 K_{\alpha\alpha} \in \hs{R}_\alpha(n) =
\begin{cases}
 \Sym_{m^\even_{\alpha\alpha}}(\field{C}) &\text{if $n = 2$,}\\
 \fso(m^\even_{\alpha\alpha},\field{C}) &\text{if $n = 6$.}
\end{cases}
\end{equation}
\end{lemma}

Thus, such a map $K \in \ker(R_n)$ is entirely specified by the components $K_{\alpha\beta} \in M_{m^\even_{\beta\alpha} \times m^\even_{\alpha\beta}}(\field{C})$ for $\alpha < \beta$ and by the $K_{\alpha\alpha} \in \hs{R}_\alpha(n)$.

Note that the discussion of the construction of Dirac operators and of the freedom in the construction provided by $\ker(R_n)$ in the case of $KO$-dimension $0$ or $4 \bmod 8$ holds also in this case. Thus we can identify $\ms{D}(\alg{A},\hs{H},\gamma,J)$ with
\begin{multline}
 \ms{D}_0(\alg{A},m^\even,n) := \prod_{\alpha \in \spec{\alg{A}}} \biggl[ M_{n_\alpha m^\even_{\alpha\alpha}}(\field{C}) / \bigl(1_{n_\alpha} \otimes \hs{R}_\alpha(n)\bigr) \\
\times \prod_{\substack{\beta \in \spec{\alg{A}} \\ \beta > \alpha}} \bigl(M_{n_\alpha m^\even_{\beta\alpha} \times n_\alpha m^\even_{\alpha\beta}}(\field{C}) \oplus M_{n_\beta m^\even_{\alpha\beta} \times n_\beta m^\even_{\beta\alpha}}(\field{C})\bigr) / M_{m^\even_{\beta\alpha} \times m^\even_{\alpha\beta}}(\field{C}) \\
\times \prod_{\substack{\beta,\gamma\in\spec{\alg{A}}\\\gamma \neq \alpha}} M_{n_\gamma m^\even_{\beta\gamma} \times n_\alpha m^\even_{\alpha\beta}}(\field{C})\biggr],
\end{multline}
where $M_{m^\even_{\beta\alpha} \times m^\even_{\alpha\beta}}(\field{C})$ is viewed as embedded in the space
\[
 M_{n_\alpha m^\even_{\beta\alpha} \times n_\alpha m^\even_{\alpha\beta}}(\field{C}) \oplus M_{n_\beta m^\even_{\alpha\beta} \times n_\beta m^\even_{\beta\alpha}}(\field{C})
\]
via the map $K \mapsto (1_{n_\alpha} \otimes K) \oplus (-\eps 1_{n_\beta} \otimes K^T)$, and identify $\lrU_\alg{A}(\hs{H}^\even)$ with $\U(\alg{A},m^\even)$. Then the action of $\lrU_\alg{A}(\hs{H}^\even)$ on $\bdd^1_\alg{A}(\hs{H}^\even,\hs{H}^\odd;J)$ corresponds under these identifications with the action of $\U(\alg{A},m^\even)$ on $\ms{D}_0(\alg{A},m^\even,n)$ defined by having $(U_{\alpha\beta}) \in \U(\alg{A},m^\even)$ act on $\bigl([M_{\alpha\alpha}^\alpha];[(M_{\alpha\beta}^\alpha,M_{\beta\alpha}^\beta)];M_{\alpha\beta}^\gamma\bigr) \in \ms{D}_0(\alg{A},m,n)$ by
\begin{align*}
 [M_{\alpha\alpha}^\alpha] &\mapsto \bigl[(1_{n_\alpha} \otimes \overline{U_{\alpha\alpha}}) M_{\alpha\alpha}^\alpha (1_{n_\alpha} \otimes U_{\alpha\alpha}^*)\bigr];\\
 \bigl[(M_{\alpha\beta}^\alpha,M_{\beta\alpha}^\beta)\bigr] &\mapsto \biggl[\bigl((1_{n_\alpha} \otimes \overline{U_{\beta\alpha}}) M_{\alpha\beta}^\alpha (1_{n_\alpha} \otimes U_{\alpha\beta}^*), (1_{n_\beta} \otimes \overline{U_{\alpha\beta}}) M_{\beta\alpha}^\alpha (1_{n_\beta} \otimes U_{\beta\alpha}^*)\bigr)\biggr];\\
 M_{\alpha\beta}^\gamma &\mapsto (1_{n_\gamma} \otimes \overline{U_{\beta\gamma}}) M_{\alpha\beta}^\gamma (1_{n_\alpha} \otimes U_{\alpha\beta}^*).
\end{align*}
This, then, proves the following:

\begin{proposition}\label{moduli26}
 Let $(\hs{H},\gamma,J)$ be a real $\alg{A}$-bimdoule of even $KO$-dimension $n \bmod 8$ for $n = 2$ or $6$, with multiplicity matrices $(m^\even,(m^\even)^T)$. Then
\begin{equation}
 \ms{D}(\alg{A},\hs{H},\gamma,J) \cong \ms{D}_0(\alg{A},m^\even,n) / \U(\alg{A},m^\even).
\end{equation}
\end{proposition}

Again, considerable simplifications are obtained in the quasi-orientable case, just as for $KO$-dimension $0$ or $4 \bmod 8$.

\subsection{Dirac operators in the Chamseddine--Connes--Marcolli model}

Let us now apply the above results on Dirac operators and moduli spaces thereof to the bimodules appearing in the Chamseddine--Connes--Marcolli model. 

We begin with $(\hs{H}_F,\gamma_F,J_F,\epsilon_F)$ as an $S^0$-real $\alg{A}_{LR}$-bimodule of $KO$-dimension $6 \bmod 8$, which, as we shall now see, is essentially $S^0$-real in structure:

\begin{proposition}
 For the $S^0$-real $\alg{A}_{LR}$-bimodule $(\hs{H}_F,\gamma_F,J_F,\epsilon_F)$ of $KO$-dimen\-sion $6 \bmod 8$,
\[
 \ms{D}_0(\alg{A}_{LR},\hs{H}_F,\gamma_F,J_F) = \ms{D}_0(\alg{A}_{LR},\hs{H}_F,\gamma_F,J_F,\epsilon_F),
\]
and
\[
 \lrU_{\alg{A}_{LR}}(\hs{H},\gamma_F,J_F) = \lrU_{\alg{A}_{LR}}(\hs{H},\gamma_F,J_F,\epsilon_F),
\]
so that
\[
 \ms{D}(\alg{A}_{LR},\hs{H}_F,\gamma_F,J_F) = \ms{D}(\alg{A}_{LR},\hs{H}_F,\gamma_F,J_F,\epsilon_F).
\]

\end{proposition}

\begin{proof}
 To prove the first part of the claim, by Proposition~\ref{even26dirac}, it suffices to show that any right $\alg{A}_{LR}$-linear operator $\hs{H}_F^\even \to \hs{H}_F^\odd$ commutes with $\epsilon_F$. Thus, let $T \in \rbdd_{\alg{A}_{LR}}(\hs{H}_F^\even,\hs{H}_F^\odd)$. Then, since the signed multiplicity matrix $\mu$ of $(\hs{H}_F,\gamma_F)$ as an orientable even $\alg{A}_{LR}$-bimodule is given by
\[
 \mu =
 N \begin{pmatrix}
 0 & 0 & -1 & +1 & 0 & 0\\
 0 & 0 & 0 & 0 & 0 & 0\\
 +1 & 0 & 0 & 0 & +1 & 0\\
 -1 & 0 & 0 & 0 & -1 & 0\\
 0 & 0 & -1 & +1 & 0 & 0\\
 0 & 0 & 0 & 0 & 0 & 0
\end{pmatrix},
\]
it follows from Proposition~\ref{linear} that the only non-zero components of $T$ are $T_{\rep{2}_R \rep{1}}^{\rep{2}_L \rep{1}}$ and $T_{\rep{2}_R \rep{3}}^{\rep{2}_L \rep{3}}$, which both have domain and range within $\hs{H}_f = (\hs{H}_F)_i$, where $\epsilon$ acts as the identity. Thus, $T$ commutes with $\epsilon_F$.

To prove the next part of the claim, it suffices to show that any left and right $\alg{A}_{LR}$-linear operator on $\hs{H}_F$ commutes with $\epsilon_F$. But again, if $K \in \lrbdd_{\alg{A}_{LR}}(\hs{H}_F)$, then the only non-zero components of $K$ are of the form $K_{\alpha\beta}^{\alpha\beta}$, each of which therefore has both domain and range either within $\hs{H}_f$ or $\hs{H}_{\overline{f}} = J_F \hs{H}_f$, so that $K$ commutes with $\epsilon_F$. The last part of the claim is then an immediate consequence of the first two parts.
\end{proof}

Thus, by Proposition~\ref{diracs0reduction}, we have that
\begin{equation}
 \ms{D}_0(\alg{A}_{LR},\hs{H}_F,\gamma_F,J_F) = \ms{D}_0(\alg{A}_{LR},\hs{H}_F,\gamma_F,J_F,\epsilon_F) \cong \ms{D}_0(\alg{A}_{LR},\hs{H}_f,\gamma_f)
\end{equation}
and
\begin{equation}
 \ms{D}(\alg{A}_{LR},\hs{H}_F,\gamma_F,J_F) = \ms{D}(\alg{A}_{LR},\hs{H}_F,\gamma_F,J_F,\epsilon_F) \cong \ms{D}(\alg{A}_{LR},\hs{H}_f,\gamma_f),
\end{equation}
where $(\hs{H}_f,\gamma_f) = ((\hs{H}_F)_i,(\gamma_F)_i)$ is the orientable even $\alg{A}_{LR}$-bimodule with signed multiplicity matrix
\[
 \mu_f =
 N \begin{pmatrix}
 0 & 0 & 0 & 0 & 0 & 0\\
 0 & 0 & 0 & 0 & 0 & 0\\
 +1 & 0 & 0 & 0 & +1 & 0\\
 -1 & 0 & 0 & 0 & -1 & 0\\
 0 & 0 & 0 & 0 & 0 & 0\\
 0 & 0 & 0 & 0 & 0 & 0
\end{pmatrix}.
\]
In particular, then, $(\hs{H}_F,\gamma_F,J_F)$ as a real $\alg{A}_{LR}$-bimodule admits no \term{off-diagonal} Dirac operators, that is, Dirac operators with non-zero $P_{-i} D P_i : \hs{H}_f \to \hs{H}_{\overline{f}}$, or equivalently, that have non-vanishing commutator with $\epsilon_F$. Let us now examine $\ms{D}_0(\alg{A}_{LR},\hs{H}_F,\gamma_F,J_F)$ and $\ms{D}(\alg{A}_{LR},\hs{H}_F,\gamma_F,J_F)$, or rather, $\ms{D}_0(\alg{A}_{LR},\hs{H}_f,\gamma_f)$ and $\ms{D}(\alg{A}_{LR},\hs{H}_f,\gamma_f)$, in more detail.

First, it follows from the form of $\mu_f$ and Proposition~\ref{linear} that $\lbdd_{\alg{A}_{LR}}(\hs{H}_f^\even,\hs{H}_f^\odd)$ vanishes, whilst
\[
 \rbdd_{\alg{A}_{LR}}(\hs{H}_f^\even,\hs{H}_f^\odd) = M_{2N}(\field{C}) \oplus (M_{2N}(\field{C}) \otimes 1_3) \cong M_{2N}(\field{C}) \oplus M_{2N}(\field{C}).
\]
so that any Dirac operator on $\hs{H}_f$ (and hence on $\hs{H}_F$) is completely specified by a choice of $M_{\rep{2}_L \rep{1}}^{\rep{2}_R}$, $M_{\rep{2}_L \rep{3}}^{\rep{2}_R} \in M_{2N}(\field{C})$. Indeed, if $(m^\even,m^\odd)$ denotes the pair of multiplicity matrices of $(\hs{H}_f,\gamma_f)$, then, in the notation of subsection~\ref{ko6},
\[
 \ms{D}_0(\alg{A}_{LR},m^\even,6) = M_{2N}(\field{C}) \oplus M_{2N}(\field{C}).
\]
At the same time,
\[
 \lrU_{\alg{A}_{LR}}(\hs{H}_f^\even) = (1_2 \otimes \U(N)) \oplus (1_2 \otimes \U(N) \otimes 1_3) \cong \U(N) \times \U(N) =: \U(\alg{A}_{LR},m^\even)
\]
and
\[
 \lrU_{\alg{A}_{LR}}(\hs{H}_f^\odd) = (1_2 \otimes \U(N)) \oplus (1_2 \otimes \U(N) \otimes 1_3) \cong \U(N) \times \U(N) =: \U(\alg{A}_{LR},m^\odd).
\]
It then follows that
\begin{align}
 \ms{D}(\alg{A}_{LR},\hs{H}_f,\gamma_f) &\cong \U(\alg{A}_{LR},m^\odd) \backslash \ms{D}_0(\alg{A}_{LR},m^\even,6) / \U(\alg{A}_{LR},m^\even) \\
 &= \bigl(\U(N) \backslash M_{2N}(\field{C}) / \U(N)\bigr)^2,
\end{align}
where $\U(N)$ acts on the left by multiplication and on the right by multiplication by the inverse as $1_2 \otimes U(N)$. The two factors of the form $\U(N) \backslash M_{2N}(\field{C}) / \U(N)$ can thus be viewed as the parameter spaces of the components $M_{\rep{2}_L \rep{1}}^{\rep{2}_R}$ and $M_{\rep{2}_L \rep{3}}^{\rep{2}_R}$, respectively. 

Let us now consider $(\hs{H}_F,\gamma_F,J_F,\epsilon_F)$ as an $S^0$-real $\alg{A}_F$-bimodule, so that the multiplicity matrices $(m^\even,m^\odd)$ of $(\hs{H}_F,\gamma_F)$ are given by
\[
 m^\even = N
\begin{pmatrix}
1 & 1 & 0 & 0 & 0\\
0 & 0 & 0 & 0 & 0\\
1 & 0 & 0 & 1 & 0\\
1 & 1 & 0 & 0 & 0\\
0 & 0 & 0 & 0 & 0
\end{pmatrix}, \quad
m^\odd = N
\begin{pmatrix}
1 & 0 & 1 & 1 & 0\\
1 & 0 & 0 & 1 & 0\\
0 & 0 & 0 & 0 & 0\\
0 & 0 & 1 & 0 & 0\\
0 & 0 & 0 & 0 & 0
\end{pmatrix} = (m^\even)^T.
\]

Now it follows from the form of $(m^\even,m^\odd)$ that
\begin{multline*}
 \rbdd_{\alg{A}_F}(\hs{H}_F^\even,\hs{H}_F^\odd) = M_N(\field{C})^{\oplus 2} \oplus M_{N \times 2N}(\field{C})^{\oplus 2} \oplus M_{N \times 3N}(\field{C})^{\oplus 2} \\ \oplus (M_{N \times 2N}(\field{C}) \otimes 1_3)^{\oplus 2},
\end{multline*}
whilst
\[
 \ker(R_6) = \fsl(N,\field{C}) \subseteq M_N(\field{C})
\]
for the copy of $M_N(\field{C})$ corresponding to $\rbdd_\alg{A}((\hs{H}_F^\even)_{\rep{1}\rep{1}},(\hs{H}_F^\odd)_{\rep{1}\rep{1}})$. Since
$M_N(\field{C}) = \Sym_N(\field{C}) \oplus \fsl(N,\field{C})$, $M_N(\field{C})/\fsl(N,\field{C})$ can be identified with $\Sym_N(\field{C})$, so that
\begin{align*}
 &\ms{D}_0(\alg{A}_F,\hs{H}_F,\gamma_F,J_F)\\  \cong &\rbdd_{\alg{A}_F}(\hs{H}_F^\even,\hs{H}_F^\odd) / \ker(R_6)\\
 = &\Sym_N(\field{C}) \oplus M_N(\field{C}) \oplus M_{N \times 2N}(\field{C})^{\oplus 2} \oplus M_{N \times 3N}(\field{C})^{\oplus 2} \oplus (M_{N \times 2N}(\field{C}) \otimes 1_3)^{\oplus 2}.
\end{align*}
Thus, a Dirac operator $D$, which is specified by a choice of class 
\[
 [M] \in \rbdd_{\alg{A}_F}(\hs{H}_F^\even,\hs{H}_F^\odd) / \ker(R_6),
\]
is therefore specified in turn by the choice of the following matrices:
\begin{itemize}
 \item $M_{\rep{1}\rep{1}}^{\rep{1}} \in \Sym_N(\field{C})$, $M_{\rep{1}\rep{1}}^{\crep{1}} \in M_N(\field{C})$;
 \item $M_{\rep{2}\rep{1}}^{\rep{1}}$, $M_{\rep{2}\rep{1}}^{\crep{1}} \in M_{N \times 2N}(\field{C})$;
 \item $M_{\rep{3}\rep{1}}^{\rep{1}}$, $M_{\rep{3}\rep{1}}^{\crep{1}} \in M_{N \times 3N}(\field{C})$;
 \item $M_{\rep{2}\rep{3}}^{\rep{1}}$, $M_{\rep{2}\rep{3}}^{\crep{1}} \in M_{N \times 2N}(\field{C})$.
\end{itemize}
Indeed, it follows that
\begin{multline}
 \ms{D}_0(\alg{A}_F,m^\even,6) = \Sym_N(\field{C}) \oplus M_N(\field{C}) \oplus M_{N \times 2N}(\field{C})^{\oplus 2} \oplus M_{N \times 3N}(\field{C})^{\oplus 2} \\ \oplus M_{N \times 2N}(\field{C})^{\oplus 2}.
\end{multline}

Next, we have that $\U(\alg{A}_F,m^\even) = \U(N)^6$, with a copy of $\U(N)$ corresponding to each of $(\hs{H}_F^\even)_{\rep{1}\rep{1}}$, $(\hs{H}_F^\even)_{\rep{1}\crep{1}}$, $(\hs{H}_F^\even)_{\rep{2}\rep{1}}$, $(\hs{H}_F^\even)_{\rep{2}\rep{3}}$, $(\hs{H}_F^\even)_{\rep{3}\rep{1}}$, and $(\hs{H}_F^\even)_{\rep{3}\crep{1}}$. Then, by Proposition~\ref{moduli26},
\begin{equation}
 \ms{D}(\alg{A}_F,\hs{H}_F,\gamma_F,J_F) \cong \ms{D}_0(\alg{A}_F,m^\even,6) / \U(\alg{A}_F,m^\even)
\end{equation}
for the action of $\U(\alg{A}_F,m^\even)$ on $\ms{D}_0(\alg{A}_F,m^\even,6)$ given by having the element $(U_{\alpha\beta}) \in \U(\alg{A}_F,m^\even)$ act on $(M_{\alpha\beta}^\gamma) \in \ms{D}_0(\alg{A}_F,m^\even,6)$ by
\[
 M_{\alpha\beta}^\gamma \mapsto (1_{n_\gamma} \otimes \overline{U_{\beta\gamma}}) M_{\alpha\beta}^\gamma (1_{n_\alpha} \otimes U_{\alpha\beta}^*).
\]
Note that in the notation of~\cite{CM08}*{\S\S 13.4, 13.5}, for $(M_{\alpha\beta}^\gamma) \in \ms{D}_0(\alg{A}_F,m^\even,6)$,
\[
 M_{\rep{1}\rep{1}}^{\rep{1}} = \frac{1}{2}\Upsilon_R,
\]
so that the so-called Majorana mass term is already present in its final form, whilst for $U \in \U(\alg{A}_F,m^\even)$,
\[
 U = (U_{\rep{1}\rep{1}},U_{\rep{1}\crep{1}},U_{\rep{2}\rep{1}},U_{\rep{2}\rep{3}},U_{\rep{3}\rep{1}},U_{\rep{3}\crep{1}}) = (\overline{V_2},\overline{V_1},V_3,W_3,\overline{W_2},\overline{W_1}).
\]

Finally, let us compute the sub-moduli space $\ms{D}(\alg{A}_F,\hs{H}_F,\gamma_F,J_F;\field{C}_F)$ for 
\[
 \field{C}_F = \{(\zeta,\diag(\zeta,\overline{\zeta}),0) \in \alg{A}_F \mid \lambda \in \field{C}\} \cong \field{C}.
\]
It is easy to see that $[M] \in \rbdd_{\alg{A}_F}(\hs{H}_F^\even,\hs{H}_F^\odd) / \ker(R_6)$ yields an element of the subspace $\ms{D}_0(\alg{A}_F,\hs{H}_F,\gamma_F,J_F;\field{C}_F)$ if and only if $M$ commutes with $\lambda(\field{C}_F)$, but this holds if and only if for all $\zeta \in \field{C}$ and $\beta \in \spec{\alg{A}_F}$,
\begin{align*}
 \zeta M_{\rep{1}\beta}^{\rep{1}} &= M_{\rep{1}\beta}^{\rep{1}} \zeta, 
 &\overline{\zeta} M_{\rep{1}\beta}^{\crep{1}} &= M_{\rep{1}\beta}^{\crep{1}} \zeta,\\
 \zeta M_{\rep{2}\beta}^{\rep{1}} &= M_{\rep{2}\beta}^{\rep{1}} (\diag(\zeta,\overline{\zeta}) \otimes 1_N),
 &\overline{\zeta} M_{\rep{2}\beta}^{\crep{1}} &= M_{\rep{2}\beta}^{\crep{1}} (\diag(\zeta,\overline{\zeta}) \otimes 1_N),\\
 0 M_{\rep{3}\beta}^{\rep{1}} &= M_{\rep{3}\beta}^{\rep{1}} \zeta,
 &0 M_{\rep{3}\beta}^{\crep{1}} &= M_{\rep{3}\beta}^{\crep{1}} \zeta,
\end{align*}
which is in turn equivalent to having $M_{\rep{1}\rep{1}}^{\crep{1}}$, $M_{\rep{3}\rep{1}}^{\rep{1}}$ and $M_{\rep{3}\rep{1}}^{\crep{1}}$ all vanish, and
\[
 M_{\rep{2}\rep{1}}^{\rep{1}} = \begin{pmatrix} \Upsilon_\nu & 0 \end{pmatrix}, \quad M_{\rep{2}\rep{1}}^{\crep{1}} = \begin{pmatrix} 0 & \Upsilon_e \end{pmatrix}, \quad M_{\rep{2}\rep{3}}^{\rep{1}} = \begin{pmatrix} \Upsilon_u & 0 \end{pmatrix}, \quad M_{\rep{2}\rep{3}}^{\crep{1}} = \begin{pmatrix} 0 & \Upsilon_d \end{pmatrix},
\]
for $\Upsilon_\nu$, $\Upsilon_e$, $\Upsilon_u$, $\Upsilon_d \in M_N(\field{C})$. One can check that our notation is consistent with that of~\cite{CM08}*{\S\S 13.4, 13.5}. Indeed, if $\ms{D}_0(\alg{A}_F,m^\even,6;\field{C}_F)$ denotes the subspace of $\ms{D}_0(\alg{A}_F,m^\even,6)$ corresponding to $\ms{D}_0(\alg{A}_F,\hs{H}_F,\gamma_F,J_F;\field{C}_F)$, then
\begin{equation}
 \ms{D}(\alg{A}_F,\hs{H}_F,\gamma_F,J_F;\field{C}_F) \cong \ms{D}_0(\alg{A}_F,m^\even,6;\field{C}_F) / \U(\alg{A}_F,m^\even) \cong \ms{C}_q \times \ms{C}_l
\end{equation}
for
\begin{equation}
 \ms{C}_q := \bigl(\U(N) \times \U(N)\bigr) \backslash \bigl(M_N(\field{C}) \times M_N(\field{C})\bigr) / \U(N),
\end{equation}
where $\U(N)$ acts diagonally by multiplication on the right, and
\[
 \ms{C}_l := \bigl(\U(N) \times \U(N)\bigr) \backslash \bigl(M_N(\field{C}) \times M_N(\field{C}) \times \Sym_N(\field{C})\bigr) / \U(N),
\]
where $\U(N) \times \U(N)$ acts trivially on $\Sym_N(\field{C})$ and $\U(N)$ acts on $\Sym_N(\field{C})$ by
\[
 (V_2,\Upsilon_R) \mapsto V_2 \Upsilon_R V_2^T;
\]
note that $\ms{C}_q$ is the parameter space for the matrices $(\Upsilon_u,\Upsilon_d)$, whilst $\ms{C}_l$ is the parameter space for the matrices $(\Upsilon_\nu,\Upsilon_e,\Upsilon_R)$. Thus we have recovered the sub-moduli space of Dirac operators considered by Chamseddine--Connes--Marcolli~\cite{CCM07}*{\S\S 2.6, 2.7} (\cf also~\cite{CM08}*{\S\S 13.4, 13.5}).

\section{Applications to the Recent Work of Chamseddine and Connes}

In this section, we reformulate the results of Chamseddine and Connes in~\cites{CC08a,CC08b} and give new proofs thereof using the theory of bimodules and bilateral triples developed above.

Before continuing, recall that, up to automorphisms, the only real forms of $M_n(\field{C})$ are $M_n(\field{C})$, $M_n(\field{R})$, and, if $n$ is even, $M_{n/2}(\field{H})$.

\subsection{Admissible real bimodules}

We begin by studying what Chamseddine and Connes call \term{irreducible triplets}, namely, real $\alg{A}$-bimodules satisfying certain representation-theoretic conditions, along the lines of~\cite{CC08b}*{\S 2}. However, we shall progress by adding Chamseddine and Connes's various requirements for irreducible triplets one by one, bringing us gradually to their classification of irreducible triplets. 

In what follows, $\alg{A}$ will once more denote a fixed real \Cstar-algebra, and for $(\hs{H},J)$ a real $\alg{A}$-bimodule of odd $KO$-dimension, $\lrbdd_\alg{A}(\hs{H};J)$ will denote the real \Star-subalgebra of $\lrbdd_\alg{A}(\hs{H})$ consisting of elements commuting with $J$.

Let us now introduce the first explict requirement for irreducible triplets.

\begin{definition}
 Let $(\hs{H},J)$ be a real $\alg{A}$-bimodule of odd $KO$-dimension. We shall say that $(\hs{H},J)$ is \term{irreducible} if $0$ and $1$ are the only projections in $\lrbdd_\alg{A}(\hs{H};J)$.
\end{definition}

To proceed, we shall need the following:

\begin{lemma}\label{irreduciblelemma}
 Let $(\hs{H},J)$ be a real $\alg{A}$-bimodule of odd $KO$-dimension $n \bmod 8$ with multiplicity matrix $m$. Then
\begin{equation}
 \lrbdd_\alg{A}(\hs{H};J) \cong
\begin{cases}
  \biggl(\bigoplus_{\alpha \in \spec{\alg{A}}} M_{m_{\alpha\alpha}}(\field{R}) \biggr) \oplus \bigoplus_{\substack{\alpha,\beta \in \spec{\alg{A}}\\ \alpha < \beta}} M_{m_{\alpha\beta}}(\field{C}), &\text{if $n = 1$ or $7 \bmod 8$,}\\
 \biggl(\bigoplus_{\alpha \in \spec{\alg{A}}} M_{m_{\alpha\alpha}/2}(\field{H}) \biggr) \oplus \bigoplus_{\substack{\alpha,\beta \in \spec{\alg{A}}\\ \alpha < \beta}} M_{m_{\alpha\beta}}(\field{C}), &\text{if $n = 3$ or $5 \bmod 8$.}
\end{cases}
\end{equation}
\end{lemma}

\begin{proof}
 Let $T = (1_{n_\alpha} \otimes T_{\alpha\beta} \otimes 1_{n_\beta}) \in \lrbdd_\alg{A}(\hs{H})$. Just as for Propositions~\ref{real17unitary} and~\ref{real35unitary}, one can show that $[T,J] = 0$ if and only if for all $\alpha$, $\beta \in \spec{\alg{A}}$, $T_{\beta\alpha} = \overline{T_{\alpha\beta}}$ if $\alpha \neq \beta$ and
\[
 T_{\alpha\alpha} \in
\begin{cases}
 M_{m_{\alpha\alpha}}(\field{R}), &\text{if $n = 1$ or $7 \bmod 8$,}\\
 M_{m_{\alpha\alpha}/2}(\field{H}), &\text{if $n = 3$ or $5 \bmod 8$.}
\end{cases}
\]
Thus, $T \in \lrbdd_\alg{A}(\hs{H};J)$ is completely specified by the matrices $T_{\alpha\alpha}$ and $T_{\alpha\beta}$ for $\alpha > \beta$, giving rise to the isomorphisms of the claim.
\end{proof}

We can now formulate the part of the results of~\cite{CC08b}*{\S 2} that depends only on this notion of irreducibility.

\begin{proposition}
 Let $(\hs{H},J)$ be a real $\alg{A}$-bimodule of odd $KO$-dimension $n \bmod 8$ with multiplicity matrix $m$. Then $(\hs{H},J)$ is irreducible if and only if one of the following holds:
\begin{enumerate}
 \item There exists $\alpha \in \spec{\alg{A}}$ such that $m = 2^{(1-\eps)/2} E_{\alpha\alpha}$;
\item There exist $\alpha$, $\beta \in \spec{\alg{A}}$, $\alpha \neq \beta$, such that $m = E_{\alpha\beta} + E_{\beta\alpha}$.
\end{enumerate}
\end{proposition}

\begin{proof}
By definition, $(\hs{H},J)$ is irreducible if and only if the only projections in the real \Cstar-algebra $\lrbdd_\alg{A}(\hs{H},J)$ are $0$ and $1$, but by Lemma~\ref{irreduciblelemma}, this in turn holds if and only if one of the following holds:
\begin{enumerate}
 \item $\lrbdd_\alg{A}(\hs{H};J) \cong \field{R}$, so that $n = 1$ or $7 \bmod 8$, and $m = E_{\alpha\alpha}$ for some $\alpha \in \spec{\alg{A}}$,
 \item $\lrbdd_\alg{A}(\hs{H};J) \cong \field{H}$, so that $n = 3$ or $5 \bmod 8$, and $m = 2 E_{\alpha\alpha}$ for some $\alpha \in \spec{\alg{A}}$,
 \item $\lrbdd_\alg{A}(\hs{H};J) \cong \field{C}$, so that $m = E_{\alpha\beta} + E_{\beta\alpha}$ for some $\alpha$, $\beta \in \spec{\alg{A}}$, $\alpha \neq \beta$,
\end{enumerate}
which yields in turn the desired result.
\end{proof}

We shall call an irreducible odd $KO$-dimensional real $\alg{A}$-bimod\-ule $(\hs{H},J)$ \term{type A} if the first case holds, and \term{type B} if the second case holds; Chamseddine and Connes's first and second case for irreducible triplets~\cite{CC08b}*{Lemma 2.2} correspond to the type A and type B case, respectively. We shall also find it convenient to define the \term{skeleton} $\skel(\hs{H},J)$ of such a bimodule as follows:
\begin{enumerate}
 \item if $(\hs{H},J)$ is type A, then $\skel(\hs{H},J) := \{\alpha\}$, where $\alpha \in \spec{\alg{A}}$ is such that $\mult[\hs{H}] = 2^{\frac{1-\eps}{2}} E_{\alpha\alpha}$;
 \item if $(\hs{H},J)$ is type B, then $\skel(\hs{H},J) := \{\alpha,\beta\}$, where $\alpha$, $\beta \in \spec{\alg{A}}$, $\alpha \neq \beta$, are such that $\mult[\hs{H}] = E_{\alpha\beta} + E_{\beta\alpha}$.
\end{enumerate}

Let us now introduce the second explicit requirement for irreducible triplets.

\begin{definition}
 An $\alg{A}$-bimodule $\hs{H}$ is \term{(left) separating} if there exists some $\xi \in \hs{H}$ such that $\lambda(\alg{A})^\prime \xi = \hs{H}$. Such a vector $\xi$ is then called a \term{separating vector} for $\alg{A}$.
\end{definition}

Recall that for a representation $\hs{X}$ of a complex \Cstar-algebra $\alg{C}$, $\xi \in \hs{X}$ is a separating vector if and only if the map $\alg{C} \to \hs{X}$ given by $c \mapsto c \xi$ is injective.

\begin{lemma}\label{separatinglemma}
 Let $p$, $q \in \semiring{N}$. There exists a separating vector $\xi$ for the usual action of $M_p(\field{C})$ on $\field{C}^p \otimes \field{C}^q$ as $M_p(\field{C}) \otimes 1_q$ if and only if $p \leq q$.
\end{lemma}

\begin{proof}
 Let $\{e_i\}_{i=1}^p$ be a basis for $\field{C}^p$, and let $\{f_j\}_{j=1}^q$ be a basis for $\field{C}^q$. 

First suppose that $p \leq q$. Let $\xi \in \field{C}^p \otimes \field{C}^q$ be given by $\xi = \sum_{i=1}^p e_i \otimes f_i$. Then for any $a$, $b \in M_p(\field{C})$,
\[
 \left(a \otimes 1_q\right) \xi - \left(b \otimes 1_q\right) \xi = \sum_{i=1}^p \left(\sum_{l=1}^p (a_i^l - b_i^l) e_l \right) \otimes f_i
\]
so that by linear independence of the $e_i$ and $f_j$, the left-hand side vanishes if and only if for each $i$ and $l$, $a_i^l - b_i^l = 0$, \ie $a = b$. Hence, $\xi$ is indeed a separating vector.

Now suppose that $p > q$. Then $\dim_{\field{C}} M_p(\field{C}) - \dim_{\field{C}} \field{C}^p \otimes \field{C}^q = p(p-q) > 0$, so that for any $\xi \in \field{C}^p \otimes \field{C}^q$, the map $M_p(\field{C}) \mapsto \field{C}^p \otimes \field{C}^q$ given by $a \mapsto \left(a \otimes 1_q \right) \xi$ cannot possibly be injective, and hence $\xi$ cannot possibly be separating.
\end{proof}

We can now reformulate that part of the results in~\cite{CC08b}*{\S 2} that depends only on irreducibility and the existence of a separating vector.

\begin{proposition}
 Let $(\hs{H},J)$ be an irreducible real $\alg{A}$-bimodule of odd $KO$-dimen\-sion $n \bmod 8$.
\begin{enumerate}
 \item If $(\hs{H},J)$ is type A, then it is separating;
 \item If $(\hs{H},J)$ is type B with skeleton $(\alpha,\beta)$, then $(\hs{H},J)$ is separating if and only if $n_{\alpha} = n_{\beta}$.
\end{enumerate}
\end{proposition}

\begin{proof}
First suppose that $(\hs{H},J)$ is type A. Let $\{\alpha\} = \skel(\hs{H},J)$, and let $m_n = 2^{(1-\eps)/2}$. Then $\hs{H} = \field{C}^{n_\alpha} \otimes \field{C}^{m_n} \otimes \field{C}^{n_\alpha} = \field{C}^{n_\alpha} \otimes \field{C}^{m_n n_\alpha}$, and the left action $\lambda$ of $\alg{A}$ on $\hs{H}$ is thus given by $\lambda_\alpha \otimes 1_{m_n n_{\alpha}}$. Now
\[
 \lambda(\alg{A})^\prime = \left(\lambda_{\alpha}(\alg{A}) \otimes 1_{m_n n_{\alpha}}\right)^\prime = \left(M_{n_\alpha}(\field{C}) \otimes 1_{m_n n_\alpha}\right)^\prime,
\]
so that the action $\lambda$ of $\alg{A}$ admits a separating vector if and only if the action of $M_{n_\alpha}(\field{C})$ as $M_{n_\alpha}(\field{C}) \otimes 1_{m_n n_\alpha}$ admits a separating vector, but by Lemma~\ref{separatinglemma} this is indeed the case, as $n_\alpha \leq m_n n_\alpha$.

Now, suppose that $(\hs{H},J)$ is type B. Let $\{\alpha,\beta\} = \skel(\hs{H},J)$. Then
\[
 \hs{H} = (\field{C}^{n_\alpha} \otimes \field{C}^{n_\beta}) \oplus (\field{C}^{n_\beta} \otimes \field{C}^{n_\alpha}),
\]
and the left action $\lambda$ of $\alg{A}$ on $\hs{H}$ is given by $\lambda = (\lambda_{\alpha} \otimes 1_{n_\beta}) \oplus (\lambda_{\beta} \otimes 1_{n_\alpha})$. Since $\alpha \neq \beta$,
\begin{align*}
 \lambda(\alg{A})^\prime &= \left((\lambda_{\alpha}(\alg{A}) \otimes 1_{n_\beta}) \oplus (\lambda_\beta(\alg{A}) \otimes 1_{n_\alpha})\right)^\prime\\ &= \left((M_{n_\alpha}(\field{C}) \otimes 1_{n_\beta}) \oplus (M_{n_\beta}(\field{C}) \otimes 1_{n_\alpha})\right)^\prime,
\end{align*}
so that the action $\lambda$ of $\alg{A}$ admits a separating vector if and only if the action of $M_{n_\alpha}(\field{C}) \oplus M_{n_\beta}(\field{C})$ as $(M_{n_\alpha}(\field{C}) \otimes 1_{n_\beta}) \oplus (M_{n_\beta}(\field{C}) \otimes 1_{n_\alpha})$ admits a separating vector. Since $\dim_{\field{C}}M_{n_\alpha}(\field{C}) \oplus M_{n_\beta}(\field{C}) - \dim_{\field{C}}\hs{H} = (n_\alpha - n_\beta)^2$, if $n_\alpha \neq n_\beta$ then no injective linear maps $M_{n_\alpha}(\field{C}) \oplus M_{n_\beta}(\field{C}) \to \hs{H}$ can exist, and in particular, there exist no separating vectors for the action of $M_{n_\alpha}(\field{C}) \oplus M_{n_\beta}(\field{C})$, and hence for $\lambda$. Suppose instead that $n_\alpha = n_\beta = n$. Then
\[
 \hs{H} = (\field{C}^n \otimes \field{C}^n) \oplus (\field{C}^n \otimes \field{C}^n)
\]
so that, since $\alpha \neq \beta$, $\lambda(\alg{A})^\prime = (M_n(\field{C}) \otimes 1_n)^\prime \oplus (M_n(\field{C}) \otimes 1_n)^\prime$. Thus, if $\xi$ is the separating vector for the action of $M_n(\field{C})$ on $\field{C}^n \otimes \field{C}^n$ given by the proof of Lemma~\ref{separatinglemma}, then $\xi\oplus\xi$ is also a separating vector for the action $\lambda$ of $\alg{A}$, and hence $(\hs{H},J)$ is indeed separating.
\end{proof}

Let us now introduce the final requirement for irreducible triplets; recall that the complex form of a real \Cstar-algebra $\alg{A}$ a real \Cstar-algebra is denoted by $\alg{A}_\field{C}$.

\begin{definition}
 We shall call an $\alg{A}$-bimodule $\hs{H}$ \term{complex-linear} if both left and right actions of $\alg{A}$ on $\hs{H}$ extend to $\field{C}$-linear actions of $\alg{A}_{\field{C}}$, making $\hs{H}$ into a complex $\alg{A}_{\field{C}}$-bimodule.
\end{definition}

It follows immediately that a $\alg{A}$-bimodule $\hs{H}$ is complex-linear if and only if for $m = \mult[\hs{H}]$, $m_{\alpha\beta} = 0$ whenever $\alpha$ or $\beta$ is conjugate-linear. In particular, by Proposition~\ref{converseorientable}, it follows that a complex-linear quasi-orientable graded bimodule is always orientable.

We can now reformulate Chamseddine and Connes's definition for irreducible triplets:

\begin{definition}
 An \term{irreducible triplet} is a triplet $(\alg{A},\hs{H},J)$, where $\alg{A}$ is a finite-dimensional real \Cstar-algebra and $(\hs{H},J)$ is a complex-linear, separating, irreducible real $\alg{A}$-bimodule of odd $KO$-dimension such that the left action of $\alg{A}$ on $\hs{H}$ is faithful.
\end{definition}

Note that for $\hs{H}$ a real $\alg{A}$-bimodule, the left action of $\alg{A}$ is faithful if and only if the right action is faithful.

By combining the above results, we immediately obtain Cham\-sed\-dine and Connes's classification of irreducible triplets:

\begin{proposition}[Chamseddine--Connes~\cite{CC08b}*{Propositions 2.5, 2.8}]
 Let $\alg{A}$ be a finite-dimensional real \Cstar-algebra, and let $(\hs{H},J)$ be a real $\alg{A}$-bimodule of odd $KO$-dimension $n \bmod 8$. Then $(\alg{A},\hs{H},J)$ is an irreducible triplet if and only if one of the following cases holds:
\begin{enumerate}
 \item There exists $n \in \semiring{N}$ such that $\alg{A} = M_{k}(\field{K})$ for a real form $M_{k}(\field{K})$ of $M_n(\field{C})$, and
\begin{equation}
 \mult[\hs{H}] = 2^{(1-\eps)/2} E_{\rep{n}\rep{n}};
\end{equation}
 \item There exists $n \in \semiring{N}$ such that $\alg{A} = M_{k_1}(\field{K}_1) \oplus M_{k_2}(\field{K}_2)$ for real forms $M_{k_1}(\field{K}_1)$ and $M_{k_2}(\field{K}_2)$ of $M_n(\field{C})$, and 
\begin{equation}
 \mult[\hs{H}] = E_{\rep{n}_1 \rep{n}_2} + E_{\rep{n}_2 \rep{n_1}}.
\end{equation}
\end{enumerate}
\end{proposition}

\subsection{Gradings}

We now seek a classification of gradings inducing even $KO$-dimen\-sional real bimodules from irreducible triplets.

\begin{definition}
 Let $(\alg{A},\hs{H},J)$ be an irreducible triplet. We shall call a $\ring{Z}_2$-grading $\gamma$ on $\hs{H}$ as a Hilbert space \term{compatible} with $(\alg{A},\hs{H},J)$ if and only if the following conditions all hold:
\begin{enumerate}
 \item For every $a \in \alg{A}$, $\gamma \lambda(a) \gamma \in \lambda(\alg{A})$;
 \item The operator $\gamma$ either commutes or anticommutes with $J$.
\end{enumerate}
\end{definition}

Given a compatible grading $\gamma$ for an irreducible triplet $(\alg{A},\hs{H},J)$, one can view $(\hs{H},\gamma,J)$ as a real $\alg{A}^\even$-bimodule of even $KO$-dimension, for $\alg{A}^\even = \{a \in \alg{A} \mid [\lambda(a),\gamma] = 0\}$, with $KO$-dimension specified by the values of $\eps$ and $\epspp$ such that $J^2 = \eps$, $\gamma J = \epspp J \gamma$.

Now, recall that a $\ring{Z}_2$-grading on a real \Cstar-algebra $\alg{A}$ is simply an automorphism $\Gamma$ on $\alg{A}$ satisfying $\Gamma^2 = \Id$; we call such a grading \term{admissible} if and only if $\Gamma$ extends to a $\field{C}$-linear grading on $\alg{A}_\field{C}$. Thus, if $(\alg{A},\hs{H},J)$ is an irreducible triplet and $\gamma$ is a grading on $\hs{H}$, then $\gamma$ satisfies the first condition for compatibility if and only if there exists some admissible grading $\Gamma$ on $\alg{A}$ such that $\Ad_\gamma \circ \lambda = \lambda \circ \Gamma$, where $\Ad_x$ denotes conjugation by $x$.

\begin{lemma}\label{admissiblegrading}
 Let $M_{k}(\field{K})$ be a real form of $M_n(\field{C})$, and let $\alpha \in \Aut(M_n(\field{C}))$. Then $\alpha$ is an admissible grading on $M_{k}(\field{K})$ if and only if there exists a self-adjoint unitary $\gamma$ in $M_{k}(\field{K})$ or $i M_{k}(\field{K})$, such that $\alpha = \Ad_\gamma$.
\end{lemma}

\begin{proof}
Suppose that $\alpha$ is an admissible grading. Let $\field{K}_0$ be $\field{C}$ if $\field{K}=\field{C}$, and $\field{R}$ otherwise. Then $M_k(\field{K})$ is central simple over $\field{K}_0$, so that there exists some invertible element $S$ of $M_k(\field{K})$ such that $\alpha = \Ad_S$. Since $\alpha$ respects the involution, for any $A \in M_k(\field{K})$ we must have
\[
 (S^{-1})^*A^*S^* = (SAS^{-1})^* = \alpha(A)^* = \alpha(A^*) = SA^*S^{-1},
\]
\ie $[A,S^*S]=0$, so that $S^*S$ is a positive central element of $M_k(\field{K})$, and hence $S^*S = c1$ for some $c > 0$. Thus, $U = c^{-1/2}S$ is a unitary element of $M_k(\field{K})$ such that $\alpha = \Ad_U$. Now, recall that $\alpha^2 = \Id$, so that $\Ad_{U^2} = \Id$, and hence $U^2 = \zeta 1$ for some $\zeta \in \mathbb{T} \cap \field{K}_0$. If $\field{K} = \field{C}$, then one can simply set $\gamma = \overline{\lambda}U$ for $\lambda$ is a square root of $\zeta$. Otherwise, $U^2 = \pm 1$, so that if $U^2 = 1$, set $\gamma = U \in M_k(\field{K})$, and if $U^2 = -1$, set $\gamma = iU \in i M_k(\field{K})$.

On the other hand, if $\gamma$ is a self-adjoint unitary in either $M_k(\field{K})$ or $i M_k(\field{K})$, then $\Ad_\gamma$ is readily seen to be an admissible grading on $M_k(\field{K})$. 
\end{proof}

Let us now give the classification of compatible gradings for a type A irreducible triplet; it is essentially a generalisation of \cite{CC08b}*{Lemma 3.1}.

\begin{proposition}
 Let $(\alg{A},\hs{H},J)$ be a type A irreducible triplet of odd $KO$-dimen\-sion $n \bmod 8$, so that $\alg{A}$ is a real form $M_k(\field{K})$ of $M_n(\field{C})$ for some $n$, and let $\gamma$ be a grading on $\hs{H}$ as a Hilbert space.  Then $\gamma$ is compatible if and only if there exists a self-adjoint unitary $g$ in $M_k(\field{K})$ or $i M_k(\field{K})$ such that
\begin{equation}
 \gamma = \pm g \otimes 1_{m_k} \otimes g^T,
\end{equation}
in which case $\gamma$ necessarily commutes with $J$.
\end{proposition}

\begin{proof}
Let $m_n = 2^{(1-\eps)/2}$. Then $\hs{H} = \field{C}^n \otimes \field{C}^{m_n} \otimes \field{C}^n$, and for all $a \in \alg{A}$,
\[
 \lambda(a) = \lambda_\alpha(a) \otimes 1_{m_n} \otimes 1_n = a \otimes 1_{m_nk} \otimes 1_n, \quad \rho(a) = 1_n \otimes 1_{m_n} \otimes \lambda_\alpha(a)^T = 1_n \otimes 1_{m_n} \otimes a^T.
\]

Suppose that $\gamma$ is compatible. Then by Lemma~\ref{admissiblegrading} there exists some self-adjoint unitary $g$ in either $M_k(\field{K})$ or $i M_k(\field{K})$ such that for all $a \in \alg{A}$,
\[
 \gamma(a \otimes 1_{m_n} \otimes 1_n)\gamma = (gag) \otimes 1_{m_k} \otimes 1_n.
\]
Now, let $\gamma_0 = g \otimes 1_{m_n} \otimes g^T$. Then, by construction, $\gamma_0$ is a compatible grading for $(\alg{A},\hs{H},J)$ that induces the same admissible grading on $\alg{A}$ as $\gamma$, and moreover commutes with $J$. Then $\nu := \gamma \gamma_0 \in \lrU_\alg{A}(\hs{H};J)$, so that $\nu = 1_n \otimes \nu_{\rep{n}\rep{n}} \otimes 1_n$ for some
\[
 \nu_{\rep{n}\rep{n}} \in
\begin{cases}
 \{\pm1\}, &\text{if $n = 1$ or $7 \bmod 8$,}\\
 \SU(2), &\text{if $k = 3$ or $5 \bmod 8$.}
\end{cases}
\]
Thus $\gamma = g \otimes \nu_{\rep{n}\rep{n}} \otimes g^T$, and hence, since $\gamma$ is self-adjoint, $\nu_{\rep{n}\rep{n}}$ must also be self-adjoint. Therefore $\nu_{\alpha\alpha} = \pm 1_{m_k}$, or equivalently, $\gamma^\prime = \pm \gamma$.

On the other hand, if $g$ is a self-adjoint unitary in either $M_k(\field{K})$ or $i M_k(\field{K})$, then $\gamma = g \otimes 1_{m_k} \otimes g^T$ is certainly a compatible grading that commutes with $J$.
\end{proof}

Thus, irreducible triplets can only give rise to real $\alg{A}^\even$-bimodules of $KO$-dimension $0$ or $4 \bmod 8$.

Let us now turn to the type B case.

\begin{proposition}
 Let $(\alg{A},\hs{H},J)$ be a type B irreducible triplet of odd $KO$-dimen\-sion $n \bmod 8$, so that for some $n \in \semiring{N}$, $\alg{A} = M_{k_1}(\field{K}_1) \oplus M_{k_2}(\field{K}_2)$ for real forms $M_{k_1}(\field{K}_1)$ and $M_{k_2}(\field{K}_2)$ of $M_n(\field{C})$, and let $\gamma$ be a grading on $\hs{H}$ as a Hilbert space. Then $\gamma$ is compatible if and only if one of the following holds:
\begin{enumerate}
 \item There exist gradings $\gamma_1$ and $\gamma_2$ on $\field{C}^n$, with $\gamma_j \in M_{k_j}(\field{K}_j)$ or $i M_{k_j}(\field{K}_j)$, such that
 \begin{equation}
  \gamma =
  \begin{pmatrix}
   \gamma_1 \otimes \gamma_2^T & 0\\
   0 & \epspp \gamma_2 \otimes \gamma_1^T
  \end{pmatrix},
 \end{equation}
 in which case $\gamma J = \epspp J \gamma$, and if $\gamma^\prime$ is any other compatible grading, $\Ad_{\gamma^\prime} = \Ad_\gamma$ if and only if $\gamma^\prime = \pm \gamma$.
 \item One has that $\field{K}_1 = \field{K}_2 = \field{K}$ and $k_1 = k_2 = k$, and there exist a unitary $u \in M_k(\field{K})$ and $\eta \in \mathbb{T}$ such that
 \begin{equation}
  \gamma =
  \begin{pmatrix}
   0 & \overline{\eta} u^* \otimes \overline{u}\\
   \eta u \otimes u^T & 0
  \end{pmatrix},
 \end{equation}
 in which case $\gamma$ necessarily commutes with $J$, and if $\gamma^\prime$ is any other compatible grading, $\Ad_{\gamma^\prime} = \Ad_\gamma$ if and only if $\gamma^\prime = (\zeta 1_{n^2} \oplus \overline{\zeta} 1_{n^2})\gamma$ for some $\zeta \in \mathbb{T}$.
\end{enumerate}
\end{proposition}

\begin{proof}
Let $\gamma$ be a compatible grading. Then, with respect to the decomposition $\hs{H} = (\field{C}^n \otimes \field{C}^n) \oplus (\field{C}^n \oplus \field{C}^n)$, let us write
\[
 \gamma = 
 \begin{pmatrix}
  A & B\\
  C & D
 \end{pmatrix}
\]
for $A$, $B$, $C$ and $D \in M_n(\field{C})\otimes M_n(\field{C})$. Applying self-adjointness of $\gamma$, we find that $A$ and $D$ must be self-adjoint, and that $B = C^*$, and then applying the fact that $\gamma^2 = 1$, we find that
\[
 A^2 + C^*C = 1, \qquad CA + DC = 0, \qquad CC^* + D^2 = 1.
\]
Finally, applying the condition that $\gamma$ commutes or anticommutes with $J$, \ie that $\gamma J = \epspp J \gamma$ for $\epspp = \pm 1$, we find that
\[
 D = \epspp X A X, \qquad C^* = \epspp X C X,
\]
where $X$ is the antiunitary on $\field{C}^n \otimes \field{C}^n$ given by $X : \xi_1 \otimes \xi_2 \mapsto \overline{\xi_2} \otimes \overline{\xi_1}$.

Now, since $\gamma$ is compatible, and since $(1,0)$ and $(0,1)$ are projections in $\alg{A}$ satisfying $(1,0) + (0,1) = 1$, there exist projections $P$ and $Q$ in $\alg{A}$ such that
\[
 \Ad_\gamma \lambda(1,0) = \lambda(P,1-Q), \qquad \Ad_\gamma \lambda(0,1) = \lambda(1-P,Q),
\]
that is,
\[
 \begin{pmatrix}
  P \otimes 1_n & 0\\
  0 & (1-Q) \otimes 1_n
 \end{pmatrix} =
 \gamma
 \begin{pmatrix}
  1 & 0\\
  0 & 0
 \end{pmatrix}
 \gamma =
 \begin{pmatrix}
  A^2 & AC^*\\
  CA & CC^*
 \end{pmatrix}
\]
and
\[
 \begin{pmatrix}
  (1-P) \otimes 1_n & 0\\
  0 & Q \otimes 1_n
 \end{pmatrix} =
 \gamma
 \begin{pmatrix}
  0 & 0\\
  0 & 1
 \end{pmatrix}
 \gamma =
 \begin{pmatrix}
  C^*C & C^*D\\
  DC & D^2
 \end{pmatrix}.
\]
Thus, $A$ is a self-adjoint partial isometry with support and range projection $P \otimes 1_n$, $D$ is a self-adjoint partial isometry with support and range projection $Q \otimes 1_n$, and $C$ is a partial isometry with support projection $(1-P)\otimes 1_n$ and range projection $(1-Q)\otimes 1_n$.

Now, recalling that $D = \epspp X A X$, we see that
\[
 Q \otimes 1_n = D^2 = X A^2 X = X P \otimes 1_n X = 1_n \otimes \overline{P}.
\]
If $Q = 0$, then certainly $P = 0$. Suppose instead that $Q \neq 0$, and let $\xi \in Q\field{C}^n \otimes \field{C}^n$ be non-zero. Then
\[
 \Id_{\xi \otimes \field{C}^n} = (Q \otimes 1_n)|_{\xi \otimes \field{C}^n} = (1 \otimes \overline{P})|_{\xi \otimes \field{C}^n},
\]
so that $P = 1$ and hence $Q = 1$ also. We therefore have two possible cases:
\begin{enumerate}
 \item We have
 \[
  \gamma =
  \begin{pmatrix}
   A & 0\\
   0 & \epspp X A X
  \end{pmatrix}
 \]
 for $A$ a grading on $\field{C}^n \otimes \field{C}^n$;
 \item We have
 \[
  \gamma =
  \begin{pmatrix}
   0 & C^*\\
   C & 0
  \end{pmatrix}
 \]
 for $C$ a unitary on $\field{C}^n \otimes \field{C}^n$ such that $C^* = (-1)^m X C X$.
\end{enumerate}
 
 First suppose that the first case holds. Then, on the one hand, $\Ad_A|_{M_n(\field{C}) \otimes 1_n}$ induces an admissible grading for $M_{k_1}(\field{K}_1)$, so that there exists a self-adjoint unitary $\gamma_1$ in either $M_{k_1}(\field{K}_1)$ or $i M_{k_1}(\field{K}_1)$ such that $\Ad_A|_{M_n(\field{C}) \otimes 1_n} = \Ad_{\gamma_1 \otimes 1_n}$, and on the other hand, $\Ad_{\epspp X A X}|_{M_n(\field{C}) \otimes 1_n}$ induces an admissible grading for $M_{k_2}(\field{K}_2)$, so that there exists a self-adjoint unitary $\gamma_2$ in $M_{k_2}(\field{K}_1)$ or $i M_{k_2}(\field{K}_1)$ such that $\Ad_{\epspp X A X}|_{M_n(\field{C}) \otimes 1_n} = \Ad_{\gamma_2 \otimes 1_n}$. Since for $a \otimes b \in M_n(\field{C}) \otimes M_n(\field{C})$ we can write
\[
 a \otimes b = (a \otimes 1_n) X (\overline{b} \otimes 1_n) X,
\]
it therefore follows that $\Ad_A = \Ad_{\gamma_1 \otimes \gamma_2^T}$ on the central simple algebra $M_n(\field{C}) \otimes M_n(\field{C}) \cong M_{n^2}(\field{C})$ over $\field{C}$. Hence, there exists some non-zero $\eta \in \field{C}$ such that $A = \eta \gamma_1 \otimes \gamma_2^T$, and since both $A$ and $\gamma_1 \otimes \gamma_2^T$ are self-adjoint and unitary, it follows that $\eta = \pm 1$. Absorbing $\pm 1$ into $\gamma_1$ or $\gamma_2$, we therefore find that
\[
 \gamma =
 \begin{pmatrix}
  \gamma_1 \otimes \gamma_2^T & 0\\
  0 & \epspp \gamma_2 \otimes \gamma_1^T
 \end{pmatrix}.
\]
On the other hand, $\gamma$ so constructed is readily seen to be a compatible grading satisfying $\gamma J = \epspp J \gamma$.

Now suppose that the second case holds. Then, since $\gamma$ is compatible, it is clear that the automorphisms $\alpha$, $\beta$ of $M_n(\field{C})$ specified by
\[
 \alpha(a) \otimes 1_n = C (a \otimes 1_n) C^*, \qquad \beta(a) \otimes 1_n = C^*(a \otimes 1_n) C,
\]
are inverses of each other, and that $\alpha$, in particular, induces an isomorphism $M_{k_1}(\field{K}_1) \to M_{k_2}(\field{K}_2)$, so that $\field{K}_1 = \field{K}_2 = \field{K}$ and $k_1 = k_2 = k$. Next, by the proof of Lemma~\ref{admissiblegrading}, there exists some unitary $u$ in $M_n(\field{C})_I$ such that $\alpha = \Ad_u$, from which it follows that $\beta = \Ad_{u^*}$. By the same trick as above, we then find that $\Ad_C = \Ad_{u \otimes u^T}$ on the central simple algebra $M_n(\field{C}) \otimes M_n(\field{C}) \cong M_{n^2}(\field{C})$ over $\field{C}$. Hence, there exists some non-zero $\eta \in \field{C}$ such that $C = \eta u \otimes u^T$, and since both $C$ and $u \otimes u^T$ are unitary, it follows that $\eta \in \mathbb{T}$. Thus,
\[
 \gamma =
 \begin{pmatrix}
  0 & \overline{\eta} u^* \otimes \overline{u}\\
  \eta u \otimes u^T & 0
 \end{pmatrix}.
\]
On the other hand, $\gamma$ so constructed is readily seen to be a compatible grading satisfying $[\gamma,J]=0$.

Finally, let $\gamma$ and $\gamma^\prime$ be two compatible gradings. Suppose that $\Ad_\gamma = \Ad_{\gamma^\prime}$, and set $U = \gamma^\prime \gamma$. Then, by construction, $U$ is a unitary element of $\lrbdd_\alg{A}(\hs{H};J)$, so that there exists some $\zeta \in \mathbb{T}$ such that
\[
 U = \zeta 1_{n^2} \oplus \overline{\zeta} 1_{n^2}.
\]
If the second case holds, then nothing more can be said, but if the first case holds, so that
\[
 \gamma =
 \begin{pmatrix}
  \gamma_1 \otimes \gamma_2^T & 0\\
  0 & \epspp \gamma_2 \otimes \gamma_1^T
 \end{pmatrix}
\]
for suitable $\gamma_1$ and $\gamma_2$, then
\[
 \gamma^\prime =
 \begin{pmatrix}
  \zeta \gamma_1 \otimes \gamma_2^T & 0\\
  0 & \epspp \overline{\zeta} \gamma_2 \otimes \gamma_1^T
 \end{pmatrix},
\]
so that by self-adjointness of $\gamma^\prime$, $\gamma_1$ and $\gamma_2$, we must have $\zeta = \pm 1$, as required.
\end{proof}

Thus, we can obtain a real bimodule of $KO$-dimension $6 \bmod 8$ only from a type B irreducible triplet together with a compatible grading satisfying the first case of the last result.

\subsection{Even subalgebras and even $KO$-dimensional bimodules}

We now consider real bimodules of $KO$-di\-men\-sion $6 \bmod 8$ obtained from irreducible triplets. Thus, let $(\alg{A},\hs{H},J)$ be a fixed type B irreducible triplet of $KO$-dimension $1$ or $7 \bmod 8$, and let $\gamma$ be a fixed compatible grading for $(\alg{A},\hs{H},J)$ anticommuting with $J$, so that for some $n \in \semiring{N}$,
\begin{itemize}
 \item $\alg{A} = M_{k_1}(\field{K}_1) \oplus M_{k_2}(\field{K}_2)$ for real forms $M_{k_j}(\field{K}_j)$ of $M_n(\field{C})$;
 \item $\mult[\hs{H}] = E_{\rep{n}_1 \rep{n}_2} + E_{\rep{n}_2 \rep{n}_1}$;
 \item There exist self-adjoint unitaries $\gamma_j \in M_{k_j}(\field{K}_j)$ or $i M_{k_j}(\field{K}_j)$ with signature $(r_j,n - r_j)$ such that
\[
 \gamma = 
\begin{pmatrix}
 \gamma_1 \otimes \gamma_2^T & 0\\
 0 & -\gamma_2 \otimes \gamma_1^T
\end{pmatrix}.
\]
\end{itemize}
It is worth noting that $(\hs{H},J)$ admits, up to sign, a unique $S^0$-real structure, given by $\epsilon = 1_{n^2} \oplus -1_{n^2}$, which certainly commutes with $\gamma$. We can exploit the symmetries present to simplify our discussion by taking, without loss of generality, $r_j > 0$, and requiring that $\gamma_1 \in i M_{k_1}(\field{K}_1)$ only if $\gamma_2 \in i M_{k_2}(\field{K}_2)$, and that $\gamma_1 = 1_n$ only if $\gamma_2 = 1_n$.

Our main goal in this section is to give an explicit description of $\alg{A}^\even$ and of $(\hs{H},\gamma,J)$ as a real $\alg{A}^\even$-bimodule. To do so, however, we first need the following:

\begin{lemma}
 Let $M_k(\field{K})$ be a real form of $M_n(\field{C})$, let $g$ be a self-adjoint unitary in $M_k(\field{K})$ or $i M_k(\field{K})$, and let $r = \nullity(g-1)$. Set $M_k(\field{K})^g := \{a \in M_k(\field{K}) \mid [a,g] = 0\}$.
\begin{itemize}
 \item If $g \in M_k(\field{K})$, then $M_k(\field{K})^g \cong M_{kr/n}(\field{K}) \oplus M_{k(n-r)/n}(\field{K})$;
 \item If $g \in i M_k(\field{K})$, then $r = n/2$ and
 \[
 	M_k(\field{K})^g \cong \{(a,b) \in M_{k/2}(\field{C})^2 \mid b = \overline{a}\} \cong M_{k/2}(\field{C}).
 \]
\end{itemize}
\end{lemma}

\begin{proof}
Let $P^+ := \frac{1}{2}(1 + g)$ and $P^- := \frac{1}{2}(1-g)$, which are thus projections in $M_n(\field{C})$ of rank $r$ and $n-r$, respectively. Define an injection $\phi : M_k(\field{K})^g \mapsto M_{r}(\field{C}) \oplus M_{n-r}(\field{C})$ by $\phi(A) := (P^\even A P^\even, P^\odd  A P^\odd)$.

First, suppose that $g \in M_k(\field{K})$. Then $P^+$ and $P^-$ are also in $M_k(\field{K})$, from which it immediately follows that $\phi(M_k(\field{K})^g) = M_{kr/n}(\field{K}) \oplus M_{k(n-r)/n}(\field{K})$.

Suppose instead that $g \in i M_k(\field{K})$ and  $\field{K} \neq \field{C}$. Then $M_k(\field{K}) = \{A \in M_n(\field{C}) \mid [A,I]=0\}$ for a suitable antiunitary $I$ on $\field{C}^n$ satisfying $I^2 = \alpha 1$, where $\alpha = 1$ if $\field{K} = \field{R}$ and $\alpha = -1$ if $\field{K} = \field{H}$. Then $\{g,I\} = 0$, and hence, with respect to the decomposition $\field{C}^n = P^+ \field{C}^n \oplus P^- \field{C}^n \cong \field{C}^r \oplus \field{C}^{n-r}$,
\[
 I =
\begin{pmatrix}
 0 & \alpha \tilde{I}^*\\
 \tilde{I} & 0
\end{pmatrix},
\]
where $\tilde{I} = P^\odd I P^\even$ is an antiunitary $\field{C}^r \mapsto \field{C}^{n-r}$. Thus, $n$ is even and $r = n/2$, and taking $\tilde{I}$, without loss of generality, to be complex conjugation on $\field{C}^r$, for all $A \in M_n(\field{C})$ commuting with $g$, $[A,I] = 0$ if and only if $P^- A P^- = \overline{P^+ A P^+}$, and hence $\phi(M_k(\field{K})^g) = \{(a,\overline{a}) \mid a \in M_{n/2}(\field{C})\} \cong M_{n/2}(\field{C})$.
\end{proof}

In light of the form of $\gamma$, this last Lemma immediately implies the aforementioned explicit description of $\alg{A}^\even$ and $(\hs{H},\gamma,J)$:

\begin{proposition}
 Let $(m^\even,m^\odd) = (m^\even,(m^\even)^T)$ be the pair of multiplicity matrices of $(\hs{H},\gamma,J)$ as an even $KO$-dimensional real $\alg{A}^\even$-bimodule. Let $r_i^\prime = n - r_i$, and, when $n$ is even, let $c = n/2$. Then:
\begin{enumerate}
 \item If $\gamma_1 \in i M_{k_1}(\field{K}_1)$, $\gamma_2 \in i M_{k_2}(\field{K})$, then
\begin{equation}
 \alg{A}^\even = M_{c}(\field{C}) \oplus M_{c}(\field{C}),
\end{equation}
and
\begin{equation}
 m^\even = E_{\rep{c}_1 \rep{c}_2} + E_{\crep{c}_1 \crep{c}_2} + E_{\rep{c}_2 \crep{c}_1} + E_{\crep{c}_2 \rep{c}_1};
\end{equation}
 \item If $\gamma_1 \in i M_{k_1}(\field{K}_1)$, $\gamma_2 \in M_{k_2}(\field{K}) \setminus \{1_n\}$, then
\begin{equation}
 \alg{A}^\even = M_{c}(\field{C}) \oplus M_{k_2 r_2/n}(\field{K}_2) \oplus M_{k_2 r_2^\prime /n}(\field{K}_2).
\end{equation}
and
\begin{equation}
 m^\even = E_{\rep{c} \rep{r}_2} + E_{\crep{c} \rep{r}_2^\prime} + E_{\rep{r}_2 \crep{c}} + E_{\rep{r}_2^\prime \rep{c}};
\end{equation}
 \item If $\gamma_1 \in i M_{k_1}(\field{K}_1)$, $\gamma_2 = 1$, then
\begin{equation}
 \alg{A}^\even = M_{c}(\field{C}) \oplus M_{k_2}(\field{K}_2),
\end{equation}
and
\begin{equation}
 m^\even = E_{\rep{c} \rep{n}} + E_{\rep{n} \crep{c}};
\end{equation}
 \item If $\gamma_1 \in M_{k_1}(\field{K}_1) \setminus \{1_n\}$, $\gamma_2 \in M_{k_2}(\field{K}_2) \setminus \{1_n\}$, then
\begin{equation}
 \alg{A}^\even = M_{k_1 r_1 /n}(\field{K}_1) \oplus M_{k_1 r_1^\prime /n}(\field{K}_1) \oplus M_{k_2 r_2 / n}(\field{K}_2) \oplus M_{k_2 r_2^\prime /n}(\field{K}_2),
\end{equation}
and
\begin{equation}
 m^\even = E_{\rep{r}_1 \rep{r}_2} + E_{\rep{r}_1^\prime \rep{r}_2^\prime} + E_{\rep{r}_2 \rep{r}_1^\prime} + E_{\rep{r}_2^\prime \rep{r}_1};
\end{equation}
 \item If $\gamma_1 \in M_{k_1}(\field{K}_1) \setminus \{1_n\}$, $\gamma_2 = 1_n$, then
\begin{equation}
 \alg{A}^\even = M_{k_1 r_1 /n}(\field{K}_1) \oplus M_{k_1 r_1^\prime /n}(\field{K}_1) \oplus M_{k_2}(\field{K}_2),
\end{equation}
and
\begin{equation}
 m^\even = E_{\rep{r}_1 \rep{n}} + E_{\rep{n} \rep{r}_1^\prime};
\end{equation}
 item If $\gamma_1 = \gamma_2 = 1_n$, then
\begin{equation}
 \alg{A}^\even = M_{k_1}(\field{K}_1) \oplus M_{k_2}(\field{K}_2),
\end{equation}
and
\begin{equation}
 m^\even = E_{\rep{n}_1 \rep{n}_2}.
\end{equation}
\end{enumerate}
\end{proposition}

One can check in each case that $(\hs{H},\gamma)$ is quasi-orientable as an even $\alg{A}^\even$-bimodule. However, Propositions~\ref{converseorientable} and~\ref{intform} immediately imply the following:

\begin{corollary}
 The following are equivalent for $(\hs{H},\gamma)$ as an even $\alg{A}^\even$-bimodule:
\begin{enumerate}
 \item $\gamma_1 \in M_{k_1}(\field{K}_1)$ and $\gamma_2 \in M_{k_2}(\field{K}_2)$;
 \item $(\hs{H},\gamma)$ is orientable;
 \item $(\hs{H},\gamma)$ has non-vanishing intersection form;
 \item $(\hs{H},\gamma)$ is complex-linear.
\end{enumerate}
\end{corollary}

This then motivates us to restrict ourselves to the case where $\gamma_1 \in M_{k_1}(\field{K}_1)$ and $\gamma_2 \in M_{k_2}(\field{K}_2)$. Note, however, that in no case is Poincar{\'e} duality possible.

\subsection{Off-diagonal Dirac operators}

Let us now consider the slightly more general $S^0$-real $\alg{A}^\even$-bimodule $(\hs{H}_F,\gamma_F,J_F,\epsilon_F)$ of $KO$-dimension $6 \bmod 8$ given by taking the direct sum of $N$ copies of $(\hs{H},\gamma,J,\epsilon)$, where $N \in \semiring{N}$. If we modify our earlier conventions slightly to allow for the summand $0$ in Wedderburn decompositions, we can therefore write
\begin{equation}
 \alg{A}^\even = M_{k_1 r_1 /n}(\field{K}_1) \oplus M_{k_1 r_1^\prime /n}(\field{K}_1) \oplus M_{k_2 r_2 / n}(\field{K}_2) \oplus M_{k_2 r_2^\prime}(\field{K}_2),
\end{equation}
so that $(\hs{H}_F,\gamma_F,J_F)$ is the real $\alg{A}^\even$-bimodule of $KO$-dimension $6 \bmod 8$ with signed multiplicity matrix
\begin{equation}
 \mu_F = N(E_{\rep{r}_1 \rep{r}_2} - E_{\rep{r}_1 \rep{r}_2^\prime} - E_{\rep{r}_1^\prime \rep{r}_2} + E_{\rep{r}_1^\prime \rep{r}_2^\prime} - E_{\rep{r}_2 \rep{r}_1} + E_{\rep{r}_2 \rep{r}_1^\prime} + E_{\rep{r}_2^\prime \rep{r}_1} - E_{\rep{r}_2^\prime \rep{r}_1^\prime}),
\end{equation}
whilst $(\hs{H}_f,\gamma_f) := ((\hs{H}_F)_i,(\gamma_F)_i)$ is the even $\alg{A}^\even$-bimodule with signed multiplicity matrix
\begin{equation}
 \mu_f = N(E_{\rep{r}_1 \rep{r}_2} - E_{\rep{r}_1 \rep{r}_2^\prime} - E_{\rep{r}_1^\prime \rep{r}_2} + E_{\rep{r}_1^\prime \rep{r}_2^\prime}).
\end{equation}
It then follows also that $(\hs{H}_{\overline{f}},\gamma_{\overline{f}}) := (J_F \hs{H}_f, - (J_F \gamma_f J_F)|_{J_F \hs{H}_f})$ is the even $\alg{A}^\even$-bimodule with signed multiplicty matrix
\[
 \mu_{\overline{f}} = -\mu_f^T = N(- E_{\rep{r}_2 \rep{r}_1} + E_{\rep{r}_2 \rep{r}_1^\prime} + E_{\rep{r}_2^\prime \rep{r}_1} - E_{\rep{r}_2^\prime \rep{r}_1^\prime}).
\]

Now, for $\alg{C}$ a unital \Star-subalgebra of $\alg{A}^\even$, let us call a Dirac operator $D \in \ms{D}_0(\alg{C},\hs{H}_F,\gamma_F,J_F)$ \term{off-diagonal} if it does not commute with $\epsilon_F$, or equivalently~\cite{CC08b}*{\S 4} if $[D,\mathcal{Z}(\alg{A})] \neq \{0\}$. If $\ms{D}_1(\alg{C},\hs{H}_F,\gamma_F,J_F,\epsilon_F) \subseteq \ms{D}_0(\alg{C},\hs{H}_F,\gamma_F,J_F)$ is the subspace consisting of Dirac operators anti-commuting with $\epsilon_F$, then, in fact,
\begin{equation}
 \ms{D}_0(\alg{C},\hs{H}_F,\gamma_F,J_F) = \ms{D}_0(\alg{C},\hs{H}_F,\gamma_F,J_F,\epsilon_F) \oplus \ms{D}_1(\alg{C},\hs{H}_F,\gamma_F,J_F,\epsilon_F),
\end{equation}
as can be seen from writing 
\[
 D = \frac{1}{2}\{D,\epsilon_F\}\epsilon_F + \frac{1}{2}[D,\epsilon_F]\epsilon_F
\]
for $D \in \ms{D}_0(\alg{C},\hs{H}_F,\gamma_F,J_F)$. Thus, non-zero off-diagonal Dirac operators exist for $(\hs{H}_F,\gamma_F,J_F,\epsilon_F)$ as an $S^0$-real $\alg{C}$-bimodule if and only if 
\[
 \ms{D}_1(\alg{C},\hs{H}_F,\gamma_F,J_F,\epsilon_F) \neq \{0\}.
\]
Our goal is to generalise Theorem 4.1 in~\cite{CC08b}*{\S 4} and characterise subalgebras of $\alg{A}^\even$ of maximal dimension admitting off-diagonal Dirac operators.

The following result is the first step in this direction:

\begin{proposition}[\cite{CC08b}*{Lemma 4.2}]
 A unital \Star-subalgebra $\alg{C} \subseteq \alg{A}^\even$ admits off-diagonal Dirac operators if and only if there exists some partial unitary $T \in \bdd(\field{C}^{r_1} \oplus \field{C}^{r_1^\prime} \oplus \field{C}^{r_2} \oplus \field{C}^{r_2^\prime})$ with support contained in one of $\field{C}^{r_1}$ or $\field{C}^{r_1^\prime}$ and range contained in one of $\field{C}^{r_2}$ or $\field{C}^{r_2^\prime}$, such that
\[
 \alg{C} \subseteq \alg{A}(T) := \{a \in \alg{A}^\even \mid [a,T] = [a^*,T] = 0\}.
\]
\end{proposition}

\begin{proof}
 First note that the map $\ms{D}_1(\alg{C},\hs{H}_F,\gamma_F,J_F,\epsilon_F) \to \bdd_\alg{C}^1(\hs{H}_f,\hs{H}_{\overline{f}})$ given by $D \mapsto P_{-i} D P_i$ is an isomorphism, so that $\alg{C}$ admits off-diagonal Dirac operators if and only if $\bdd_\alg{C}^1(\hs{H}_f,\hs{H}_{\overline{f}}) \neq \{0\}$. Since a map $S \in \bdd(\hs{H}_f,\hs{H}_{\overline{f}})$ satisfies the generalised order one condition for $\alg{C}$ if and only if $\rho_{\overline{f}}(c) S - S \rho_f(C)$ is left $\alg{C}$-linear for all $c \in \alg{C}$, $\alg{C}$ admits off-diagonal Dirac operators only if
\[
 \{S \in \lbdd_\alg{C}(\hs{H}_f,\hs{H}_{\overline{f}}) \mid -\gamma_{\overline{f}}S = S\gamma_f\} \neq \{0\},
\]
or equivalently,
\[
 \alg{C} \subseteq \alg{A}_S := \{a \in \alg{C}^\even \mid \lambda_{\overline{f}}(a)S = S\lambda(a), \; \lambda_{\overline{f}}(a^*)S = S\lambda(a^*)\}
\]
for some non-zero $S \in \bdd(\hs{H}_f,\hs{H}_{\overline{f}})$ such that $-\gamma_{\overline{f}}S = S\gamma_f$. 

Now, let $S \in \bdd(\hs{H}_f,\hs{H}_{\overline{f}})$ be non-zero and such that $-\gamma_{\overline{f}}S = S\gamma_f$. Then, the support of $S$ must have non-zero intersection with one of $(\hs{H}_F)_{\rep{r}_1\rep{r}_2}$ or $(\hs{H}_F)_{\rep{r}_1^\prime \rep{r}_2^\prime}$, and the range of $S$ must have non-zero intersection with one of $(\hs{H}_F)_{\rep{r}_2\rep{r}_1}$ or $(\hs{H}_F)_{\rep{r}_2^\prime \rep{r}_1^\prime}$. Thus, $S_{\alpha\beta}^{\gamma\delta} \neq 0$ for some $(\alpha,\beta) \in \{(\rep{r}_1,\rep{r}_2),(\rep{r}_1^\prime,\rep{r}_2^\prime)\}$ and $(\gamma,\delta) \in \{(\rep{r}_2,\rep{r}_1),(\rep{r}_2^\prime, \rep{r}_1^\prime)\}$, so that $\alg{A}_S \subseteq \alg{A}_{S_{\alpha\beta}^{\gamma\delta}}$. Let us now write
\[
 S_{\alpha\beta}^{\gamma\delta} = \sum_{i} A_i \otimes B_i
\]
for non-zero $A_i \in M_{n_\gamma \times n_\alpha}(\field{C})$ and for linearly independent $B_i \in M_{N n_\delta \times N n_\beta}(\field{C})$. Then for all $a \in \alg{A}^\even$,
\[
 \lambda_{\overline{f}}(a)S - S\lambda_f(a) = \sum_i \bigl(\lambda_\gamma(a) A_i - A_i \lambda_\alpha(a)\bigr) \otimes B_i,
\]
so by linear independence of the $B_i$, $a \in \alg{A}_{S_{\alpha\beta}^{\gamma\delta}}$ if and only if for each $i$,
\[
 \lambda_\gamma(a) A_i = A_i \lambda_\alpha(a), \quad \lambda_\gamma(a^*) A_i = A_i \lambda_\alpha(a),
\]
and hence
\[
 \alg{A}(S) \subseteq \alg{A}_{S_{\alpha\beta}^{\gamma\delta}} \subseteq \alg{A}(T_0) := \{a \in \alg{A}^\even \mid [a,T_0] = 0, [a^*,T_1]=0\}
\]
for $T_0 = A_1$, say, viewing $T_0$ and the elements of $\alg{A}^\even$ as operators on $\field{C}^{r_1} \oplus \field{C}^{r_1^\prime} \oplus \field{C}^{r_2} \oplus \field{C}^{r_2^\prime}$. However, if $T_0 = P T$ is the polar decomposition of $T_0$ into a positive operator $P$ on $\field{C}^{n_\gamma}$ and a partial isometry $T : \field{C}^{n_\alpha} \to \field{C}^{n_\gamma}$, it follows that $a \in \alg{A}^\even$ commutes with $T_0$ only if it commutes with $T$, and hence $\alg{A}_0 \subseteq \alg{A}(T_0) \subseteq \alg{A}(T)$, proving the one direction of the claim.

Now suppose that $\alg{C} = \alg{A}(T)$ for a suitable partial isometry $T$, which we view as a partial isometry $\field{C}^{n_{\alpha_0}} \to \field{C}^{n_{\gamma_0}}$ for some $\alpha_0 \in \{\rep{r}_1,\rep{r}_1^\prime\}$, $\gamma_0 \in \{\rep{r}_2,\rep{r}_2^\prime\}$. Then for any non-zero $\Upsilon \in M_N(\field{C})$, we can define an element $S(\Upsilon) \in \lrbdd_\alg{C}(\hs{H}_f,\hs{H}_{\overline{f}})$ by setting
\[
 S(\Upsilon)_{\alpha\beta}^{\gamma\delta} =
\begin{cases}
 T \otimes \Upsilon \otimes T^* &\text{if $\alpha = \delta = \alpha_0$, $\beta = \gamma = \gamma_0$,}\\
 0 &\text{otherwise,}
\end{cases}
\]
which, as noted above, corresponds to a unique non-zero element of the space $\ms{D}_1(\alg{C},\hs{H}_F,\gamma_F,J_F,\epsilon_F)$, so that $\alg{C}$ does indeed admit off-diagonal Dirac operators.
\end{proof}

In light of the above characterisation, it suffices to consider subalgebras $\alg{A}(T)$ for partial isometries $T : \field{C}^{r_1} \to \field{C}^{r_2}$, so that
\begin{align}
 \alg{A}(T) &= \{(a_1,a_2,b_1,b_2) \in \alg{A}^\even \mid b_1 T = T a_1, \; b_1^* T = T a_1^*\}\\
 &\cong \alg{A}_0(T) \oplus M_{k_1 r_1^\prime / n}(\field{K}_1) \oplus M_{k_2 r_2^\prime /n}(\field{K}_k),
\end{align}
where
\begin{equation}
 \alg{A}_0(T) := \{(a,b) \in M_{k_1 r_1/n}(\field{K}_1) \oplus M_{k_2 r_2 /n}(\field{K}_2) \mid b T = T a, \; b^* T = T a^*\},
\end{equation}
so that our problem is reduced to that of maximising the dimension of $\alg{A}_0(T)$.

It is reasonable to assume that $T$ is, in some sense, compatible with the algebraic structures of $M_{k_1 r_1/n}(\field{K}_1)$ and $M_{k_2 r_2/n}(\field{K}_2)$, so as to minimise the restrictiveness of the defining condition on $\alg{A}_0(T)$, and hence maximise the dimension of $\alg{A}_0(T)$. It turns out that this notion of compatibility takes the form of the following conditions on $T$:
\begin{enumerate}
 \item The subspace $\supp(T)$ of $\field{C}^{r_1}$ is either a $\field{K}_1$-linear subspace of $\field{C}^{r_1} = \field{K}_1^{k_1 r_1/n}$ or, if $\field{K}_1 = \field{H}$, $\supp(T) = E \oplus \field{C}$ for $E$ an $\field{H}$-linear subspace of $\field{C}^{r_1} = \field{H}^{r_1/2}$;
 \item The subspace $\im(T)$ of $\field{C}^{r_2}$ is either a $\field{K}_2$-linear subspace of $\field{C}^{r_2} = \field{K}_1^{k_2 r_2/n}$ or, if $\field{K}_2 = \field{H}$, $\im(T) = E \oplus \field{C}$ for $E$ an $\field{H}$-linear subspace of $\field{C}^{r_2} = \field{H}^{r_2/2}$.
\end{enumerate}
Now, let $r = \rank(T)$, let $d(r) = \dim_{\field{R}}(\alg{A}_0(T))$, and let
\[
 d_i = 
\begin{cases}
 1 &\text{if $\field{K}_i = \field{R}$,}\\
 2 &\text{if $\field{K}_i = \field{C}$,}\\
 \frac{1}{2} &\text{if $\field{K}_i = \field{H}$.}
\end{cases}
\]
Under these assumptions, then, one can show that
\begin{enumerate}
 \item If $\field{K}_1 = \field{K}_2$ or $\field{K}_2 = \field{C}$, and, if $\field{K}_1 = \field{H}$, $r$ is even, then
\begin{equation}
 \alg{A}_0(T) \cong M_{k_1 r /n}(\field{K}_1) \oplus M_{k_1 (r_1 - r)/n}(\field{K}_1) \oplus M_{k_2 (r_2 - r)/n}(\field{K}_2),
\end{equation}
and hence
\[
 d(r) = d_1 r^2 + d_1 (r - r_1)^2 + d_2 (r - r_2)^2;
\]
 \item If $(\field{K}_1,\field{K}_2) = (\field{H},\field{R})$ and $r$ is odd, then
\begin{equation}
 \alg{A}_0(T) \cong \bigl(M_{(r-1)/2}(\field{H}) \cap M_{r-1}(\field{R})\bigr) \oplus \field{R} \oplus M_{(r_2-r-1)/2}(\field{H}) \oplus M_{r_1 - r}(\field{R}),
\end{equation}
and hence
\[
 d(r) = (r-1)^2 + 1 + \frac{1}{2}(r - r_2 + 1)^2 + (r - r_1)^2;
\]
 \item If $(\field{K}_1,\field{K}_2) = (\field{H},\field{C})$ and $r$ is odd, then
\begin{equation}
 \alg{A}_0(T) \cong M_{(r-1)/2}(\field{H}) \oplus \field{C} \oplus M_{(r_2 - r - 1)/2}(\field{H}) \oplus M_{r_1 - r}(\field{C}),
\end{equation}
and hence
\[
 d(r) = \frac{1}{2}(r-1)^2 + 2 + \frac{1}{2}(r-r_2+1)^2 + 2(r-r_1)^2;
\]
 \item If $\field{K}_1 = \field{K}_2 = \field{H}$ and $r$ is odd, then
\begin{equation}
 \alg{A}_0(T) \cong M_{(r-1)/2}(\field{H}) \oplus \field{C} \oplus M_{(r_1 - r - 1)/2}(\field{H}) \oplus M_{(r_2-r-1)/2}(\field{H}),
\end{equation}
and hence
\begin{equation}
 d(r) = \frac{1}{2}(r-1)^2 + 2 + \frac{1}{2}(r-r_1+1)^2 + \frac{1}{2}(r-r_2+1)^2.
\end{equation}
\end{enumerate}
The other cases are obtained easily, by symmetry, from the ones listed above.

Now, let $R_{\textrm max}$ be the set of all $r \in \{1, \dotsc, \min(r_1,r_2)\}$ maximising the value of $d(r)$. By checking case by case, one can arrive at the following generalisation of Theorem 4.1 in~\cite{CC08b}:

\begin{proposition}
 Let $T : \field{C}^{r_1} \to \field{C}^{r_2}$ be a partial isometry. Then $\alg{A}(T)$ attains maximal dimension only if $\rank(T) \in R_{\textrm max}$, where $R_{\textrm max} = \{1\}$ except in the following cases:
\begin{enumerate}
 \item $(\field{K}_1,\field{K}_2)=(\field{C},\field{C})$ and $(r_1,r_2)=(2,2)$, in which case $R_{\textrm max} = \{2\}$;
 \item $(\field{K}_1,\field{K}_2)=(\field{C},\field{C})$ and $(r_1,r_2)=(3,3)$, in which case $R_{\textrm max} = \{1,2\}$;
 \item $(\field{K}_1,\field{K}_2)=(\field{C},\field{R})$ and $(r_1,r_2)=(2,2)$, in which case $R_{\textrm max} = \{1,2\}$;
 \item $(\field{K}_1,\field{K}_2)=(\field{C},\field{H})$ and $(r_1,r_2)=(2,2)$, in which case $R_{\textrm max} = \{1,2\}$;
 \item $(\field{K}_1,\field{K}_2)=(\field{R},\field{C})$ and $(r_1,r_2)=(2,2)$, in which case $R_{\textrm max} = \{1,2\}$;
 \item $(\field{K}_1,\field{K}_2)=(\field{R},\field{R})$ and $(r_1,r_2)=(2,2)$, in which case $R_{\textrm max} = \{2\}$;
 \item $(\field{K}_1,\field{K}_2)=(\field{R},\field{R})$ and $(r_1,r_2)=(3,3)$, in which case $R_{\textrm max} = \{1,2\}$;
 \item $(\field{K}_1,\field{K}_2)=(\field{R},\field{H})$ and $r_1=2$, in which case $R_{\textrm max} = \{1,2\}$;
 \item $(\field{K}_1,\field{K}_2)=(\field{H},\field{C})$ and $(r_1,r_2)=(2,2)$, in which case $R_{\textrm max} = \{1,2\}$;
 \item $(\field{K}_1,\field{K}_2)=(\field{H},\field{R})$ and $r_2=2$, in which case $R_{\textrm max} = \{1,2\}$;
 \item $(\field{K}_1,\field{K}_2)=(\field{H},\field{H})$ and $(r_1,r_2)=(4,4)$, in which case $R_{\textrm max} = \{4\}$;
 \item $(\field{K}_1,\field{K}_2)=(\field{H},\field{H})$ and $(r_1,r_2)\neq(4,4)$, in which case $R_{\textrm max} = \{2\}$.
\end{enumerate}
Moreover, if $T$ satisfies the aforementioned compatibility conditions, then $\alg{A}(T)$ does indeed attain maximal dimension whenever $\rank(T) \in R_{\textrm max}$.
\end{proposition}

One must carry out the same calculations for the other possibilities for the domain and range of $T$, but this can be done simply by replacing $(r_1,r_2)$ in the above equations and claims with $(r_1,r_2^\prime)$, $(r_1^\prime,r_2)$ and $(r_1^\prime,r_2^\prime)$. Thus, one can determine the maximal dimension of a subalgebra of $\alg{A}^\even$ admitting off-diagonal operators by comparing the maximal values of $\dim_{\field{R}}(\alg{A}(T))$ for $T : \field{C}^{r_1} \to \field{C}^{r_2}$, $T : \field{C}^{r_1} \to \field{C}^{r_2^\prime}$, $T : \field{C}^{r_1^\prime} \to \field{C}^{r_2}$, and $T : \field{C}^{r_1^\prime} \to \field{C}^{r_2^\prime}$.

Finally, by means of the discussion above and the fact that $\Sp(n)$ acts transitively on $1$-dimensional subspaces of $\field{C}^n$, one can readily check that the real \Cstar-algebra $\alg{A}_F$ and the $S^0$-real $\alg{A}_F$-bimodule $(\hs{H}_F,\gamma_F,J_F,\epsilon_F)$ of $KO$-dimension $6 \bmod 8$ of the NCG Standard Model are uniquely determined, up to inner automorphisms of $\alg{A}^\even$ and unitary equivalence, by the following choice of inputs:
\begin{itemize}
 \item $n = 4$;
 \item $(\field{K}_1,\field{K}_2) = (\field{H},\field{C})$;
 \item $g_1 \in M_2(\field{H})$, $g_2 \in M_4(\field{C})$;
 \item $(r_1,r_2) = (2,4)$;
 \item $N = 3$.
\end{itemize}
The value of $N$, by construction, corresponds to the number of generations of fermions, whilst the values of $n$, $r_1$ and $r_2$ give rise to the number of species of fermion of each chirality per generation. The significance of the other inputs remains to be seen.

\section{Conclusion}

As we have seen, the structure theory first developed by Paschke and Sitarz~\cite{PS98} and by Krajewski~\cite{Kraj98} for finite real spectral triples of $KO$-dimension $0 \bmod 8$ and satisfying orientability and Poincar{\'e} duality can be extended quite fully to the case of arbitrary $KO$-dimension and without the assumptions of orientability and Poincar{\'e} duality. In particular, once a suitable ordering is fixed on the spectrum of a finite-dimensional real \Cstar-algebra $\alg{A}$, the study of finite real spectral triples with algebra $\alg{A}$ reduces completely to the study of the appropriate multiplicity matrices and of certain moduli spaces constructed using those matrices. This reduction is what has allowed for the success of Krajewski's diagrammatic approach~\cite{Kraj98}*{\S 4} in the cases dealt with by Iochum, Jureit, Sch{\"u}cker, and Stephan~\cites{ACG1,ACG2,ACG3,ACG4,ACG5,JS08,Sch05}. We have also seen how to apply this theory both to the ``finite geometries'' of the current version of the NCG Standard Model~\cites{Connes06,CCM07,CM08} and to Chamseddine and Connes's framework~\cites{CC08a,CC08b} for deriving the same finite geometries.

Dropping the orientability requirement comes at a fairly steep cost, as even bimodules of various sorts generally have fairly intricate moduli spaces of Dirac operators. It would therefore be useful to characterise the precise nature of the failure of orientability (and of Poincar{\'e} duality) for the finite spectral triple of the current noncommutative-geometric Standard Model. It would also be useful to generalise and study the physically-desirable conditions identified in the extant literature on finite spectral triples, such as dynamical non-degeneracy~\cite{Sch05} and anomaly cancellation~\cite{Kraj98}. Indeed, it would be natural to generalise Krajewski diagrams~\cite{Kraj98} and the combinatorial analysis they facilitate~\cite{JS08} to bilateral spectral triples of all types. The paper by Paschke and Sitarz~\cite{PS98} also contains further material for generalisation, namely discussion of the noncommutative differential calculus of a finite spectral triple and of quantum group symmetries. In particular, one might hope to characterise finite spectral triples equivariant under the action or coaction of a suitable Hopf algebra~\cites{PS00,Sit03}.

Finally, as was mentioned earlier, the finite geometry of the current NCG Standard Model fails to be $S^0$-real. However, this failure is specifically the failure of the Dirac operator $D$ to commute with the $S^0$-real structure $\epsilon$. The ``off-diagonal'' part of $D$ does, however, take a very special form; we hope to provide in future work a more geometrical interpretation of this term, which provides for Majorana fermions and for the so-called see-saw mechanism~\cite{CCM07}.

\end{document}